\documentclass[10pt]{article}

\usepackage{fullpage}

\usepackage{graphicx}%
\usepackage{multirow}%
\usepackage{amsmath,amssymb,amsfonts}%
\usepackage{amsthm}%
\usepackage{mathrsfs}%
\usepackage[title]{appendix}%
\usepackage{xcolor}%
\usepackage{textcomp}%
\usepackage{manyfoot}%
\usepackage{booktabs}%
\usepackage{algorithm}%
\usepackage{algorithmicx}%
\usepackage{algpseudocode}%
\usepackage{listings}%
\usepackage{xspace}
\usepackage{thm-restate}
\usepackage{enumitem}

\usepackage[hyphens]{url}
\usepackage{hyperref}
\usepackage[hyphenbreaks]{breakurl}
\hypersetup{
	hidelinks,
	colorlinks=true,
	allcolors=black,
	pdfstartview=Fit,
	breaklinks=true
}

\newtheorem{theorem}{Theorem}
\newtheorem{lemma}[theorem]{Lemma}%
\newtheorem{invariant}[theorem]{Invariant}

\newtheorem{remark}{Remark}%

\newtheorem{definition}{Definition}%
\newtheorem{example}{Example}%

\let\oldunder\underbar
\renewcommand{\underbar}[1]{\emph{\oldunder{#1}}}

\makeatletter
\@ifpackageloaded{stix}{%
}{%
	\DeclareFontEncoding{LS2}{}{\noaccents@}
	\DeclareFontSubstitution{LS2}{stix}{m}{n}
	\DeclareSymbolFont{stix@largesymbols}{LS2}{stixex}{m}{n}
	\SetSymbolFont{stix@largesymbols}{bold}{LS2}{stixex}{b}{n}
	\DeclareMathDelimiter{\lBrace}{\mathopen} {stix@largesymbols}{"E8}%
	{stix@largesymbols}{"0E}
	\DeclareMathDelimiter{\rBrace}{\mathclose}{stix@largesymbols}{"E9}%
	{stix@largesymbols}{"0F}
}
\makeatother

\newcommand{\abs}[1]{| #1 |}

\newcommand{\multiset}[1]{\lBrace #1 \rBrace}

\newcommand{\source}{\mathcal{C}_{S}} 
\newcommand{\target}{\mathcal{C}_{T}} 
\newcommand{\capacity}{\kappa} 
\newcommand{\universe}{U} 
\newcommand{\configuration}{\mathcal{C}} 
\newcommand{\configurationtwo}{\mathcal{D}} 

\newcommand{\bin}{bunch\xspace}
\newcommand{\bins}{bunches\xspace}

\newcommand{\binb}{B} 
\newcommand{\bintypesnoparameters}{\textsc{BunchTypes}} 
\newcommand{\binmultfunction}{\mathit{mult}}
\newcommand{\binmult}[2]{\binmultfunction(#1, #2)} 

\newcommand{\configurationFFD}{{\configuration}_{\mathit{FFD}}} 

\newcommand{\AtLeast}[3]{\mathit{AtLeast}(#1,#2,#3)}

\newcommand{\maxi}{\ensuremath{u_i}\xspace}

\newcommand{\bundleformname}{Bundles}
\newcommand{\bundleformrm}[2]{\ensuremath{\mathrm{\bundleformname}(#1,#2)}\xspace}
\newcommand{\bundleformit}[2]{\ensuremath{\mathit{\bundleformname}(#1,#2)}\xspace}
\newcommand{\bundleformbf}[2]{\ensuremath{\mathbf{\bundleformname}(#1,#2)}\xspace}

\newcommand{\AlgPowOfTwo}{\texttt{SettleItems}\xspace}
\newcommand{\ms}{Merge-Slack\xspace}

\newcommand{\All}{{A}ll}

\newcommand{\largest}{\ell{(\configuration, \target)}}
\newcommand{\largestSource}{\ell{(\source, \target)}}

\newcommand{\brepackingk}{\textsc{$\beta$-Repacking-$\capacity$}}
\newcommand{\partition}{\textbf{P}} 
\newcommand{\binnumber}{\beta} 

\newcommand{\partitionbound}{\beta} 
\newcommand{\partitionpart}{P} 
\newcommand{\boundedpartition}{\partitionbound\text{-bounded partition}\xspace} 
\newcommand{\sequence}[1]{\left( #1 \right)} 

\newcommand{\boundedsubs}[2]{\beta\text{-}\textsc{Subs}(#1, #2)}
\newcommand{\bsubs}{\beta\text{-}\textsc{Subs}}

\newcommand{\subconfig}{\mathcal{I}} 
\newcommand{\subconfigtwo}{\mathcal{J}} 
\newcommand{\subconfigthree}{\mathcal{K}} 
\newcommand{\subconfigfour}{\mathcal{L}} 

\newcommand{\assignment}{\textbf{A}}
\newcommand{\sourceassignment}{\textbf{A}_{\cal S}}
\newcommand{\targetassignment}{\textbf{A}_{\cal T}}
\newcommand{\assignmentpart}{\mathcal{A}}
\newcommand{\assignmenttwo}{\textbf{B}}

\newcommand{\areconfsequence}{\mathbb{A}}
\newcommand{\portion}[2]{\textsc{Portion}(#1 , #2)}

\newcommand{\partitionILP}{\textsc{Partition ILP}}
\newcommand{\demand}{d}
\newcommand{\flow}{f}

\newcommand{\graph}{G} 
\newcommand{\vertexset}{V(G)} 
\newcommand{\edgeset}{E(G)} 
\newcommand{\edgesetgeneric}{E} 
\newcommand{\euv}{uv} 
\newcommand{\inedges}{E^{+}} 
\newcommand{\outedges}{E^{-}} 

\newcommand{\flowsum}{\oplus}
\newcommand{\flowsummation}{\bigoplus}
\newcommand{\flowsub}{\ominus}

\newcommand{\innode}{x} 
\newcommand{\internode}{y} 
\newcommand{\outnode}{z} 
\newcommand{\inputnodes}{X} 
\newcommand{\internodes}{Y} 
\newcommand{\outputnodes}{Z} 

\usepackage{microtype}

\date{}
\begin{document}
	
	\title{Reconfiguration of Multisets \\ with Applications to \textsc{Bin Packing}\thanks{A preliminary version of this paper appeared in the proceedings of the 18th International Conference and Workshops on Algorithms and Computation (WALCOM 2024).\\Research supported by the Natural Sciences and Engineering Research Council of Canada (NSERC)}}

	\author{\small Jeffrey Kam \\ \textit{\footnotesize University of Waterloo}
		\and
		\small Shahin Kamali \\ \textit{\footnotesize York University}
		\and
		\small Avery Miller \\ \textit{\footnotesize University of Manitoba}
		\and 
		\small Naomi Nishimura \\ \textit{\footnotesize University of Waterloo}}
	
		\maketitle
		
	\begin{abstract}
		We use the reconfiguration framework to analyze problems that involve the rearrangement of items among groups. In various applications, a group of items could correspond to the files or jobs assigned to a particular machine, and the goal of rearrangement could be improving efficiency or increasing locality. 
		
		To cover problems arising in a wide range of application areas, we define the general \textsc{Repacking} problem as the rearrangement of multisets of multisets. We present hardness results for the general case and algorithms for various classes of instances that arise in real-life scenarios. 
		By limiting the total size of items in each multiset, our results can be viewed as an offline approach to \textsc{Bin Packing}, in which each bin is represented as a multiset.
		
		In addition to providing the first results on reconfiguration of multisets, our contributions open up several research avenues:  the interplay between reconfiguration and online algorithms and parallel algorithms; the use of the tools of linear programming  in reconfiguration; and, in the longer term, a focus on extra resources in reconfiguration.
		
		\end{abstract}

	\linespread{1.1}\selectfont
	\section{Introduction}
	
	We consider the problem of rearranging items in multisets, from a given source arrangement to a specified target arrangement. Although our techniques draw on the problem of \textsc{Bin Packing} and the area of reconfiguration, each one is considered in a non-traditional way: we view bins and items of the same size as indistinguishable, and we focus on the feasibility of reconfiguration under various conditions. In doing so, we set the stage for the exploration of {\em resource-focused reconfiguration}, in which the goal is to determine which extra resources, if any, are needed to make reconfiguration possible.
	
	The reconfiguration framework~\cite{Ito2011} has been used to consider the step-by-step modification among configurations, which may encode information such as solutions to the instance of a problem or the state of a geometric object, game, or puzzle. Common types of questions framed using reconfiguration include structural properties of the {\em reconfiguration graph}, formed by adding edges between a configuration and any other configuration resulting from the execution of a single {\em reconfiguration step}, as well as the reachability of one configuration from another~\cite{Heuvel2013,Nishimura2018}. In our context, we wish to transform a source configuration into a target configuration by a sequence of reconfiguration steps (or \emph{reconfiguration sequence)}; of particular import is the fact that each intermediate configuration in the sequence conform to the same constraints as the source and target.
	
	Ito et al. studied a related problem~\cite{ItoDErr,ItoD14}, where the objective is to reconfigure one ``feasible packing" of an instance of the \textsc{knapsack} problem to another, where a feasible packing is defined as a subset of items summing to a value in a given range. Under the assumption that the intermediate packings must be feasible, the authors present hardness results for the decision problem and a polynomial-time approximation scheme (PTAS) for the optimization problem. In another related problem~\cite{Cardinal2020}, Cardinal et al. showed the PSPACE-hardness of reconfiguring between two subsets of items with the same sum, where each step is defined by adding or removing at most three items while maintaining the same sum across intermediate subsets. In contrast, in our work, each configuration consists of multiple bins, not a single bin, and every item must be packed.
	
	We define a problem, \textsc{Repacking}, where, as in \textsc{Bin Packing}, the goal is to maintain a grouping of items that forms a packing. Packings are naturally modeled as multisets, since neither the ordering of multisets nor the ordering of items within a multiset are important. At a high level, we are considering the reconfiguration of items in unlabeled, and hence indistinguishable, bins. To avoid confusion with traditional bin packing (where the bins are distinguishable containers), we refer to each group of items as a \emph{\bin} instead of a bin. For full generality, there are no constraints on the numbers or types of items, composition of \bins, or allowable packings. 
	
	Due to the hardness of the most general form of the problem (Section~\ref{sec-hardness}), we turn to situations that might naturally arise in real-life uses of repacking. In many settings, such as virtual machine placement, the maximum sum of sizes of items in a multiset can be set to a fixed constant.
	To mimic the capacity of a container, we define a \textit{capacity} as the upper bound on the sum of sizes of items in a \bin. 
	From a practical standpoint, it may be reasonable to consider scenarios in which the sizes of items and capacities are constant with respect to the number of multisets under consideration. 
	For example, a typical data center contains thousands of homogeneous Physical Machines (PMs), each having a bounded capacity, in terms of bandwidth, CPU, memory, and so on, and hosting a relatively small number of Virtual Machines (VMs). Virtual Machines (``items") have various loads, and the total load of VMs assigned to a PM (a ``$\textrm{\bin}$") must not exceed the uniform capacity of the PMs. 
	It is desirable to reconfigure a given assignment of VMs to PMs for reasons such as performing maintenance, balancing loads, collocating VMs that communicate with each other, and moving VMs apart to minimize the fault domain~\cite{VMWareVM,MedinaG14}. Such reconfiguration must occur in steps that involve migrating a VM from one PM to another while respecting the bounded capacity of the PMs at each step.
	
	Our paper is structured as follows. In Section~\ref{sec-prelimes}, we formally define the \textsc{Repacking} problem. We prove that \textsc{Repacking} is strongly NP-hard in Section~\ref{sec-hardness}. Then, we provide algorithms for instances relevant to real-life scenarios. In Section~\ref{sect:specialSmallItems}, we consider instances in which all item sizes are bounded above by a constant fraction of the capacity, and give an algorithm that reconfigures any instance in which the unoccupied space is sufficiently large. In Section~\ref{Sect:SpecialPowers}, we consider a setting in which all item sizes and the capacity are powers of 2, fully characterize when reconfiguration is possible in such instances, and give a reconfiguration algorithm for all reconfigurable instances. Motivated by the possibility of solving reconfiguration problems in parallel, Section~\ref{sect:specialStructure} gives an algorithm that determines, for any given instance, whether or not reconfiguration is possible by partitioning the source configuration into smaller parts and only moving an item within its assigned part. In Section~\ref{sec-future}, we present directions for future work.


	\section{Preliminaries}\label{sec-prelimes}
	
	Our goal is to determine whether it is possible rearrange items from one configuration (multiset of multisets of items) to another.  Stated in the most general terms, items can have a variety of attributes, such as distinct IDs, and multisets can have not only various attributes but also different capacities. In the remainder of the paper, we simplify the discussion and notation by assuming that {\em items} are positive integers representing sizes (where items of the same size are indistinguishable) and {\em \bins} are multisets of items (where \bins corresponding to the same multiset are indistinguishable). To facilitate explanation of algorithms or analysis, at times we might assign names to specific items or \bins for greater clarity. 
	
	To avoid confusion among items, \bins (multisets of items), and multisets of \bins (multisets of multisets of items), in our choice of notation, we use the convention that a lower-case letter refers to an item or a number (where at times Greek letters will also be used for numbers), an upper-case letter refers to a set or multiset, and a calligraphic letter refers to a multiset of multisets. 
	
	To give an example of rearrangement, suppose we wish to transform the two \bins, $\left\lBrace 1, 1, 2, 6\right\rBrace$ and $\left\lBrace 2, 3, 5\right\rBrace$ (the first \bin and second \bin, respectively) into the two \bins $\left\lBrace 1, 3, 6\right\rBrace$ and $\left\lBrace 1, 2, 2, 5\right\rBrace$. In the absence of constraints, we can achieve our goal by moving items 1 and 2 from the first \bin to the second \bin and item 3 from the second \bin to the first \bin. The task proves more challenging when there is an upper bound on the sum of items in any \bin, and rendered impossible without extra storage space if the upper bound is 10.
	
	Although in general the sum of all items in a \bin may be unconstrained,  in this paper we view each \bin as having a fixed positive integer {\em capacity}, typically denoted by $\capacity$, as an upper bound. We use $\mathit{vol}(B)$, the {\em volume of $B$}, to denote the sum of all the items in \bin $B$, and require that $\mathit{vol}(B) \le \capacity$. The term {\em slack}, denoted $\mathit{slack}(B)$, is used to refer to the amount of space in $B$ that is not occupied by items, where $\mathit{slack}(B) = \capacity - \mathit{vol}(B)$. By definition, the slack of a \bin is non-negative. A \bin is {\em empty} if it contains no items; in such a case, the slack will equal the capacity. 
	
	Formalizing our terminology, we define a {\em configuration} to be a multiset of \bins, and a {\em legal configuration for capacity $\capacity$} to be a configuration in which the volume of each \bin is at most $\capacity$. The {\em underlying set of a configuration} $\configuration$,  denoted $U({\configuration})$, is the multiset union of items in all \bins in a configuration; $U$ is used without a parameter to denote the underlying set of the source and target configurations. 
	
	We transform one configuration to another by a series of steps, each corresponding to the move of a single item.
	More formally, we move a single item $u$ from one \bin (the {\em donor \bin}, $B_d$) to another \bin (the {\em recipient \bin}, $B_r$), where we reuse the names $B_d$ and $B_r$ to describe two different \bins, and two different points in time (before and after the move). Between the two points of time, $B_d$ is changed only by the removal of $u$, and $B_r$ is changed only by the addition of $u$. The pairs of \bins ($B_d$ and $B_r$ at the first point in time and $B_d$ and $B_r$ at the second point in time) are said to correspond to the {\em move of a single item}.
	
	We consider legal configurations ${\configuration}$ and ${\configurationtwo}$ for capacity $\capacity$ to be {\em adjacent} if it is possible to form ${\configurationtwo}$ from ${\configuration}$ by the move of a single item. More formally,
	${\configuration} \setminus {\configurationtwo}$ consists of two {\em old \bins},
	${\configurationtwo} \setminus {\configuration}$ consists of two {\em new \bins}, and the pair of old \bins and the pair of new \bins correspond to the move of a single item. We say that we can get from ${\configuration}$ to ${\configurationtwo}$ in a single {\em reconfiguration step} if the two configurations are adjacent.
	A {\em reconfiguration sequence for capacity $\capacity$} from a source configuration $\source$ to $\target$ consists of a sequence of legal configurations for capacity $\capacity$,  $\source = {\configuration}_0, {\configuration}_1, \ldots, \target$ such that each pair of consecutive configurations ${\configuration}_i$ and ${\configuration}_{i+1}$ are adjacent.
	
	We define the following problem \textsc{Repacking}, with instance $\left(\mathcal{C}_S, \mathcal{C}_T, \capacity\right)$, as defined below:
	
	\begin{description}
		\item[Input:] Source and target legal configurations $\source$ and $\target$ for capacity $\capacity$, with $|\source|=|\target|$.
		\item[Question:] Is there a reconfiguration sequence for capacity $\capacity$ from $\source$ to $\target$?
	\end{description}
	
	Since, as we will show in Theorem~\ref{lemma-NPcomplete}, \textsc{Repacking} is strongly NP-hard in general, it is natural to consider extensions and variants of the problem.   As defined, the number of \bins is preserved at each step in the sequence, and hence the number of \bins available for reconfiguration is the number of \bins in the source (and hence the target) configuration, denoted $|\source|$. More generally, we consider the interplay among various properties of an instance and the relationship between the sum of the items and the total capacity of all \bins. 
	
	
	\section{Complexity of \textsc{Repacking}}\label{sec-hardness}
	In this section, we prove that \textsc{Repacking} is strongly NP-hard via a reduction from \textsc{Bin Packing}, which is known to be strongly NP-hard~\cite{GareyJ79}. The \textsc{Bin Packing} problem is defined as follows: given a multiset of $n$ positive integers $\multiset{z_1,...,z_n}$, a positive integer $m$, and a positive integer $\alpha$, determine whether or not $\multiset{z_1,...,z_n}$ can be partitioned into at most $m$ multisets such that the numbers in each multiset sum to at most $\alpha$. At a high level, the reduction creates an instance of \textsc{Repacking} where the source and target configurations each consist of $n+2$ \bins, and only differ in two \bins in such a way that the reconfiguration requires a swap of two large items between these two \bins. But to perform the swap, many small items must be temporarily moved out from one of the two \bins into the remaining $n$ \bins not involved in the swap, and, to create enough temporary storage space for these small items, the items with sizes $z_1,\ldots,z_n$ must be rearranged so that their arrangement is ``compressed" from the $n$ \bins down to at most $m$ \bins that each only have $\alpha$ free space. Thus, the swap of the two large items can be carried out if and only if the integers $z_1,\ldots,z_n$ can be partitioned into at most $m$ parts with the sum of integers in each part being at most $\alpha$. 

    Our reduction consists of two main steps. First, we define an intermediate problem (a restricted version of \textsc{Bin Packing}) and prove that it is also strongly NP-hard. In particular, we define \textsc{Restricted Bin Packing} as follows:
    \begin{description}
		\item[Input:] a multiset of $n$ positive integers $\multiset{z_1,...,z_n}$, a positive integer $m$, and a positive integer $\alpha$ such that $\alpha \geq \max\{z_1,\ldots,z_n\}$, $\alpha \geq 2$, and $n \geq 2m+2$.
		\item[Question:] Can $\multiset{z_1,...,z_n}$ be partitioned into at most $m$ multisets such that the numbers in each multiset sum to at most $\alpha$?
	\end{description}
    
    \begin{restatable}{lemma}{restrictedBP}\label{lem:restrictedBP}
    \textsc{Restricted Bin Packing} is strongly NP-hard.
    \end{restatable}
    The proof of Lemma \ref{lem:restrictedBP} is provided in Appendix \ref{app:hardness}. Next, we prove that \textsc{Repacking} is strongly NP-hard via a reduction from \textsc{Restricted Bin Packing}.

	\begin{theorem}\label{lemma-NPcomplete} 
		\textsc{Repacking} is strongly NP-hard.
	\end{theorem}
	
	\begin{proof}
		We give a polynomial-time reduction from \textsc{Restricted Bin Packing} to \textsc{Repacking}, which is sufficient by Lemma \ref{lem:restrictedBP}. Consider any instance $P=(z_1,...,z_n,m,\alpha)$ of \textsc{Restricted Bin Packing}. Define $\alpha_1,\ldots,\alpha_{n-m}$ as $n-m$ items of size $\alpha$. The output of our reduction is the instance $(\source,\target,\capacity)$ of $\textsc{Repacking}$ with
		\begin{align*}
			\source &= \multiset{\multiset{z_1,\capacity-\alpha},\ldots,\multiset{z_n,\capacity-\alpha},\multiset{\capacity/2,\capacity/2},\multiset{\capacity/2 - 1,\alpha_1,\ldots,\alpha_{n-m}}}\\
			\target &= \multiset{\multiset{z_1,\capacity-\alpha},\ldots,\multiset{z_n,\capacity-\alpha},\multiset{\capacity/2-1,\capacity/2},\multiset{\capacity/2,\alpha_1,\ldots,\alpha_{n-m}}}\\
			\capacity &= 2(n-m)\alpha
		\end{align*}

		Note that the source and target configurations each consist of $n+2$ \bins and are nearly identical: the first $n$ \bins of $\source$ are the same as the first $n$ \bins of $\target$, and the last two \bins of $\target$  are obtained from the last two \bins of $\source$ by swapping an item of size $\capacity/2$ with the item of size $\capacity/2-1$. In what follows, we refer to the first $n$ \bins listed above in the source and target configurations as the {\em matching \bins}, and we refer to the last 2 \bins as the {\em nonmatching \bins}. Also note that, by our choice of $\capacity$, the sum of the sizes of items $\alpha_1,\ldots,\alpha_{n-m}$ is exactly $\capacity/2$. An illustration of the source and target configurations is provided in Figure \ref{fig:Reduction}.

  \begin{figure}[h!]
		\centering
		\includegraphics[scale=0.6]{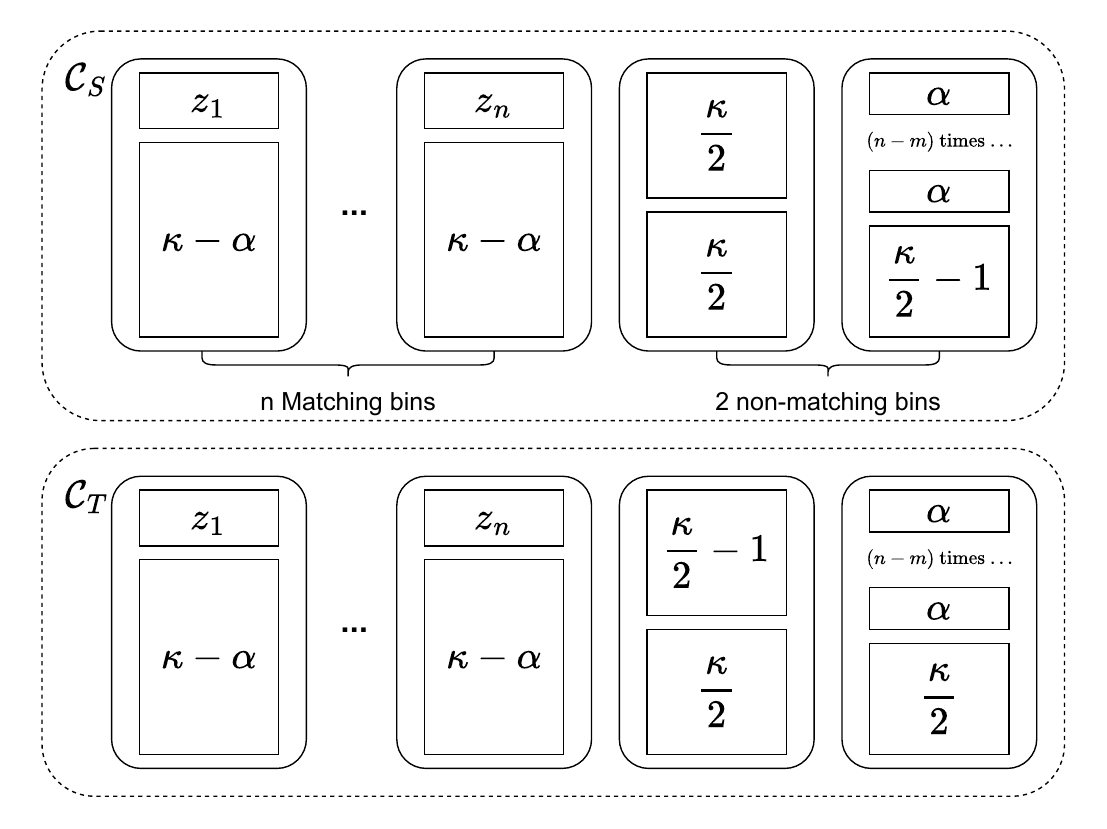}
		\caption{An illustration of our reduction's output instance of \textsc{Repacking}.}
		\label{fig:Reduction}
	\end{figure}
		
		We now proceed to prove that the reduction is correct. First, suppose that $P=(z_1,...,z_n,m,\alpha)$ is a yes-instance of \textsc{Restricted Bin Packing}. Recall that this means that the integers $z_1,\ldots,z_n$ can be partitioned into at most $m$ parts of size at most $\alpha$. It follows that, starting from the source configuration $\source$, we are able to rearrange the items with sizes $z_1,\ldots,z_n$ in the $n$ matching \bins so that these $n$ items are located in at most $m$ of the first $n$ \bins (along with the item of size $\capacity - \alpha$ in each \bin). After doing so, there are at least $n-m$ remaining matching \bins, each consisting of one item of size $\capacity-\alpha$, and thus each has slack equal to $\alpha$. Then, we move each of the $n-m$ items of size $\alpha$ (i.e., $\alpha_1,\ldots,\alpha_{n-m}$) from the last nonmatching \bin into one of the above-mentioned matching \bins that has slack equal to $\alpha$. Since the sum of the sizes of items $\alpha_1,\ldots,\alpha_{n-m}$ is exactly $\capacity/2$, it follows that the last nonmatching \bin now has slack equal to $\capacity/2 + 1$, and thus we exchange an item of size $\capacity/2$ with the item of size $\capacity/2-1$ between the two nonmatching \bins. Finally, we reverse the steps described above to return all other items to their original locations, and we have reached the target configuration $\target$, which proves that $(\source, \target, \capacity)$ is a yes-instance of \textsc{Repacking}, as required.
		
		Conversely, suppose that $(\source, \target, \capacity)$ is a yes-instance of \textsc{Repacking}. In $\source$, the first nonmatching \bin contains the only two items of size $\capacity/2$, and in $\target$ there is a \bin $\multiset{\capacity/2-1,\capacity/2}$, so it follows that an item of size $\capacity/2$ must be moved out of the first nonmatching \bin of $\source$ at some step of the reconfiguration sequence. The remainder of the proof argues that there is only one way this can happen.
		
		First, we set out to prove that no item of size $\capacity/2$ or $\capacity/2 - 1$ can ever be moved into any of the $n$ matching \bins. In the source configuration $\source$, observe that each of the $n$ matching \bins contains one item of size $\capacity-\alpha$ and one item whose size is a positive integer, so each of the $n$ matching \bins has slack at most $\alpha-1$, which implies that the total slack across these \bins is at most $n(\alpha-1)$. Further, the first nonmatching \bin has slack equal to 0, and the second nonmatching \bin has slack equal to 1. Thus, in $\source$, the total slack across all $n+2$ \bins is at most $n(\alpha-1)+1$. Since no items or \bins can be added or removed during reconfiguration, it follows that the total slack across all \bins is at most $n(\alpha-1)+1$ in every configuration of every reconfiguration sequence.
		Now, using the fact that we set $\capacity=2(n-m)\alpha$ and using the input restriction $n \geq 2m+2$ from the definition of \textsc{Restricted Bin Packing}, we see that $\capacity - \alpha = 2(n-m)\alpha - \alpha \geq 2(n-\frac{n}{2}+1)\alpha - \alpha = (n+1)\alpha > n\alpha \geq n(\alpha-1)+1$.
		This means that each item of size $\capacity - \alpha$ is strictly larger than the total slack in every configuration, so none of the items of size $\capacity - \alpha$ can ever be moved. Therefore, each of the $n$ matching \bins always contains an item of size $\capacity - \alpha$, i.e., the amount of slack in each such \bin is always at most $\alpha$. Next, since $\capacity = 2(n-m)\alpha$ and $n \geq 2m+2 \geq m+2$, we see that $\alpha \leq \capacity/4$, so it follows that the amount of slack in each of the $n$ matching \bins will never be more than $\capacity/4$. Finally, using the input restrictions $n \geq 2m+2$ and $\alpha \geq 2$ from the definition of \textsc{Restricted Bin Packing}, we see that $\capacity = 2(n-m)\alpha > 16$, so $\capacity/4 < \capacity/2-1$. Thus, we have shown that the amount of slack in each of the $n$ matching \bins is always strictly less than $\capacity/2-1$, so no item of size $\capacity/2$ or $\capacity/2 - 1$ can ever be moved to one of the $n$ matching \bins, as desired.
		
		From the fact proven in the previous paragraph, we know that to move one of the items of size $\capacity/2$, it must be moved directly from the first nonmatching \bin to the second nonmatching \bin in a single move. For such a move to occur, the reconfiguration sequence must first ensure that the amount of slack in the second nonmatching \bin is at least $\capacity/2$. To create this amount of slack in the second nonmatching \bin, note that we cannot move out the item of size $\capacity/2-1$ since we proved above that it will never fit into any of the $n$ matching \bins. So the only option is to move out items of size $\alpha$. If we do not move out all $n-m$ items of size $\alpha$, then the \bin will contain an item of size $\capacity/2 - 1$ and at least one item of size $\alpha \geq 2$, i.e., the volume would be at least $\capacity/2 + 1$, which means that the amount of slack would be at most $\capacity/2 - 1$, which is insufficient. Thus, we conclude that all $n-m$ items of size $\alpha$ must be moved out of the second nonmatching \bin. Since the amount of slack in the first nonmatching \bin is equal to 0, these $n-m$ items of size $\alpha$ must be moved to the matching \bins. To be able to do this, we need to reconfigure the $n$ matching \bins so that the amount of slack in at least $n-m$ of them is exactly $\alpha$ (this is because each matching \bin always contains an item of size $\capacity-\alpha$). Thus, the items with sizes $z_1,\ldots,z_n$ must all be moved into at most $m$ of the $n$ matching \bins, each also containing an item of size $\capacity - \alpha$. But this means that the items with sizes $z_1,\ldots,z_n$ have been partitioned into at most $m$ parts with the item sizes in each part summing to at most $\alpha$, and so $(z_1,\ldots,z_n,m,\alpha)$ is a yes-instance of \textsc{Restricted Bin Packing}, as required.
	\end{proof}
	
	
	\section{Special Case: Instances with Small Items}\label{sect:specialSmallItems}
	Intuitively, if items are relatively small and there is enough slack in the source and target configurations, it seems feasible to reconfigure them. In this section, we formalize and prove this intuition. 
	In particular, we consider a setting where all items are of size at most $\capacity/\alpha$, for some integer $\alpha >1$. 
	We define the \emph{average slack} of \bins in a configuration $\configuration$ by $\frac{1}{|\configuration|} \sum_{\binb \in \configuration}\mathit{slack}(B)$, and we assume that the average slack of $\source$ is at least $\frac{\capacity}{\alpha+1} + \frac{3\alpha\capacity}{(\alpha+1)|\source|}$. Since $|\source|=|\target|$ and the total volume of all \bins is the same in $\source$ and $\target$, this assumption implies that the average slack of \bins of $\target$ is at least $\frac{\capacity}{\alpha+1} + \frac{3\alpha\capacity}{(\alpha+1)|\target|}$. 
	For larger inputs, this means that the average slack of \bins converges to a value at least $\frac{\capacity}{\alpha+1}$. 
	We will describe an algorithm that, under these assumptions, reconfigures  $\source$ to  $\target$ (as formally stated in Theorem~\ref{thm-small-items}).
	
	Consider a configuration resulting from sorting items in the underlying set $U$ in non-increasing order of their sizes and placing them, one by one, in the first \bin with enough slack, and placing them in a new \bin if no such \bin exists.  
	This configuration, which we call the \emph{First-Fit-Decreasing (FFD) configuration}, and denote by $\configurationFFD$, is a canonical configuration for our reconfiguration algorithm. That is, we show how to reconfigure both  $\source$ and  $\target$ to  $\configurationFFD$ and therefore to each other. 
	
	In what follows, we describe how a configuration ${\configuration}$ with \bins of average slack at least $\frac{\capacity}{\alpha+1} + \frac{3\alpha\capacity}{(\alpha+1)|\configuration|}$ can be reconfigured to $\configurationFFD$. This process is applied to reconfigure $\source$ and $\target$ to $\configurationFFD$.  The algorithm works in stages: each stage starts with a ``compression phase" and then an ``FFD-retrieval phase" follows. As we will show, the compression phase always results in at least one empty \bin, which is subsequently packed according to the FFD rule. Next, we formally describe this process. 
	
	At the beginning of stage $i$, we have two types of \bins: FFD \bins and general \bins. At the beginning of the reconfiguration process, all \bins are general \bins. 
	As mentioned earlier, items in FFD \bins are packed according to the FFD rule. 
	
	In the compression phase of stage $i$, we process general \bins one by one in an arbitrary order. Items within the \bin being processed are also processed one by one in an arbitrary order. For each item $x$, we check the \bins that were previously processed in this stage and move $x$ to the first \bin (with respect to the ordering) with slack at least equal to $x$; if no such \bin exists, $x$ remains in its \bin. 
	
	Next, in the FFD-retrieval phase of stage $i$, we first declare any empty general \bin to be an FFD \bin. In Lemma~\ref{lem:smallitems}, we will prove that at least one such \bin exists at the end of the compression phase. We pack these empty FFD \bins according to the FFD rule: we process items in the general \bins in non-increasing order of their sizes and place each item into the first FFD \bin with enough slack to host it. This process continues until no more items are in the general \bins or no FFD \bin has enough slack for the largest remaining item in the general \bins. In the former case, the canonical configuration $\configurationFFD$ is reached and the process ends.
	In the latter case, stage $i$ ends and the process continues with the compression phase of the new stage $i+1$.
	
	We now set out to prove the correctness of the algorithm. The following fact holds at the end of the compression phase of any stage.

	\begin{lemma}\label{lem:comp}
		At the end of the compression phase, at most two non-empty general \bins can have slack of larger than $\capacity/(\alpha+1)$. 
	\end{lemma}
	
	\begin{proof}
		Let $\binb$ denote the first processed \bin that is non-empty and has slack at least $\capacity /(\alpha+1)$. Any item that remains in a \bin processed after $\binb$ must have size larger than $\capacity/(\alpha+1)$, and hence in the range $(\capacity/(\alpha+1),\capacity/\alpha]$. Since exactly $\alpha$ items in this range fit into a \bin, any non-empty \bin processed after $\binb$, except possibly the last \bin, includes exactly $\alpha$ items, so will have volume larger than $\capacity \alpha/(\alpha+1)$, and hence has slack at most $\capacity - (\capacity\alpha/(\alpha+1)) = \capacity/(\alpha+1)$.
	\end{proof}
	
	Next, we show that the compression process always results in at least one empty \bin, as needed to carry out the FFD-retrieval phase.
	
	\begin{lemma}\label{lem:smallitems}
		If all items in the underlying set $U$ of the configuration $\configuration$ are of size at most $\capacity/\alpha$ for some integer $\alpha>1$ and the average slack of all \bins in $\configuration$ is at least $\frac{\capacity}{\alpha+1} + \frac{3\alpha\capacity}{(\alpha+1)|\configuration|}$, then at the end of the compression phase of any stage $i\geq 1$, there is at least one empty general \bin. 
	\end{lemma}
	
	\begin{proof}
		Let $\beta_i$ denote the number of general \bins, and $s_i$ denote the total size of items in general \bins at the beginning of stage $i$ of the reconfiguration process. We will prove that the average slack of general \bins at the beginning of each stage $i \geq 1$ of the reconfiguration process is at least $\frac{\capacity}{\alpha+1} + \frac{3\alpha\capacity}{(\alpha+1)\beta_i}$. That is, $\capacity - \frac{s_i}{\beta_i} \geq \frac{\capacity}{\alpha+1} + \frac{3\alpha \capacity}{(\alpha+1)\beta_i}$, or, equivalently,
		\begin{align}\label{ineq:small}
			\beta_i \geq \frac{(\alpha+1)s_i}{\capacity\alpha } + 3    
		\end{align}
		Since $\beta_i$ is an integer, it follows that $\beta_i \geq \left\lceil\frac{(\alpha+1)s_i}{\capacity\alpha }\right\rceil + 3$. Note that our desired result follows directly from Inequality~\ref{ineq:small} and Lemma~\ref{lem:comp}: by Lemma~\ref{lem:comp}, at the end of the compression phase of stage $i$, all non-empty general \bins except possibly two have a volume of at least $\alpha\capacity/(\alpha+1)$. 
		So, a total volume of $s_i$ is distributed between non-empty general \bins that, excluding at most two of them, are each of volume at least $\alpha \capacity/(\alpha+1)$. Therefore, the number of non-empty \bins will be at most $\left\lceil \frac{s_i}{\alpha \capacity/(\alpha+1)} \right\rceil+2 = \left\lceil\frac{(\alpha+1)s_i}{\capacity\alpha}\right\rceil + 2$.
		In other words, at least $\beta_i-\left(\left\lceil\frac{(\alpha+1)s_i}{\capacity\alpha}\right\rceil +2\right) \geq \left(\left\lceil\frac{(\alpha+1)s_i}{\capacity\alpha }\right\rceil + 3    \right) - \left(\left\lceil\frac{(\alpha+1)s_i}{\capacity\alpha}\right\rceil +2\right) \geq 1$ general \bins will be empty by the end of the compression phase of stage $i$.
		
		It remains to prove Inequality~\ref{ineq:small}, for which we will use induction.
		Initially, for $i=1$, the inequality holds due to the lemma statement's assumption about the average slack: note that $\beta_1 = |{\configuration}|$ since all \bins are initially general \bins. 
		
		Assuming that Inequality~\ref{ineq:small} holds for some $i\geq 1$, we prove that it holds for $i+1$. 
		For that, we consider the empty general \bins after the compression phase of stage $i$. Let's call these \bins\ {\em critical \bins} and let $c_i$ denote the number of these \bins. Given that Inequality~\ref{ineq:small} holds for $i$, we have $c_i \geq 1$, as discussed above. Recall that the critical \bins are declared FFD \bins and are packed accordingly during the FFD-retrieval phase of stage $i$.
		Let $x$ denote the size of the largest item in a general \bin at the end of the FFD-retrieval phase of stage $i$. We consider the following two possibilities:
		
		\begin{itemize}
			\item Suppose that $x > \capacity/(\alpha+1)$. Given that we  move items in non-increasing order of their sizes, all items moved to critical \bins are of size at least equal to $x$ and thus larger than $\capacity/(\alpha+1)$. Since all items are of size at most $\capacity/\alpha$, each critical \bin has received exactly $\alpha$ items in the range $(\capacity/(\alpha+1), \capacity/\alpha]$.
			Therefore, its volume is larger than $\capacity \alpha/(\alpha+1)$.
			\item Otherwise, suppose that $x \leq \capacity/(\alpha+1)$. Since $x$ was not packed in any critical \bin, the slack of each critical \bin is less than $x$. But this means that the volume of each \bin is larger than $\capacity - x \geq \capacity - \capacity/(\alpha+1) = \capacity\alpha/(\alpha+1)$.  
		\end{itemize}
		In both cases, all critical \bins have a volume larger $\capacity \alpha/(\alpha+1)$ at the end of stage $i$. 
		Thus, the total size of items moved from general \bins to FFD \bins at stage $i$ is larger than $\capacity \alpha c_i/(\alpha+1)$. Consequently, at the end of stage $i$, the total size of items in general \bins is $s_{i+1} \leq s_i - \capacity \alpha c_i/(\alpha+1)$, while the number of general \bins is $\beta_{i+1} = \beta_i - c_i$. To prove that Inequality~\ref{ineq:small} holds, we see that:
		\begin{align*}
			\beta_{i+1} & = \beta_i - c_i\\
			& \geq \left( \frac{(\alpha+1)s_i}{\capacity \alpha} + 3 \right) - c_i \tag{by the induction hypothesis}\\
			& = \frac{(\alpha+1)(s_i-\capacity\alpha c_i/(\alpha+1))}{\capacity\alpha}+3\\
			& \geq \frac{(\alpha+1)s_{i+1}}{\capacity\alpha}+3
		\end{align*}%
  \end{proof}
	
	Provided with Lemma~\ref{lem:smallitems}, we are now ready to prove that the algorithm correctly produces a reconfiguration sequence from $\source$ to $\target$.

	\begin{theorem}\label{thm-small-items}
		Suppose the underlying set $U$ of $\source$ and $\target$ is formed by items of size at most $\capacity/\alpha$ for some integer $\alpha>1$ and the average slack of \bins in $\source$ and $\target$ is at least $\frac{\capacity}{\alpha+1} + \frac{3\alpha\capacity}{(\alpha+1)|\source|}$. 
		Then, it is always possible to reconfigure both $\source$ and $\target$ to the FFD configuration and, hence, to each other.
	\end{theorem}
	
	\begin{proof}
		We use the reconfiguration process described in this section starting with configuration $\source$. By applying Lemma~\ref{lem:smallitems} with $\configuration=\source$, we are guaranteed the presence of at least one empty \bin at the end of the compression phase of each stage, so the number of FFD \bins strictly increases during each stage. Eventually, at the end of some stage, all \bins become FFD \bins, and we obtain the canonical configuration $\configurationFFD$. This gives a reconfiguration sequence $\mathbb{S}_1$ from $\source$ to $\configurationFFD$. Similarly, using the reconfiguration process described in this section starting with configuration $\target$, we obtain a reconfiguration sequence from $\target$ to $\configurationFFD$, and we reverse the moves made in this sequence to obtain a reconfiguration sequence $\mathbb{S}_2$ from $\configurationFFD$ to $\target$. Concatenating the sequence $\mathbb{S}_2$ to the end of sequence $\mathbb{S}_1$ gives a reconfiguration sequence from $\source$ to~$\target$. 
	\end{proof}
	
	
	\section{Special Case: Items and Capacities are Powers of 2}\label{Sect:SpecialPowers}
	Consider a setting where the capacity $\capacity$ and all item sizes are powers of 2. We characterize the input instances for which reconfiguration is possible.
	In particular, we provide an algorithm that reconfigures $\source$ to $\target$, as long as the total slack across all \bins is at least equal to the size of the largest item that must be moved. Since this is the minimal requirement for a feasible reconfiguration, in essence our results show that powers of 2 provide an ideal setting for reconfiguration. 
	
Our reconfiguration algorithm, called \AlgPowOfTwo, processes items in stages by non-increasing order of size, such that once items in the current stage have been ``settled'', they remain in place during all subsequent stages. Each stage entails repeatedly identifying two \bins, $B_s$ and $B_d$, with a surplus and a deficit, respectively, of items of the current size with respect to the target configuration, and then moving an item from $B_s$ to $B_d$. Each move takes place in two phases: in the {\em compression phase}, slack equal to the size of the item is accumulated in a \bin $B_{temp}$, and in the {\em transfer phase}, $B_{temp}$ is used as temporary storage in the move of an item from $B_s$ to $B_d$. The remainder of this section formalizes what it means for an item to be ``settled''.
	
	\noindent\textbf{Settling.} 
	Our results depend on the ability to look at an item size $a$ in a configuration $\configuration$ and determine: is there a reconfiguration sequence that reconfigures $\configuration$ to $\target$ such that no step involves moving an item with size at least $a$? If we ask this question about the item sizes in $\configuration$ from largest to smallest, the first size for which the answer is `no' is called \emph{the size of the largest item in $\configuration$ that must be moved} and denoted by $\largest$. For example, if $\configuration = \multiset{\multiset{{32, 16}}, \multiset{{4,4,2}}}$ and $\target = \multiset{\multiset{{32,4,4,2}}, \multiset{{16}}}$, then $\largest = 16$. 
	
	To enable us to determine $\largest$, we introduce the idea of an item size being ``settled". Informally, an item of size $s$ is settled in its current \bin $\binb$ if there is a \bin in the target configuration that contains the same multiset of items of size at least $s$ as $B$ does, i.e., none of these items need to be moved from their current positions. The largest item in $\binb$ whose size is not settled would be ``out of place" if all larger items in $B$ stay where they are for the remainder of the reconfiguration sequence, i.e., its size is a candidate for the size of the largest item in the configuration that must be moved. Formally, for any configuration $\configurationtwo$, any \bin $\binb \in \configurationtwo$, and any item size $s$, we define $\AtLeast{s}{\binb}{\configurationtwo}$ to be the multiset of items in $\binb$ that have size at least $s$. For any $s$, we say that item size $s$ is \emph{settled in $\configurationtwo$} if there is a bijection $\varphi$ from the multiset of \bins containing an item of size at least $s$ in $\configurationtwo$ to the multiset of \bins containing an item of size at least $s$ in $\target$, and $\AtLeast{s}{B}{\configurationtwo} = \AtLeast{s}{\varphi(B)}{\target}$ for each \bin $\binb$ in the domain of $\varphi$. Such a bijection $\varphi$ is called \emph{a settling bijection for size $s$}.
	
	\begin{definition}
		For any configuration $\configuration$, define $\largest$ to be largest item size that is not settled in $\configuration$. This is also referred to as \emph{the size of the largest item in $\configuration$ that must be moved}.
	\end{definition}
	
	\noindent\textbf{Computing $\largest$.} 
	For any configuration $\configuration$, we can compute $\largest$ by checking the item sizes in $\configuration$, in decreasing order, until we find one that is not settled. So it remains to describe, for any item size $s$, how to check whether or not item size $s$ is settled in $\configuration$. To this end, we provide the following procedure \textit{Check-Settled$(s,\configuration)$}, which will also produce a settling bijection $\varphi$ for size $s$ if item size $s$ is settled in $\configuration$.
	
	Start with a mapping $\varphi$ with an empty domain and range, and in what follows, we write $\configuration\setminus\varphi$ to denote the \bins in $\configuration$ that are not in the domain of $\varphi$, and we write $\target\setminus\varphi$ to denote the \bins in $\target$ that are not in the range of $\varphi$. In increasing order, perform the following for each item size $u$ in $\configuration$ such that $u \geq s$: find a bijection $\varphi_u$ from the multiset of \bins $B$ in $\configuration\setminus\varphi$ that contain an item of size $u$ to the multiset of \bins in $\target\setminus\varphi$ that contain an item of size $u$ in such a way that $\AtLeast{u}{B}{\configuration} = \AtLeast{u}{\varphi_u(B)}{\target}$, and if such a $\varphi_u$ exists, then extend $\varphi$ by adding $\varphi_u$ to it. The procedure returns ``settled" if the above process was successful for all $u \geq s$ (i.e., it was possible to find the described bijection $\varphi_u$ for each item size $u \geq s$ in $\configuration$) and returns ``not settled" otherwise. The following result establishes that the \textit{Check-Settled} procedure correctly determines whether or not $s$ is settled in $\configuration$, and, moreover, if $s$ is settled in $\configuration$, then the $\varphi$ produced by the procedure is a settling bijection for size $s$. The proof is provided in Appendix \ref{App:CheckSettled}.
	\begin{restatable}{lemma}{CheckSettledCorrect}
		The procedure \textit{Check-Settled$(s,\configuration)$} correctly determines whether or not item size $s$ is settled in $\configuration$, and, produces a settling bijection for size $s$ if $s$ is settled in $\configuration$.
	\end{restatable}

	We are now ready to describe our reconfiguration algorithm, which works as long as the total slack across all \bins is at least equal to the size of the largest item that must be moved, i.e., at least $\largestSource$.
 
	\subsection{Description of \AlgPowOfTwo}\label{sec:pow2alg}
	Our reconfiguration algorithm works in stages. For an arbitrary $i \geq 1$, let $\configuration_i$ denote the configuration at the start of the $i$th stage, and denote by $\maxi$ the $i$th largest item size in the underlying set $U$. During the $i$th stage, the algorithm proceeds by moving items of size at most $\maxi$ between \bins to ensure that, at the end of the stage, item size $\maxi$ is settled. After a finite number of stages, all item sizes are settled, which means that the current configuration is $\target$. 
	
	So it remains to describe how the algorithm settles all items of size $\maxi$ during the $i$th stage. There are two main steps to accomplish this. In the first step, called the \emph{Mapping Step}, the algorithm considers the configuration $\configuration_i$ at the start of the $i$th stage, and specifies a mapping from \bins in $\configuration_i$ to \bins in $\target$, and, in a sense, this mapping serves as a ``reconfiguration plan" (i.e., it specifies which \bins in $\configuration_i$ will be reconfigured to which \bins in $\target$). In the second step, called the \emph{Action Step}, the algorithm has to carry out this ``plan". 
 
 More specifically, at the start of the $i$th stage, the algorithm chooses a bijection $\varphi_i$ between a subset of the \bins of $\configuration_i$ and those of $\target$. Initially, before the start of the first stage, the domain of this bijection is empty. During the $i$th stage, the bijection $\varphi_i$ is chosen so that its domain and range include all \bins whose largest item has size at least $\maxi$. Intuitively, for any \bin $B\in \configuration_i$ whose largest item has size at least $\maxi$, $\varphi_i(B)$ indicates the \bin in $\target$ that $B$ must reconfigure to.
	Our goal is that, at the end of the final stage, the domain of $\varphi_i$ will include all \bins of $\configuration_i$, and for any $B\in \configuration_i$, \bins $B$ and $\varphi_i(B)$ contain the same multiset of items, i.e., $\configuration_i$ is $\target$.\\
	
	\noindent\textbf{Mapping Step of the $i$th stage: Choosing $\varphi_i$}\\
	Compute $\ell(\configuration_i,\target)$, i.e., the largest item size in $\configuration_i$ that is not settled. If $\ell(\configuration_i,\target) < \maxi$, then the $i$th stage terminates here and the algorithm moves on to the $(i+1)$st stage. Otherwise, we initialize $\varphi_i$ to be a settling bijection for item size $u_{i-1}$, and then extend $\varphi_i$ as follows. Let $\configuration_i'$ to be the multiset consisting of all \bins in $\configuration_i$ whose largest item has size $\maxi$, and, let $\target'$ be the multiset consisting of all \bins in $\target$ whose largest item has size $\maxi$. If the numbers of \bins in the multisets $\configuration_i'$ and $\target'$ are not equal, we make them equal by taking the smaller of $\configuration_i'$ and $\target'$ and adding sufficiently many \bins whose items are all smaller than $\maxi$. We extend $\varphi_i$ by adding to it an arbitrary one-to-one mapping from $\configuration_i'$ to $\target'$. 
	
	\begin{example}
		Suppose that $\configuration_2 = \multiset{\multiset{{32,8}},\multiset{{8,8,4,4}},\multiset{{8,8,4,2,2}},\multiset{{8, 8}}, \multiset{{4,4,1}}}$ and suppose that $\target = \multiset{\multiset{{32,8,1}}, \multiset{{8, 8, 8, 8}}, \multiset{{8,8,4,2}},\multiset{{4,4,2}},\multiset{{4,4}}}$. At this point, we are at the start of the second stage of the algorithm, only $32$ is settled, and $\varphi_i$ initially maps $\multiset{{32,8}}$ to $\multiset{{32,8,1}}$.
		Therefore, $u_2 = 8$, $\configuration_2'  = \multiset{\multiset{{8,8,4,4}}, \multiset{{8,8,4,2,2}},\multiset{{8, 8}}}$, and $\configuration_{\mathcal{T}}' = \multiset{\multiset{{8, 8, 8, 8}}, \multiset{{8,8,4,2}},\multiset{{4,4,2}}}$ (the last \bin is included in $\configuration_{\mathcal{T}}'$ to ensure $|\configuration_{2}'| = |\configuration_{\mathcal{T}}'|$). An arbitrary one-to-one mapping between \bins in $\configuration_{2}'$ to \bins in $\configuration_{\mathcal{T}}'$ extends the domain of $\varphi_i$ to \bins that contain $8$. For example, $\multiset{{8,8,4,4}} \leftrightarrow \multiset{{8, 8, 8, 8}}$, $\multiset{{8,8,4,2,2}} \leftrightarrow \multiset{{8,8,4,2}}$, and $\multiset{{8, 8}}\leftrightarrow \multiset{{4,4,2}}$.
	\end{example}
	
	\noindent\textbf{Action Step of the $i$th stage: Moving items until item size $\maxi$ is settled}\\
	This step is repeated until the following termination condition is met: for each \bin $B$ in the domain of $\varphi_i$, the number of items of size $\maxi$ is the same in \bins $B$ and $\varphi_i(B)$. If the termination condition is not met, we consider two \bins: a \emph{surplus \bin} 
	$B_s$ containing more items of size $\maxi$ than in $\varphi_i(B_s)$ and a \emph{deficit \bin} $B_d$ containing fewer items of size $\maxi$ than in $\varphi_i(B_d)$. In the example above, $B_s=\multiset{{8, 8}}$ and $B_d = \multiset{{8,8,4,4}}$. We now describe a procedure for moving an item of size $\maxi$ from $B_s$ to $B_d$. 
	
	We move an item of size $\maxi$ from $B_s$ to $B_d$ in two phases: 
	in the \emph{compression phase}, items of size at most $\maxi$ are moved between \bins to accumulate enough slack  in at least one \bin $B_{temp}$ to host an item of size $\maxi$, and in the \emph{transfer phase}, an item of size $\maxi$ is moved from $B_s$ to $B_d$ by using $B_{temp}$  as a temporary host for the item. We now provide terminology and a subroutine called \bundleformrm{B}{p} to aid in formalizing the two phases.
	
	We partition the slack of each \bin $B$ into \emph{slack items} that are maximal powers of 2; for example, when $\mathit{slack}(B) = 14$, $B$ contains slack items of sizes ${2},{4},$ and ${8}$. For clarity, we refer to items in the configuration as ``actual items",  
	using $\mathit{\All}(B)$ to denote the multiset formed by the actual and slack items in $B$, all powers of 2, which sum to $\capacity$. We use the term \emph{bundle} to refer to any multiset of items from $\mathit{\All}(B)$, the sum of which (actual items and slack items) will be called its \emph{bsum}.
	
	\noindent\textbf{Procedure \bundleformbf{B}{p}.}
	Given a \bin $B$ and an integer $p < \log_2 \capacity$ such that $\mathit{\All}(B)$ contains at least one element of size at most $2^p$, the procedure returns two disjoint bundles, each with bsum equal to $2^p$. To accomplish this, the procedure starts with bundles consisting of single actual or slack items, and repeatedly merges any two bundles with bsum $2^i$ to form a new bundle with bsum $2^{i+1}$ until there is no pair of bundles each with bsum $2^i$ for any $i < p$. Finally, two bundles with bsum $2^p$ are chosen arbitrarily and returned as output.
	
	\begin{example}\label{exam:bu}
		Suppose $\capacity = 64$, and $B= \{ 32, 4, 4, 4, 4, 2 \}$. Then $\mathit{slack}(B) = 14$, and slack items are $\underbar{8}, \underbar{4},\underbar{2}$ (underbars represent slack items). Therefore, $\mathit{\All}(B) = \{ 32,\underbar{8},\underbar{4},4,4,4,4,\underbar{2},2 \}$.
		Suppose $p = 3$, and note that there are elements less than or equal to $2^p$ in $\mathit{\All}(B)$.
		The execution of \bundleformrm{B}{3} repeatedly merges bundles of $\mathit{\All}(B)$ that have bsum less than 8. After all merges, the bundles will be \linebreak $\{32\}, \{\underbar{8}\}, \{\underbar{4},4\}, \{4,4\}, \{4,\underbar{2},2\}$. Note that there are four bundles with bsum equal to 8. \bundleformrm{B}{3} returns any two of these bundles, say the first two, $\{\underbar{8}\}$ and $\{\underbar 4, 4\}$, as its output.
	\end{example}

	\noindent\textbf{The Compression Phase.} This phase reduces the number of slack items in the configuration by repeatedly merging pairs of slack items of the same size, resulting in enough space to make a transfer from $B_s$ to $B_d$ possible. Merges take place on any two \bins that contain slack items of equal size less than $\maxi$. More formally, while there exist at least two \bins $B_1,B_2$ that each have a slack item of size $2^q$ with $q < \log_2(u_i)$, we execute a procedure, called \emph{\ms}, that operates as follows.

 \noindent\textbf{Procedure \ms.} Apply {\bundleformrm{B_1}{q}} to obtain two disjoint bundles consisting of elements of $\mathit{\All}(B_1)$, each with bsum $2^q$, then, move the actual items of one of these \bins from $B_1$ to $B_2$.
	
	At the end of the compression phase, we pick one \bin containing a slack item of size at least $\maxi$. In the remainder of the algorithm's description, we will refer to this \bin as $B_{\mathit{temp}}$.
	
	\begin{example}\label{exam:compress}
		Suppose $\capacity = 64$ and let $B_1 = \{ 32, 4, 4, 4, 4, 2\}$ (with slack items $\{\underbar{8},\underbar{4},\underbar{2}\}$) and $B_2 = \{32,16,8\}$ (with slack item $\{\underbar{8}\}$). Each \bin contains a slack item of size $8$; thus $q=3$. \ms calls \bundleformrm{B_1}{3} which, as discussed in Example~\ref{exam:bu}, returns two multisets $\{\underbar{8}\}$ and $\{ \underbar{4},4 \}$. The  actual item, of size 4, in the second multiset is moved to $B_2$. After the merge, \bins become $B_1 = \{ 32, 4, 4, 4, 2\}$ and $B_2 = \{32,16,8,4\}$. Note that the set of slack items of $B_1$ is changed to $\{\underbar{16},\underbar{2}\}$.
	\end{example}

	\noindent\textbf{The Transfer Phase.} We transfer an item of size $\maxi$ from $B_s$ to $B_d$ in $\configuration_i$ in at most three steps. If $B_d$ has slack at least $\maxi$, then we transfer an item of size $\maxi$ from $B_s$ to $B_d$ directly. Otherwise, we perform the following three steps:
	(1) move an item size $\maxi$ from $B_s$ to $B_{\mathit{temp}}$, then (2) apply \bundleformrm{B_d}{\log_2(\maxi/2)} to obtain two bundles from $\All({B_d})$, each with bsum equal to $\maxi/2$, and move all actual items from these two bundles to $B_s$, then (3) move an item of size $\maxi$ from $B_{\mathit{temp}}$ to $B_d$. 
	
	This concludes the description of \AlgPowOfTwo, whose pseudocode is provided in Algorithms \ref{alg:reconfpow2}, \ref{alg:bundles}, and \ref{alg:mergeslack}. The proof of correctness is provided in Section~\ref{sec:pow2correct}.

\linespread{1.0}\selectfont
 \begin{algorithm}[h!]
    \footnotesize
		\caption{Reconfiguration when the total slack is at least $\largestSource$}\label{alg:reconfpow2}
		\begin{algorithmic}[1]
            \Statex {\color{gray} \% The set of item sizes is $\{u_1,\ldots,u_k\}$, where $u_i$ denotes the $i$th largest item size}
            \State $k \leftarrow$ size of the underlying set $U$ \Comment{$k$ is the number of item sizes}
            \For{$i \leftarrow 1,\ldots,k$} \Comment{{\bf Start of the $i$th stage}}\label{stageloop}
            \State $\configuration_i \leftarrow$ configuration at start of the $i$th stage
            \State Compute $\ell(\configuration_i,\target)$ \Comment{{\bf Mapping Step}}
            \If{$\ell(\configuration_i,\target) < \maxi$}
                \State continue to the next iteration of \textbf{for} loop at line \ref{stageloop}
            \Else
                \State $\varphi_i \leftarrow $ a settling bijection for size $u_{i-1}$
                \State $\configuration_i' \leftarrow $ multiset of all \bins in $\configuration_i$ with largest item size $\maxi$
                \State $\target' \leftarrow $ multiset of all \bins in $\target$ with largest item size $\maxi$
                \If{$|\configuration_i'| < |\target'|$}
                    \State Add $|\target'|-|\configuration_i'|$ \bins to $\configuration_i'$ from $\configuration_i$, each with largest item size $< \maxi$
                \ElsIf{$|\configuration_i'| > |\target'|$}
                    \State Add $|\configuration_i'|-|\target'|$ \bins to $\target'$ from $\target$, each with largest item size $< \maxi$
                \EndIf
                \State $\varphi_i' \leftarrow$ an arbitrary bijection from $\configuration_i'$ to $\target'$
                \State extend $\varphi_i$ by including $\varphi_i'$
            \EndIf
			\While{$\exists B$ in domain of $\varphi_i$ such that
   
   \hspace*{8mm} $B$ and $\varphi_i(B)$ have different \# of items of size $\maxi$} \Comment{{\bf Action Step}}
                \State $B_s \leftarrow$ \bin containing more items of size $\maxi$ than $\varphi_i(B_s)$
                \State $B_d \leftarrow$ \bin containing fewer items of size $\maxi$ than $\varphi_i(B_d)$
                \While{$\exists B_1,B_2$ and $q < \log_2(\maxi)$ such that
                
                \hspace*{13mm} each of $B_1,B_2$ has a slack item of size $2^q$} \Comment{{\bf Compression Phase}}
                    \State Merge-Slack($B_1,B_2,q$)
                \EndWhile
                \State $B_{\mathit{temp}} \leftarrow$ a \bin containing a slack item of size at least $\maxi$
                \If{$B_d$ has slack at least $\maxi$} \Comment{{\bf Transfer Phase}}
                    \State move an actual item of size $\maxi$ from $B_s$ to $B_d$
                \Else
                    \State move an actual item of size $\maxi$ from $B_s$ to $B_{\mathit{temp}}$
                    \State $(\mathit{bundle}_1,\mathit{bundle}_2) \leftarrow \bundleformrm{B_d}{\log_2(\maxi/2)}$
                    \State move all actual items in $\mathit{bundle}_1$ and $\mathit{bundle}_2$ from $B_d$ to $B_s$
                    \State move an actual item of size $\maxi$ from $B_{\mathit{temp}}$ to $B_d$
                \EndIf
            \EndWhile
        \EndFor
		\end{algorithmic}
\end{algorithm}

 \begin{algorithm}[h!]
    \footnotesize
		\caption{\bundleformrm{B}{p}}\label{alg:bundles}
		\begin{algorithmic}[1]
            \State $\mathit{BundleList} \leftarrow $ list of bundles, each bundle consisting of one element of $\All(B)$
            \For{$i \leftarrow 0,\ldots,p-1$}
                \While{$\exists b_1,b_2 \in \mathit{BundleList}$, each with bsum equal to $2^i$}
                    \State remove $b_1,b_2$ from $\mathit{BundleList}$
                    \State create a new bundle by merging $b_1,b_2$, then append it to $\mathit{BundleList}$
                \EndWhile
            \EndFor
            \State Choose any two bundles $b_1,b_2$ from $\mathit{BundleList}$ that each have bsum equal to $2^p$
            \State \Return $(b_1,b_2)$
        \end{algorithmic}
\end{algorithm}

 \begin{algorithm}[h!]
    \footnotesize
		\caption{Merge-Slack($B_1,B_2,q$)}\label{alg:mergeslack}
		\begin{algorithmic}[1]
            \State $(\mathit{bundle}_1,\mathit{bundle}_2) \leftarrow \bundleformrm{B_1}{q}$
                    \If{$\mathit{bundle}_1$ contains an actual item}
                        \State move all actual items in $\mathit{bundle}_1$ from $B_1$ to $B_2$
                    \Else
                        \State move all actual items in $\mathit{bundle}_2$ from $B_1$ to $B_2$
                    \EndIf
        \end{algorithmic}
\end{algorithm}

  \linespread{1.1}\selectfont

	\subsection{Correctness of \AlgPowOfTwo}\label{sec:pow2correct}
	In what follows, denote by $k$ the number of distinct item sizes in the underlying set $U$, and recall that the size of each item in each configuration belongs to the set $\{u_1,\ldots,u_k\}$ where each $u_i$ denotes the $i$th largest item size in $U$. To prove the correctness of \AlgPowOfTwo, as described in Section \ref{sec:pow2alg}, we depend on the following invariant that we will prove holds for all stages $i \geq 1$ under the assumption that the total slack across all \bins is at least $\largestSource$.
	\begin{invariant}\label{inv:settled}
		At the start of the $i$th stage of \AlgPowOfTwo, each item size $u_j$, for $1 \leq j < i$, is already settled.
	\end{invariant}
	First, we show why Invariant \ref{inv:settled} implies the correctness of the algorithm. At the end of the $k$th stage of the algorithm, i.e., in configuration $\configuration_{k+1}$ at the start of the $(k+1)$st stage of the algorithm, we know from Invariant \ref{inv:settled} that item size $u_k$ is settled, which, by definition, means that there exists a bijection $\varphi$ from the multiset of \bins containing an item of size at least $u_k$ in $\configuration_{k+1}$ to the multiset of \bins containing an item of size at least $u_k$ in $\target$, and $\AtLeast{u_k}{B}{\configuration_{k+1}} = \AtLeast{u_k}{\varphi(B)}{\target}$ for each \bin $\binb$ in the domain of $\varphi$. Since $u_k$ is the smallest item size, it follows that all \bins in $\configuration_{k+1}$ are in the domain of $\varphi$, the bijection $\varphi$ maps each \bin $B \in \configuration_{k+1}$ to a \bin $\varphi(B) \in \target$ such that $B = \varphi(B)$, and thus $\configuration_{k+1} = \target$. Therefore, at the end of the $k$th stage of the algorithm, the target configuration has been reached, as desired.
	
	So, it remains to show that Invariant \ref{inv:settled} holds for all stages $i \geq 1$, which we prove in the remainder of this section.
	
	The invariant vacuously holds for $i=1$ (i.e., at the start of the algorithm) since the range for $j$ is empty.
	
	Next, for any $i \geq 1$ for which the invariant holds, we prove that the invariant holds at the end of the $i$th stage (i.e., at the start of the $(i+1)$st stage).

	In the Mapping Step of the $i$th stage, the algorithm starts by determining $\ell(\configuration_i,\target)$ (the largest item size in $\configuration_i$ that is not settled). If $\ell(\configuration_i,\target) < \maxi$, then the algorithm stops executing the $i$th stage, and we see that Invariant \ref{inv:settled} holds: $\ell(\configuration_i,\target) < \maxi$ implies that $\maxi$ is settled, and, since Invariant \ref{inv:settled} holds at the start of the $i$th stage, we also know that item sizes $u_1,\ldots,u_{i-1}$ are settled. 
	
	So, in the remainder of the proof, we assume that the largest unsettled item size computed in the Mapping Step satisfies $\ell(\configuration_i,\target) \geq \maxi$. Since Invariant \ref{inv:settled} holds at the start of the $i$th stage, we know that item sizes $u_1,\ldots,u_{i-1}$ are settled, so it follows that $\ell(\configuration_i,\target) = \maxi$.
	
	In the following result, we prove that the mapping $\varphi_i$ produced by the Mapping Step satisfies three properties that will be crucial for the correctness of the $i$th stage.
	
	\begin{lemma}\label{lem:mappingproperties}
		At the end of the Mapping Step of the $i$th stage, the following three properties hold:
		\begin{enumerate}[label=(\alph*)]
			\item the domain and range of $\varphi_i$ have the same size,\label{prop:domrange}
			\item the domain of $\varphi_i$ includes all \bins in $\configuration_i$ that contain an item of size at least $\maxi$, and, \label{prop:domincludes}
			\item the range of $\varphi_i$ includes all \bins in $\target$ that contain an item of size at least $\maxi$.\label{prop:rangeincludes}
		\end{enumerate}
	\end{lemma}
 	\begin{proof}
 		The Mapping Step initializes $\varphi_i$ to be a settling bijection for item size $u_{i-1}$, which we know exists since item size $u_{i-1}$ is settled. From the definition of settling bijection for size $u_{i-1}$, we know that properties \ref{prop:domrange}-\ref{prop:rangeincludes} initially hold for $\varphi_i$, but with properties \ref{prop:domincludes} and \ref{prop:rangeincludes} referring to items of size at least $u_{i-1}$.
 		
 		Then, $\varphi_i$ is extended (as described in the Mapping Step of the $i$th stage), and we prove that properties \ref{prop:domrange}-\ref{prop:rangeincludes} hold for item size $\maxi$ as well. 
 		
 		To prove that property \ref{prop:domrange} holds for item size $\maxi$ at the end of the Mapping Step, from the description of the Mapping Step, we see that $\varphi_i$ is extended using a one-to-one mapping whose domain $\configuration_i'$ and range $\target'$ have equal size, and, since the domain and range of $\varphi_i$ had equal size initially, it follows that the domain and range of $\varphi_i$ have equal size at the end of the Mapping Step of the $i$th stage. It remains to verify that the claimed construction of $\configuration_i'$ and $\target'$ is always possible. Since $\varphi_i$ is initially a mapping from $\configuration_i$ and $\target$ whose domain and range have equal size, and all configurations have the same number of \bins, it follows that $\configuration_i$ and $\target$ have the same number of \bins that are \emph{not} in the initial domain and range of $\varphi_i$, respectively. Moreover, by properties \ref{prop:domincludes} and \ref{prop:rangeincludes} of $\varphi_i$ initially, any \bins that were not in the initial domain of $\varphi_i$ only contain items of size at most $\maxi$, and, any \bins that were not in the initial range of $\varphi_i$ only contain items of size at most $\maxi$. Thus, during the construction of $\configuration_i'$ and $\target'$ in the Mapping Step, any \bin whose largest item does not have size $\maxi$ has the property that all its items are smaller than $\maxi$, so there are always \bins available to make $|\configuration_i'|=|\target'|$, as required.
 		
 		To prove properties \ref{prop:domincludes} and \ref{prop:rangeincludes} hold for item size $\maxi$ at the end of the Mapping Step, consider any \bin $B \in \configuration_i$ that contains at least one item of size at least $\maxi$. If $\maxi$ is not the size of the largest item in $B$, then there is an item in $B$ that has size $u_j$ for some $j < i$, and we know that the initial domain of $\varphi_i$ already contained $B$. Otherwise, if $\maxi$ is the size of the largest item in $B$, then from the description of the Mapping Step of the $i$th stage, we know that $\varphi_i$ is extended so that its domain includes $B$. In both cases, at the end of the Mapping Step of the $i$th stage, the domain of $\varphi_i$ contains $B$, so we conclude that all \bins in $\configuration_i$ that contain at least one item of size at least $\maxi$ are included in the domain of $\varphi_i$ at the end of the Mapping Step. An identical proof about the range of $\varphi_i$ proves that each \bin in $\target$ that contains at least one item of size at least $\maxi$ is included in the range of $\varphi_i$ at the end of the Mapping Step.
 	\end{proof} 
	
	Next, suppose that the Mapping Step of the $i$th stage has been executed, and that the termination condition of the Action Step has not yet been met. From the description of the Action Step, this means that there is at least one \bin $B$ in the domain of $\varphi_i$ such that the number of items of size $\maxi$ in $B$ is not equal to the number of items of size $\maxi$ in $\varphi_i(B)$. The following result confirms that there must exist a surplus \bin $B_s$ and a deficit \bin $B_d$. 
	
	\begin{lemma}\label{lem:SurplusDeficit}
		At the start of any iteration of the Action Step, suppose that there is at least one \bin $B$ in the domain of $\varphi_i$ such that the number of items of size $\maxi$ in $B$ is not equal to the number of items of size $\maxi$ in $\varphi_i(B)$. Then, there exists a \bin $B_s$ containing more items of size $\maxi$ than in $\varphi_i(B_s)$, and, a \bin $B_d$ containing fewer items of size $\maxi$ than in $\varphi_i(B_d)$.
	\end{lemma}
	\begin{proof}
		Recall from the model that items cannot be added or removed during reconfiguration, so configurations $\configuration_i$ and $\target$ contain the same number of items of size $\maxi$. By Lemma \ref{lem:mappingproperties}, the domain of $\varphi_i$ contains all \bins in $\configuration_i$ that have at least one item of size $\maxi$, and, the range of $\varphi_i$ contains all \bins in $\target$ that have at least one item of size $\maxi$. In other words, the multiset union of all \bins in the domain of $\varphi_i$ contains all the items of size $\maxi$ in configuration $\configuration_i$, and, the multiset union of all \bins in the range of $\varphi_i$ contains all the items of size $\maxi$ in configuration $\target$, and the number of items of size $\maxi$ in these two unions must be equal. Thus, if there exists a \bin $B$ in the domain of $\varphi_i$ such that the number of items of size $\maxi$ in $B$ is strictly greater than the number of items of size $\maxi$ in $\varphi_i(B)$, then there must exist another \bin $B'$ in the domain of $\varphi_i$ such that the number of items of size $\maxi$ in $B'$ is strictly less than the number of items of size $\maxi$ in $\varphi_i(B')$, and vice versa.
	\end{proof}

	Our next goal is to prove that our algorithm carries out a sequence of moves that results in an item of size $\maxi$ being taken from $B_s$ and put into $B_d$. If the amount of slack in $B_d$ is at least $\maxi$, then this can be carried out directly in one move. However, more generally, there is not necessarily enough slack in $B_d$. We prove that, during the compression phase of the Action Step, our algorithm is able to reach a configuration where at least one \bin $B_{\mathit{temp}}$ has slack at least $\maxi$. We start by proving the correctness of the procedure \bundleformrm{B}{p}.
	
	\begin{lemma}\label{lem:BundlesCorrect}
		Let $B$ be an arbitrary \bin, and, let $p$ be an integer such that $p < \log_2 \capacity$ and there is at least one element of size at most $2^p$ in $\mathit{\All}(B)$. Then, \bundleformit{B}{p} returns two disjoint bundles consisting of elements of $\mathit{\All}(B)$, each with bsum equal to $2^p$. Moreover, at least one of the two returned bundles contains an actual item.
	\end{lemma}
	\begin{proof}
		First, note that \bundleformrm{B}{p} terminates after a finite number of iterations: an item in $\mathit{\All}(B)$ with initial size $2^j$ can be involved in at most $p-j$ merges, since more merges would require the merging of two bundles with bsum $2^p$.
		
		Next, we show that when \bundleformrm{B}{p} terminates, there are at least two bundles that each have bsum equal to $2^p$.
		When \bundleformrm{B}{p} terminates, the following conditions hold:
		
		\begin{itemize}
			\item[(i)] For any $i< p$, there is at most one bundle with bsum $2^i$, since, otherwise, two such bundles would have been merged before \bundleformrm{B}{p} terminates. 
			\item[(ii)] There exists at least one bundle of size exactly $2^i$ for some $i \in \{0,\ldots,p\}$. This is because, by assumption, $\mathit{\All}(B)$ contains at least one element $y$ of size $2^j$ for some $j \leq p$, and, when \bundleformrm{B}{p} terminates, the final bundle containing $y$ will have bsum at least $2^j$ and at most $2^p$: otherwise, if it has a bsum more than $2^p$, two bundles each with bsum $2^p$ must have been merged, but the algorithm does not merge bundles with bsum $2^p$ or larger.
		\end{itemize} 
		
		For the sake of contradiction, assume that after \bundleformrm{B}{p} terminates, the following holds: (iii) there are 0 or 1 bundles with bsum $2^p$. Then, from (i), (iii), and the fact that the sum of elements in $\mathit{\All}(B)$ sum to $\capacity$, a bundle with bsum $2^i$, for $i \leq p$, represents a `1' as the $(i+1)$st least significant digits in the binary encoding of $\capacity$, and from (ii), at least one such bundle exists. We conclude that there is at least one `1' among the $p+1$ least significant digits in the binary encoding of $\capacity$.
		This contradicts the fact that $\capacity$ is a power of 2 greater than or equal to $2^{p+1}$ (since $p$ is an integer strictly less than $\log_2 \capacity$, by assumption), because the $p+1$ least significant digits in the binary representation of such a power of 2 are all 0.
		
		Finally, we prove that, when \bundleformrm{B}{p} terminates, there is at most one bundle with bsum equal to $2^p$ that consists entirely of slack items. To obtain a contradiction, assume that there are two bundles $b_1,b_2$ that each have bsum equal to $2^p$ and that consist entirely of slack items (with each such slack item having size at most $2^p$). Then, merging these two bundles gives a bundle of slack items with bsum equal to exactly $2^{p+1}$. But this gives a contradiction: recall that the slack items in a \bin are the maximal powers of 2 whose sum adds up to the total slack in the \bin $B$, so the the slack items in $b_1,b_2$ cannot exist (i.e., they could all be replaced with a single slack item of size $2^{p+1}$). Therefore, by way of contradiction, we proved that there is at most one bundle with bsum equal to $2^p$ that consists entirely of slack items, so we conclude that each other bundle with bsum equal to $2^p$ contains at least one actual item.
		
		In conclusion, we have shown that, when \bundleformrm{B}{p} terminates, there are at least two bundles that each have bsum equal to $2^p$ that can be returned, and, for any two returned bundles, at least one of them contains an actual item, as desired.
	\end{proof}

	The compression phase consists of repeated executions of the \ms procedure. When proving the correctness of the compression phase, we will use the following property of \ms.

	\begin{lemma}\label{lem:MergeslackCorrect}
		For any two \bins $B_1$ and $B_2$, each having exactly one slack item of size $2^q$ with $q < \log_2(u_i)$, the execution of \ms on these two \bins results in a configuration with at least one \bin having a slack item of size $2^{q+1}$, and the total number of slack items of size $2^q$ across all \bins is reduced. Moreover, each actual item moved during the execution of \ms has size strictly less than $\maxi$.
	\end{lemma}
	\begin{proof}
		Consider any two \bins $B_1$ and $B_2$, each having exactly one slack item of size $2^q$ with $q < \log_2(u_i)$. In the execution of \ms on these two \bins, the procedure \bundleformrm{B_1}{q} is executed, and by Lemma \ref{lem:BundlesCorrect}, this produces two disjoint bundles consisting of elements of $\mathit{\All}(B_1)$, each with bsum equal to $2^q$. Since these two bundles are disjoint, one of the bundles (call it $\mathit{Move}$) does not contain the slack item (call it $y$) of $B_1$ that has size $2^q$. Since slack item sizes are defined to be maximal powers of 2, and each bundle in the output of {\bundleformrm{B_1}{q}} has bsum $2^{q}$, $\mathit{Move}$ cannot be entirely formed by slack items (otherwise, its slack items, together with $y$, should have formed a larger slack item in $B_1$ of size $2^{q+1}$). \ms moves the actual items of the $\mathit{Move}$ bundle from $B_1$ to $B_2$. This is possible because the sum of the actual items in such a bundle is at most $2^q$, which is no more than the available slack in $B_2$ since $B_2$ is assumed to contain a slack item of size $2^q$. As a result of this move, there will be two slack items of size $2^q$ in $B_1$, which are subsequently combined into one slack item of size $2^{q+1}$. Overall, there are fewer slack items of size $2^q$ at the end of this execution of \ms than at the start, since: the slack item $y$ in $B_1$ disappears as a result of the final merge of two slack items of size $2^q$ in $B_1$, and, the number of slack items of size $2^q$ in $B_2$ did not increase since no \bin can ever have more than one slack item of a given size (slack item sizes are defined to be maximal powers of 2). Also, note that each actual item moved during this execution of \ms has size strictly less than $\maxi$ since the bsum of all items in $\mathit{Move}$ is $2^{q} < \maxi$.
	\end{proof}
	
	Using Lemma \ref{lem:MergeslackCorrect}, we now prove that, after a finite number of \ms executions, we end up in a configuration where there is at most one slack item of size $2^i$ for each $i < \log_2(\maxi)$. Moreover, in such a configuration, there exists at least one \bin $B_{\mathit{temp}}$ that has slack at least $\maxi$.
	
	\begin{lemma}\label{lem:CompressCorrect}
		Suppose that, at the start of the compression phase in the $i$th stage, there is at least one \bin $B$ in the domain of $\varphi_i$ such that the number of items of size $\maxi$ in $B$ is not equal to the number of items of size $\maxi$ in $\varphi_i(B)$. Then, during the compression phase of the $i$th stage, a configuration is reached that contains a \bin $B_{\mathit{temp}}$ with slack at least $\maxi$. Moreover, each actual item moved during the execution of the compression phase has size strictly less than $\maxi$.
	\end{lemma}

	\begin{proof} 
		First, we prove that the total slack across all \bins is at least $\maxi$. Suppose that the start of the compression phase in the Action Step of the $i$th stage has been reached. To obtain a contradiction, assume that the total slack across all \bins is strictly less than $\maxi$, which implies that $\largestSource < \maxi$ (since we are in the process of proving Invariant \ref{inv:settled} under the assumption that the total slack across all \bins is at least $\largestSource$). But, $\largestSource < \maxi$ implies that $\largestSource < u_j$ for each $j \leq i$, since $u_j$ is defined as the $j$th largest item size. We prove, by induction, that for each $j \leq i$, it is the case that $\configuration_j = \source$ and the Action Step of the $j$th stage is skipped. When $j=1$, we know that $\configuration_j = \source$ since $\configuration_1$ is the configuration at the start of the first stage, and, the first stage terminates during the Mapping Step since we showed above that $\largestSource < u_1$. Under the induction assumption that $\configuration_{j-1} = \source$ and the Action Step of the $(j-1)$th stage is skipped for some $2 \leq j \leq i$, we see that $\configuration_{j} = \configuration_{j-1} = \source$ since no items were moved during the $(j-1)$th stage (since the Action Step of the $(j-1)$th stage was skipped), and we see that the Action Step of the $j$th stage will be skipped since we showed above that $\largestSource < u_{j}$, which concludes the inductive step. Thus, we have shown that the Action Step will be skipped in the $i$th stage, which contradicts the fact that the start of the compression phase in the Action Step of the $i$th stage has been reached. By reaching a contradiction, we confirm that the total slack across all \bins cannot be strictly less than $\maxi$, as desired.

		Now, suppose that $\maxi = 2^x$ for some $x \geq 0$. The compression phase consists of repeatedly executing \ms on two \bins that each have a slack item of size $2^q$ with $q < \log_2(\maxi)=x$. By Lemma \ref{lem:MergeslackCorrect}, each actual item moved during each such execution of \ms has size strictly less than $\maxi$, and, each such execution results in an overall decrease in the number of slack items of size $2^q$. So, after a finite number of executions of \ms, we reach a configuration where there is at most one slack item of each size $2^q$ for $0 \leq q < x$. We claim that, when this configuration is reached, there is a slack item of size at least $2^x = \maxi$, and we let $B_{\mathit{temp}}$ be the \bin containing that slack item. To prove this claim, assume otherwise, and notice that since there is at most one slack item of any size $2^q$ for $q < x$, the total size of slack items across all \bins will be at most $\sum_{q=0}^{x-1} 2^q = 2^x-1 = \maxi-1$, which contradicts the fact, established above, that the total slack across all \bins is at least $\maxi$. 
	\end{proof}

	Next, we prove that during the transfer phase of the Action Step in the $i$th stage, an item from $B_s$ of size $\maxi$ is moved to $B_d$.
	
	\begin{lemma}\label{lem:TransferCorrect}
		The transfer phase of the Action Step consists of a sequence of moves such that the first move takes out an item of size $\maxi$ from $B_s$, the final move puts an item of size $\maxi$ into $B_d$, and all other moved items have size less than $\maxi$.
	\end{lemma}
	\begin{proof}
		By Lemma~\ref{lem:CompressCorrect}, after the compression phase, there is a \bin $B_{\mathit{temp}}$ that has slack at least $u_i$. In the transfer phase, we first move an item of size $\maxi$ from $B_s$ to $B_\mathit{temp}$. Then we attempt to create enough slack in $B_d$ to allow an item of size $\maxi$ to be moved from $B_\mathit{temp}$ to $B_d$, and there are two cases to consider:
		\begin{itemize}
			\item \textbf{$B_d$ does not contain an item of size less than $\maxi$.}\\
			Since $\varphi_i$ was initialized as a settling bijection for item size $u_{i-1}$, and no items with size greater than $u_i$ have been moved in the $i$th stage, it follows that, for each $j < i$, the \bin $\varphi_i(B_d)$ has the same number of items of size $u_j$ as the \bin $B_d$. Moreover, by the definition of deficit \bin $B_d$, the \bin $B_d$ has fewer items of size $\maxi$ than the \bin $\varphi_i(B_d)$, and by assumption, $B_d$ has no items of size less than $\maxi$. Since $B$ and $\varphi_i(B_d)$ have the same capacity, it follows that $B_d$ must already have slack at least $\maxi$, as desired.
			\item \textbf{$B_d$ contains at least one item of size less than $\maxi$.}\\
			We apply \bundleformrm{B_d}{\log_2(\maxi/2)} to obtain two bundles from $\mathit{\All}(B_d)$, and we know that, by Lemma \ref{lem:BundlesCorrect}, each of these bundles has bsum equal to $\maxi/2$ and at least one of them contains an actual item. We move all actual items from these two bundles to $B_s$, and from the fact that the total bsum of these two bundles is exactly $\maxi$, we get the following two facts: (1) moving the actual items from the two bundles to $B_s$ is possible since the first thing we did in the transfer phase was to move out an item of size $\maxi$ from $B_s$ to $B_{\mathit{temp}}$, and, (2) after moving the actual items from the two bundles to $B_s$, the amount of slack in $B_d$ is at least $\maxi$, as desired.
		\end{itemize}
		
		The final move of the transfer phase moves an item of size $\maxi$ from $B_{\mathit{temp}}$ to $B_d$, which is possible because we showed above that the amount of slack in $B_d$ is at least $\maxi$. 
		
		Finally, we confirm that all moved items, other than the first and last, have size strictly less than $\maxi$. In particular, all other moves occur in the second case above (i.e., when $B_d$ contains at least one item of size less than $\maxi$), and note that each actual item that is moved comes from a bundle whose bsum is $\maxi/2$, so the size of the item is bounded above by $\maxi/2$.
	\end{proof}

	By Lemma \ref{lem:TransferCorrect}, the result of the transfer phase in one iteration of the Action Step is that the number of items of size $\maxi$ in the surplus \bin $B_s$ decreases by exactly one, the number of items of size $\maxi$ in the deficit \bin $B_d$ increases by exactly one, and the number of items of size $\maxi$ in any other \bin is left unchanged. Thus, after finitely many iterations of the Action Step, there are no remaining surplus and deficit \bins, so Lemma \ref{lem:SurplusDeficit} implies that, for each \bin $B$ in the domain of $\varphi_i$, the number of items of size $\maxi$ is the same in $B$ and $\varphi_i(B)$ at the end of the $i$th stage (and at the start of the $(i+1)$st stage). Moreover, for each \bin $B$ in the domain of $\varphi_i$, we can conclude that the number of items of size strictly greater than $\maxi$ is the same in $B$ and $\varphi_i(B)$ at the end of the $i$th stage: this was true at the start of the $i$th stage since $\varphi_i$ was initialized to be a settling bijection for item size $u_{i-1}$, and, Lemmas \ref{lem:CompressCorrect} and \ref{lem:TransferCorrect} tell us that no items of size greater than $\maxi$ were moved during the $i$th stage. Altogether, this proves that for each \bin $B$ containing an item of size at least $u_i$, we know $\AtLeast{u_i}{B}{\configuration_{i+1}} = \AtLeast{u_i}{\varphi_i(B)}{\target}$. It follows that $\varphi_i$ is a settling bijection from $\configuration_{i+1}$ to $\target$ for size $\maxi$, which completes the proof of Invariant \ref{inv:settled} at the start of the $(i+1)$st stage, which completes the induction proof of Invariant \ref{inv:settled}. As discussed at the start of this section, this implies that our reconfiguration algorithm from Section \ref{sec:pow2alg} correctly reconfigures $\source$ to $\target$, which implies the following feasibility result.
	
	\begin{lemma}\label{pow2feasible}
		If the total slack across all \bins is at least $\largestSource$, then \AlgPowOfTwo reconfigures $\source$ to $\target$.
	\end{lemma}

	Using Lemma \ref{pow2feasible}, we can now complete the characterization of feasible instances of \textsc{Repacking} when the item sizes and capacity are powers of 2.
	\begin{theorem}\label{thm:pow2correct}
		Suppose that all items in the underlying set $U$ of the configuration are powers of 2, and assume that \bin capacity $\kappa$ is also a power of 2. 
		Let $\largestSource$ denote the size of the largest item that must be moved.
		It is possible to reconfigure $\source$ to $\target$ if and only if the total slack across all \bins is at least $\largestSource$. Moreover, our algorithm \AlgPowOfTwo solves all feasible instances.
	\end{theorem}
	\begin{proof} 
		By Lemma \ref{pow2feasible}, \AlgPowOfTwo reconfigures $\source$ to $\target$ if the total slack across all \bins is at least $\largestSource$.
		
		Conversely, suppose that the total slack across all \bins is less than $\largestSource$, and consider any item with size $u = \largestSource$. To obtain a contradiction, assume that it is possible to reconfigure $\source$ to $\target$. Due to the limited slack, no item with size at least $u$ can be moved in any reconfiguration step, so, there must be a bijection $\varphi$ between the multiset of \bins containing an item of size at least $u$ in $\source$ and the multiset of \bins containing an item of size at least $u$ in $\target$ such that $\AtLeast{u}{B}{\source} = \AtLeast{u}{\varphi(B)}{\target}$ for each \bin $B$ in $\source$ that contains an item of size at least $u$. By definition, this means that item size $u$ is settled in $\source$, which contradicts the definition of $\largestSource$.
	\end{proof}

	
	\section{Special Case: Reconfiguration within Partitions}\label{sect:specialStructure}
	
	In order to efficiently parallelize the process of repacking, we wish to determine whether the source and target configurations can be each split into smaller pieces, each containing at most $\beta$ \bins for some fixed constant $\beta$, such that each piece can be reconfigured independently. In this section, we prove that determining the existence of such splits can be formulated as a variant of a transshipment problem~\cite{Cook1997,Korte2018}, that is, a network flow problem that in turn can be solved using integer linear programming (ILP) in linear time~\cite{Eisenbrand2003}. To do that, we represent an instance of repacking as a directed graph for a problem that we call $\partitionILP$ (defined below) by creating a vertex for each way of assigning the items in each piece to a bounded number of \bins, and adding edges between vertices that capture the movement of a single item from one \bin to another. We will show that if the reconfiguration is possible, we can decompose the total flow into a set of paths, {\em path flows}, such that each corresponds to the reconfiguration of a single piece.
	
	Below, we define the repacking problem \textsc{$\partitionbound$-Repacking-$\capacity$} (Section~\ref{sec-partition-problem}) and  $\partitionILP$ (Section~\ref{sec-partition-ILP}). We solve \textsc{$\partitionbound$-Repacking-$\capacity$} (Theorem~\ref{thm-partition}) by proving that it is equivalent to $\partitionILP$ (Lemma~\ref{lem:partition_reduction}). The proof of Lemma~\ref{lem:partition_reduction} is provided in Sections \ref{sec:appendix_partition_forward_new} and \ref{sec:appendix_partition_backward_new}. The proof of Theorem~\ref{thm-partition} is provided in Section \ref{sec:appendix_thm}.
	
	\begin{lemma}\label{lem:partition_reduction}
		$\partitionILP$ has a feasible solution if and only if the corresponding instance $(\source, \target)$ for \textsc{$\partitionbound$-Repacking-$\capacity$} is reconfigurable.
	\end{lemma}
	
	\begin{restatable}{theorem}{thmpartition}\label{thm-partition}
		For arbitrary positive integer constants $\beta,\kappa$, each instance of \textsc{$\partitionbound$-Repacking-$\capacity$} can be solved in time that is linear in the length of the input's binary encoding.
	\end{restatable}

	\subsection{Defining repacking with partitions}\label{sec-partition-problem}
	
	To decompose our source and target (and hence, all intermediate) configurations into pieces, we split the underlying set $U$ into ``underlying subsets'' and allocate a specific number of \bins to each ``underlying subset''. Then, each configuration of $U$ can be viewed as the disjoint union of ``pieces'', where each piece is a configuration of an ``underlying subset''. By restricting moves to occur within pieces, each piece can be reconfigured independently. Below, we formalize these definitions, and extend notations of adjacency and reconfiguration sequences to capture the idea of repacking with partitions.
	
	To formalize ``underlying subsets'', we define a \emph{partition} $\partition = \sequence{\partitionpart_1, \partitionpart_2, \ldots, \partitionpart_{|\partition|}}$ of underlying set $\universe$ to be a sequence of parts, where a {\em part} is a pair $\partitionpart_i = \left(\universe_i, \binnumber_i\right)$, $\universe_i$ is a multiset of items, $\binnumber_i$ is a positive integer number of \bins, and the multiset union of all $\universe_i$'s is equal to $\universe$; it is  a {\em $\boundedpartition$} if $\binnumber_i \leq \partitionbound$ for all values of $i$.
	
	Given a partition $\partition$, we define an \emph{assignment of $\partition$ for capacity $\capacity$}
	as a sequence of multisets of \bins $\assignment = \sequence{\assignmentpart_1, \assignmentpart_2, \ldots, \assignmentpart_{|\assignment|}}$ with $|\assignment| = |\partition|$, such that for each part $\partitionpart_i = \left(\universe_i, \binnumber_i\right)$ in $\partition$, 
	the multiset union of items in $\assignmentpart_i$ is equal to $\universe_i$, 
	the number of \bins in $\assignmentpart_i$ is $\binnumber_i$, and 
	each \bin in $\assignmentpart_i$ has volume at most $\capacity$. The multiset $\assignmentpart_i$ in $\assignment$ is called the $i$th {\em portion}, sometimes written as $\portion{\assignment}{i}$. An assignment $\assignment$ is \emph{consistent with a configuration $\configuration$} if the disjoint multiset union of all $\assignmentpart_i$'s is equal to $\configuration$; in this case, the underlying set of $\configuration$ will be the union of the $\universe_i$'s. 
	
	In a reconfiguration sequence that ensures that each move occurs within a part, each configuration can be expressed as an assignment, where two consecutive assignments differ only in a single portion. Two assignments $\assignment$ and $\assignment'$ of a partition $\partition$ for capacity $\capacity$ are {\em adjacent} if there exists exactly one index $j$ such that $\portion{\assignment}{j}$ and $\portion{\assignment'}{j}$ are adjacent (viewed as configurations with underlying set $\universe_j$) and, for all $i \ne j$, 
	$\portion{\assignment}{i} = \portion{\assignment'}{i}$. 
	Accordingly, a
	{\em $\partition$-conforming reconfiguration sequence} for capacity $\capacity$ from a source configuration $\source$ to a target configuration $\target$, where $\partition$ is a partition for the underlying set of $\source$ (and $\target$), is a sequence of assignments of $\partition$ for capacity $\capacity$
	such that: the first and last assignments are consistent with $\source$ and $\target$, respectively, and, for each pair of consecutive assignments in the sequence, the two assignments are adjacent.

	For any positive integer constants $\partitionbound, \kappa$, we define the following problem as  \textsc{$\partitionbound$-Repacking-$\capacity$}, with instance $(\source, \target)$:
	
	\begin{description}
		\item[Input:] Legal source and target configurations $\source$ and $\target$ for capacity $\capacity$.
		\item[Question:] Is there a $\boundedpartition$ $\partition$ of the underlying set of $\source$ such that there exists a $\partition$-conforming reconfiguration sequence for capacity $\capacity$ from $\source$ to $\target$?
	\end{description}
	
	\begin{example}\label{example:partition}
		Let $\capacity = 9$ and $\beta = 2$. Suppose that the input instance to the 2-\textsc{Repacking}-9 problem is $(\source,\target)$ with $\source = \multiset{\multiset{1,2,5},\multiset{2,2,3},\multiset{1,1,3,3},\multiset{4,4}}$ and $\target = \multiset{\multiset{1,5},\multiset{2,2,2,3},\multiset{1,3,3},\multiset{1,4,4}}$. To show that this is a yes-instance, we first define a 2-bounded partition $\partition = ((\multiset{1,2,2,2,3,5},2),(\multiset{1,1,3,3,4,4},2))$. To create an assignment $\assignment_{\mathcal{S}}$ of $\partition$ that is consistent with $\source$, we let $\assignment_{\mathcal{S}} = (S_1,S_2)$ where the first portion $S_1$ is the multiset consisting of the first and second \bins of $\source$, and the second portion $S_2$ is the multiset consisting of the third and fourth \bins of $\source$. To create an assignment of $\partition$ that is consistent with $\target$, we let $\assignment_{\mathcal{T}} = (T_1,T_2)$ where the first portion $T_1$ is the multiset consisting of the first and second \bins of $\target$, and the second portion $T_2$ is the multiset consisting of the third and fourth \bins of $\target$. The source and target configurations, along with these two assignments, are illustrated in Figure \ref{fig:Assignments}. Finally, we see that there is a $\partition$-conforming reconfiguration sequence from $\source$ to $\target$ consisting of 3 assignments of $\partition$: $(\assignment_{\mathcal{S}},(T_1,S_2),\assignment_{\mathcal{T}})$, where the first move is the item of size 2 in the first \bin of $S_1$ to the second \bin of $S_1$, and the second move is an item of size 1 from the first \bin of $S_2$ to the second \bin of $S_2$.
	\end{example}
	\begin{figure}[h!]
		\centering
		\includegraphics[scale=0.4]{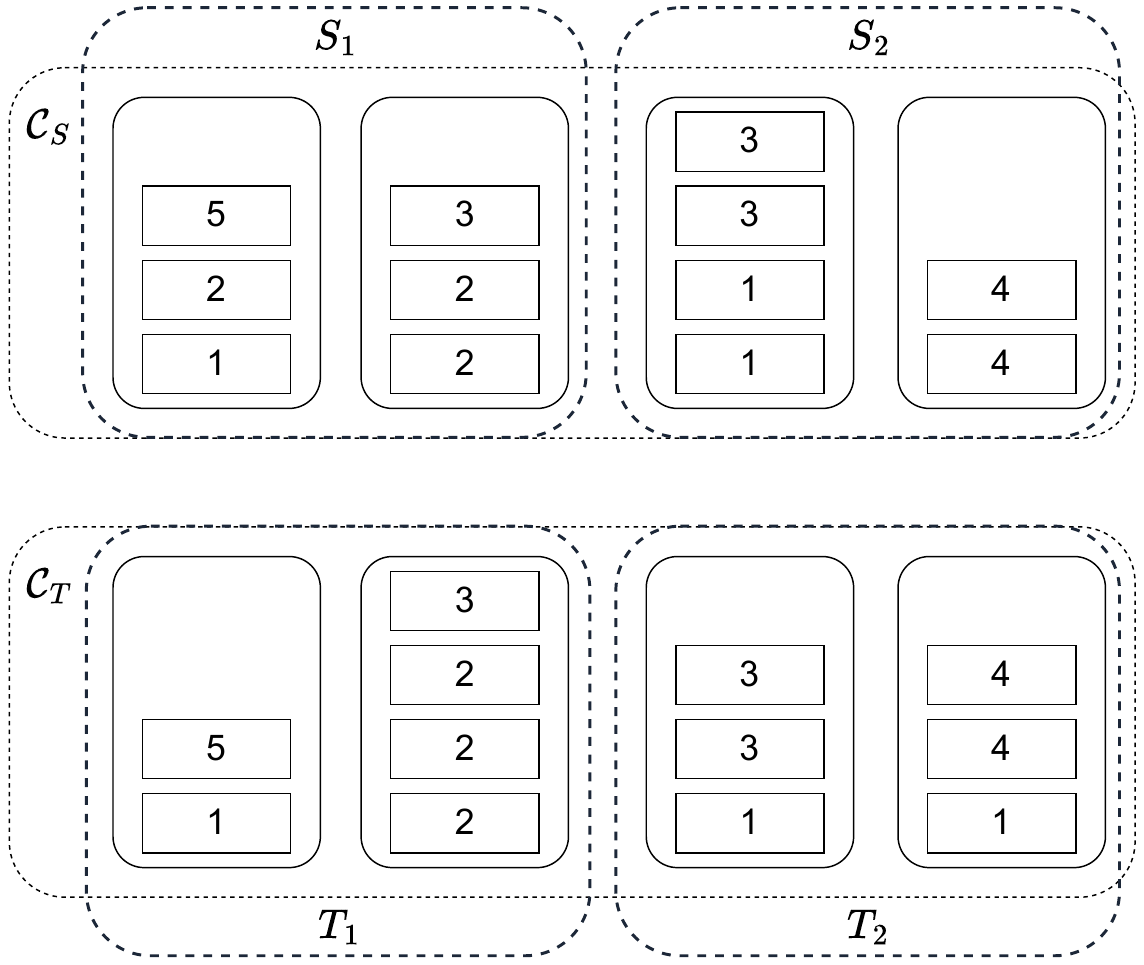}
		\caption{An illustration of $\source$ and $\target$ in Example \ref{example:partition}. The portions $S_1,S_2$ of $\assignment_{\mathcal{S}}$ and the portions $T_1,T_2$ of $\assignment_{\mathcal{T}}$ are circled with a dashed line and labeled.}
		\label{fig:Assignments}
	\end{figure}
	
	\subsection{Defining \texorpdfstring{$\partitionILP$}{Partition ILP}}\label{sec-partition-ILP}
	
	A transshipment problem is a situation in which there are multiple ``origins" that supply a certain good, multiple ``destinations" to which the good must be delivered, and multiple ``intermediate" locations that can be used to transfer the good along the way (and potentially switching from one form of transportation to another). The locations also have numeric ``demand" values: the demand value at an origin represents how much of the good needs to be shipped out from that location (usually a non-positive value to represent the desired net decrease in the amount of good at that location), the demand value at a destination represents how much of the good must be delivered to that location (usually a non-negative value to represent the desired net increase in the amount of good at that location), and the intermediate locations each have demand equal to zero (as they are considered only temporary stops for the goods). Given a transshipment problem, the usual goal is to determine a transportation plan that optimizes the cost and/or time of delivering the good in a way that satisfies all the demands. In order to represent our reconfiguration problem as a directed graph, we view it as a transshipment problem where instead of transporting a good from origin nodes to destination nodes through intermediate nodes, the network represents the reconfiguration of a packing. Since our goal is to determine feasibility, not cost, our network does not have edge weights (which are used in transshipment problems to represent the time or cost incurred by each transportation segment) and we do not set an objective function to optimize. The demands at nodes are set in such a way to ensure that our network represents a reconfiguration problem that uses exactly the \bins in the source configuration $\source$ and target configuration $\target$: the demands of the origin nodes induce an assignment $\sourceassignment$ that is consistent with the source configuration, and the demands of the destination nodes induce an assignment $\targetassignment$ that is consistent with the target configuration.
	
	To determine whether it is possible to partition the source and target configurations, we try all possibilities. A configuration with underlying set $\universe_i \subseteq U$ is a {\em subconfiguration} of a configuration $\configuration$ with underlying set $\universe$, namely a multiset of \bins that is formed by deleting zero or more items and removing zero or more empty \bins from $\configuration$. Two subconfigurations $\subconfig,\subconfigtwo$ are \emph{adjacent} if it is possible to form $\subconfigtwo$ from $\subconfig$ by the move of a single item.
	Given a underlying set $\universe$ and a capacity $\capacity$, we define $\boundedsubs{\universe}{\capacity}$ to be the set of non-empty subconfigurations using subsets of items in $\universe$ and at most $\binnumber$ \bins of capacity $\capacity$ (written as $\bsubs$ when $\universe$ and $\capacity$ are clear from context). 
	
	To form our network, we create three nodes for each subconfiguration $\subconfig$ in $\bsubs$: an origin node $\innode_{\subconfig}$ with non-positive demand, an intermediate node $\internode_{\subconfig}$ with zero demand, and a destination node $\outnode_{\subconfig}$ with non-negative demand.
	Using $\inputnodes$,   $\internodes$, and $\outputnodes$ to denote the sets of origin, intermediate, and destination nodes, respectively, the graph $\graph$ is defined as
	$\vertexset = \inputnodes \cup \internodes \cup \outputnodes$ and directed edges 
	$\edgeset = \{\innode_\subconfig \internode_\subconfig: \subconfig \in \bsubs \} 
	\cup \{\internode_{\subconfig} \internode_{\subconfigtwo}: \subconfig, \subconfigtwo \in \bsubs \text{, $\subconfig$ and $\subconfigtwo$ are adjacent} \}
	\cup \{\internode_\subconfig \outnode_\subconfig: \subconfig \in \bsubs \}$.
	
	Next, we formally define $\partitionILP$ and then give an intuitive interpretation of its constraints. In the formal definition below, $\demand: \vertexset \rightarrow \mathbb{Z}$ and $\flow: \edgeset \rightarrow \mathbb{R}_{\geq 0}$ are functions with $\flow$ extended to arbitrary subsets $\edgesetgeneric$ of $\edgeset$ by defining $\flow(\edgesetgeneric) = \sum_{\euv \in \edgesetgeneric} \flow(\euv)$. For any vertex $v \in V(G)$, we use $\inedges(v)$ and $\outedges(v)$ to denote the sets of in-edges to and out-edges from $v$, respectively. We use $\binmult{B}{\configuration}$ to denote the number of \bins that are equal to $B$ in the configuration $\configuration$, namely, the {\em multiplicity} of $B$ in $\configuration$,
	and denote as $\binmult{\subconfig}{\assignment}$ the number of elements of the assignment $\assignment$ that are equal to the subconfiguration $\subconfig$.
	We use $\bintypesnoparameters$ to denote the set of all possible \bins formed of items taken from $\universe$ such that the volume of the \bin is at most $\capacity$.

	\begin{definition}[\partitionILP]\label{def-part}
		\begin{align}
			\nonumber \text{minimize} ~ 0^T \begin{bmatrix} \flow \\ \demand \end{bmatrix} & ~ \\
			- \flow(\outedges(\innode_\subconfig)) - \demand(\innode_\subconfig) &= 0 & ~ \forall \subconfig \in \bsubs \label{con:one} \\
			\flow(\inedges(\internode_{\subconfig})) - \flow(\outedges(\internode_{\subconfig})) &= 0 & ~ \forall \subconfig \in \bsubs \label{con:two} \\
			\flow(\inedges(\outnode_\subconfig)) - \demand(\outnode_\subconfig) &= 0 & ~ \forall \subconfig \in \bsubs \label{con:three} \\
			\sum_{\subconfig \in \bsubs} \binmult{B}{\subconfig} \cdot \demand(\innode_\subconfig) &= -\binmult{B}{\source} & ~ \forall B \in \bintypesnoparameters \label{con:four} \\
			\sum_{\subconfig \in \bsubs} \binmult{B}{\subconfig} \cdot \demand(\outnode_\subconfig) &= \binmult{B}{\target} & ~ \forall B \in \bintypesnoparameters \label{con:five} \\
			\flow(uv) &\in \mathbb{Z}_{\geq 0} & ~ \forall uv \in \edgeset \label{con:six} \\
			\demand(\innode_\subconfig) \in \mathbb{Z}_{\leq 0}, ~
			\demand(\outnode_\subconfig) &\in \mathbb{Z}_{\geq 0} & ~ \forall \subconfig \in \bsubs \label{con:seven} 
		\end{align}
	\end{definition}
	
	We designed $\partitionILP$ in such a way that, intuitively, given a solution $\begin{bmatrix} \flow & \demand \end{bmatrix}$, we can think of $\flow$ as the flow vector and $\demand$ as the demand vector. 
	We wish to ensure that the total flow out of a particular node is equal to the demands for the node (Constraints \ref{con:one}, \ref{con:two}, \ref{con:three}), and that origin, intermediate, and destination nodes have demands in the appropriate ranges (Constraint \ref{con:seven}).
	We want each unit of flow to correspond to a single reconfiguration step, so we require that a feasible flow be non-negative and integral (Constraint \ref{con:six}).
	Constraints \ref{con:four} and \ref{con:five} relate the assigned demands for each node in $X$ and $Z$, respectively, to the multiplicity of \bins in the source and target configurations, respectively. In particular, Constraint \ref{con:four} ensures that the demands $\demand$ from a feasible solution $\begin{bmatrix} \flow & \demand \end{bmatrix}$ can be directly related to a $\boundedpartition$ $\partition$ of the underlying set $\universe$ and an assignment $\sourceassignment$ of $\partition$ that is consistent with $\source$: each unit of demand on each $\innode_\subconfig$ corresponds to a portion in $\sourceassignment$ that is equivalent to $\subconfig$, and, corresponds to a part in $\partition$ with the same items and number of \bins as $\subconfig$. In a similar way, Constraint \ref{con:five} ensures that the demands $\demand$ can be directly related to an assignment $\targetassignment$ that is consistent with $\target$.
	
	\subsection{Setup}~\label{sec:appendix_defs}
	We use the definitions of $\partitionILP$ and the directed graph $\graph$ that were provided in Section \ref{sec-partition-ILP}. From the definition of $\graph$, we see that it is a directed graph with source vertices $\inputnodes$ and sink vertices $\outputnodes$, where $\inputnodes \cap \outputnodes = \emptyset$. A directed path from any $\innode \in \inputnodes$ to any $\outnode \in \outputnodes$ is called an \emph{$(\inputnodes,\outputnodes)$-path}. 
	
	A \emph{flow} on $\graph$ is a function $\flow: \edgeset \rightarrow \mathbb{R}_{\geq 0}$ satisfying the \emph{conservation constraint}: $\flow(\inedges(v)) - \flow(\outedges(v)) = 0$ for all $v \in \vertexset \setminus (\inputnodes \cup \outputnodes)$ (where $\flow$ is applied to arbitrary subsets $\edgesetgeneric$ of $\edgeset$ by $\flow(\edgesetgeneric) = \sum_{\euv \in \edgesetgeneric} \flow(\euv)$). The \emph{value of $\flow$} is defined as $\abs{\flow} = \sum_{\innode \in \inputnodes} \flow(\outedges(\innode))$. A flow is said to be \emph{integral} if $\flow(uv) \in \mathbb{Z}_{\geq 0}$ for all $uv \in \edgeset$. A flow is a \emph{path flow} if it is positive for all edges $uv$ on a $(\inputnodes,\outputnodes)$-path, but zero on all other edges, and we refer to the underlying $(\inputnodes, \outputnodes)$-path as the \emph{path of the flow}. A path flow whose value is one is called a \emph{unit path flow}. 
	
	\subsection{Flow Preliminaries}~\label{sec:appendix_flow}
	This section provides some fundamental results and properties of flows on $\graph$. As the proofs are technical and uninteresting, they have been relegated to Appendix~\ref{app:flows}. 
	
	We first introduce some operations that can be viewed as combining or splitting flows. Given two flows $\flow_1$ and $\flow_2$ on $G$, the \emph{sum of $\flow_1$ and $\flow_2$}, denoted by $\flow = \flow_1 \flowsum \flow_2$, is defined as $\flow(uv) = \flow_1(uv) + \flow_2(uv)$ for all $uv \in \edgeset$. For $k>2$ flows, we use $\flowsummation_{i=1}^{k} \flow_{i}$ to denote the sum of $k$ flows $\flow_1, \ldots ,\flow_k$. A flow $\flow'$ on $\graph$ is said to be a \emph{subflow of $\flow$} if $\flow'(uv) \leq \flow(uv)$ for all $uv \in \edgeset$.
	For a subflow $\flow_2$ of $\flow_1$, the \emph{subtraction of $\flow_2$ from $\flow_1$}, denoted by $\flow = \flow_1 \flowsub \flow_2$, is defined as $\flow(uv) = \flow_1(uv) - \flow_2(uv)$ for all $uv \in \edgeset$.
	
	The following lemma shows how the above operations can be used to form new flows.
	
	\begin{restatable}{lemma}{flowsumsub}\label{lem:flow_sum_sub}
		Let $\flow_1$ and $\flow_2$ be two flows on $\graph$. Then, the following conditions hold:
		\begin{enumerate}
			\item $\flow = \flow_1 \flowsum \flow_2$ is a flow on $G$ with $\abs{\flow} = \abs{\flow_1} + \abs{\flow_2}$, and
			\item if $\flow_2$ is a subflow of $\flow_1$, then $\flow =\flow_1 \flowsub \flow_2$ is a flow on $G$ with $\abs{\flow} = \abs{\flow_1} - \abs{\flow_2}$.
		\end{enumerate}
	\end{restatable}
	
	The following lemma demonstrates how to extract a collection of subflows from an arbitrary flow $\flow$ on $G$ such that the sum of all the subflows in the collection is also a subflow of $\flow$.
	
	\begin{restatable}{lemma}{flowdecomp}\label{lem:flow_decomp}
		For any integral flow $\flow$ on $\graph$, there exist $\abs{f}$ unit path flows on $\graph$, denoted by $\flow_1, \ldots, \flow_{\abs{f}}$, such that:
		\begin{enumerate}
			\item $\flow_i$ is a subflow of $\flow$ for each $i \in \{1,\ldots,\abs{\flow}\}$, and,
			\item $\flowsummation_{i=1}^{\abs{\flow}} \flow_{i}$ is a subflow of $\flow$.
		\end{enumerate}
	\end{restatable}
	
	The following result gives us a useful property of any flow on $G$: the sum of a flow over all edges outgoing from the origin nodes is equal to the sum of the flow over all edges incoming to the destination nodes.
	\begin{restatable}{lemma}{flowconserve}\label{lem:conserve}
		For any flow $\flow$ on $G$, we have
		$$ 
		\sum_{\subconfig \in \bsubs} {\flow}(\outedges(\innode_\subconfig)) = \sum_{\subconfig \in \bsubs} {\flow}(\inedges(\outnode_\subconfig)).$$
	\end{restatable}

	\subsection{Proof of Lemma \ref{lem:partition_reduction} - forward direction}\label{sec:appendix_partition_forward_new}
	
	The goal of this section is to prove the following lemma. 
	\begin{restatable}{lemma}{partitionreductionforward}\label{lem:partition_reduction_forward}
		Consider any instance $(\source, \target)$ for \textsc{$\partitionbound$-Repacking-$\capacity$}. If the corresponding $\partitionILP$ has a feasible solution, then $(\source, \target)$ is a yes-instance.
	\end{restatable}
	
	To prove Lemma~\ref{lem:partition_reduction_forward}, we start with an arbitrary instance $(\source, \target)$ for \textsc{$\partitionbound$-Repacking-$\capacity$}, and construct its corresponding directed graph $G$ and its corresponding $\partitionILP$, as described in Section \ref{sec-partition-ILP}. We assume that $\partitionILP$ has a feasible solution, and let $\begin{bmatrix} \hat{\flow} & \hat{\demand} \end{bmatrix}$ be any such solution.
	
	By the definition of \textsc{$\partitionbound$-Repacking-$\capacity$}, to prove that $(\source, \target)$ is a yes-instance, we must prove that there exists a $\partitionbound$-bounded partition $\partition$ of the universe of $\source$ such that there exists a $\partition$-conforming reconfiguration sequence for capacity $\capacity$ from $\source$ to $\target$. More specifically, it suffices to:
	\begin{enumerate}[label=(\arabic*)]
		\item Define a $\partitionbound$-bounded partition $\partition$ of the universe of $\source$,
		\item Define an assignment $\sourceassignment$ of $\partition$ for capacity $\kappa$ that is consistent with $\source$, and,
		\item Define a sequence $\areconfsequence$ such that:
		\begin{itemize}
			\item Each element of $\areconfsequence$ is an assignment of $\partition$ for capacity $\kappa$.
			\item The first assignment of $\areconfsequence$ is $\sourceassignment$.
			\item The final assignment of $\areconfsequence$ is consistent with $\target$.
			\item Each pair of consecutive assignments in $\areconfsequence$ are adjacent.
		\end{itemize}
	\end{enumerate}
	
	Steps (1) and (2) are provided by the following result.
	\begin{restatable}{lemma}{createpartition}\label{lem:create_partition}
		For any feasible solution $\begin{bmatrix} \hat{\flow} & \hat{\demand} \end{bmatrix}$ to $\partitionILP$, there exist a $\partitionbound$-bounded partition $\partition$ of the universe of $\source$ and an assignment $\sourceassignment$ of $\partition$ for capacity $\capacity$ that is consistent with $\source$. Additionally, for each $\subconfig \in \bsubs$, $\sourceassignment$ has exactly $\abs{\hat{\demand}(\innode_\subconfig)}$ portions that are equal to $\subconfig$.
	\end{restatable}
	
	\begin{proof}
		Consider any feasible solution $\begin{bmatrix} \hat{\flow} & \hat{\demand} \end{bmatrix}$ to $\partitionILP$. Denote by $U$ the universe of $\source$. First, we create partition $\partition$ as follows: for each $\subconfig$ in $\bsubs$, add $\abs{\hat{\demand}(\innode_\subconfig)}$ copies of $(\universe(\subconfig), \abs{\subconfig})$ to $\partition$. Next, create an assignment $\sourceassignment$ as follows: for each $\subconfig$ in $\bsubs$, add $\abs{\hat{\demand}(\innode_\subconfig)}$ copies of $\subconfig$ to $\sourceassignment$.
		
		We see immediately from the construction of $\sourceassignment$ that, for each $\subconfig \in \bsubs$, $\sourceassignment$ has exactly $\abs{\hat{\demand}(\innode_\subconfig)}$ portions that are equal to $\subconfig$, as claimed in the lemma's statement.
		
		Next, from the construction of $\partition$, note that the number of \bins in each part of the partition is equal to $\abs{\subconfig}$ for some $\subconfig \in \bsubs$. By the definition of $\bsubs$, each subconfiguration in $\bsubs$ has at most $\partitionbound$ \bins, i.e., $\abs{\subconfig} \leq \partitionbound$ for all $\subconfig \in \bsubs$. This proves that $\partition$ is a $\partitionbound$-bounded partition, as claimed.
		
		From the constructions of $\partition$ and $\sourceassignment$, we see that for each $\subconfig$ in $\bsubs$, the number of times that $\subconfig$ is present as a portion in $\sourceassignment$ is exactly equal to the number of times that $(\universe(\subconfig), \abs{\subconfig})$ is present as a part in $\partition$ (since both numbers are equal to $\abs{\hat{\demand}(\innode_\subconfig)}$), which proves that $\sourceassignment$ is an assignment of $\partition$. Moreover, since each portion of the assignment is a subconfiguration from $\bsubs$, then by the definition of $\bsubs$, the volume of each \bin in each portion is at most $\capacity$, and hence $\sourceassignment$ is an assignment of $\partition$ for capacity $\capacity$, as claimed.
		
		There are two additional claims left to prove: we need to show that $\partition$ is a partition of $\universe$, and, we need to show that $\sourceassignment$ is consistent with $\source$. Since $\sourceassignment$ is an assignment of $\partition$, the first claim follows from the second. To prove that $\sourceassignment$ is consistent with $\source$, it is sufficient to consider each possible \bin $B$ that can be formed of items taken from $U$ such that the volume of the \bin is at most $\capacity$, and prove that the multiplicity of $B$ is the same in both $\sourceassignment$ and $\source$. In particular, for each $B \in \bintypesnoparameters$, the multiplicity of $B$ in $\sourceassignment$ can be computed by taking the sum of multiplicities of $B$ across all portions of $\sourceassignment$, and we see that:\\
		\begin{align*}
			\sum_{\assignmentpart \in \sourceassignment} \binmult{B}{\assignmentpart} &= \sum_{\subconfig \in \bsubs} \binmult{B}{\subconfig} \cdot \abs{d(\innode_\subconfig)} \tag{By construction of $\sourceassignment$} \\
			&= -\sum_{\subconfig \in \bsubs} \binmult{B}{\subconfig} \cdot d(\innode_\subconfig) \tag{Constraint \ref{con:seven}} \\
			&= \binmult{B}{\source} \tag{Constraint \ref{con:four}}.
		\end{align*} 
	\end{proof}
	
	The remainder of this section is dedicated to Step (3), i.e., the definition of the reconfiguration sequence $\areconfsequence$. The idea is to take $\hat{\flow}$ from the feasible solution to $\partitionILP$, show that it is an integral flow on $\graph$, extract a collection of unit path flows from $\hat{\flow}$, and, from each such unit path flow, obtain a reconfiguration sequence that acts completely within one part of the partition $\partition$. These reconfiguration sequences can then be concatenated together to form a reconfiguration sequence for the entire instance. We start by showing that $\hat{\flow}$ is an integral flow on $\graph$.
	
	\begin{lemma}\label{lem:rearranged_ilp}
		For any solution $\begin{bmatrix} \hat{\flow} & \hat{\demand} \end{bmatrix}$ to $\partitionILP$, $\hat{\flow}$ is an integral flow on $\graph$.
	\end{lemma}
	
	\begin{proof}
		To show that $\hat{\flow}$ is a flow, it suffices to show that the conservation constraint is satisfied. Since $\vertexset \setminus (\inputnodes \cup \outputnodes)$ is simply $\internodes$, Constraint \ref{con:two} of $\partitionILP$ implies that $\hat{\flow}(\inedges(\internode_\subconfig)) - \hat{\flow}(\outedges(\internode_\subconfig)) = 0$ for all $\subconfig \in \bsubs$, so the conservation constraint is satisfied. Moreover, by Constraint \ref{con:six} of $\partitionILP$, it follows that $\hat{\flow}(uv) \in \mathbb{Z}_{\geq 0}$ for any edge $uv$ in $\graph$, so $\hat{\flow}$ is an integral flow on $\graph$. 
	\end{proof}
	
	Next, we apply Lemma~\ref{lem:flow_decomp} to $\hat{\flow}$ to extract a collection of $\abs{\hat{\flow}}$ unit path flows $\hat{\flow}_1,\ldots,\hat{\flow}_{\abs{\hat{\flow}}}$ on $\graph$, each of which is a subflow of $\hat{\flow}$, and whose sum is a subflow of $\hat{\flow}$. From the definition of a unit path flow, each unit path flow in this collection corresponds to an underlying $(\inputnodes,\outputnodes)$-path in $\graph$ (in particular, the positive edges of the flow are the edges of such a path). The following result tells us how to interpret the demands given by $\hat{\demand}$. In particular, the demands given by $\hat{\demand}$ on the origin and destination nodes of $\graph$ specify how many times each of these nodes is used as a source or sink in the underlying $(\inputnodes,\outputnodes)$-paths.
	
	\begin{lemma}\label{lem:number_of_unit_path_flows}
		Consider any collection of unit path flows $\hat{\flow}_1,\ldots,\hat{\flow}_{\abs{\hat{\flow}}}$ such that $\flowsummation_{i=1}^{\abs{\hat{\flow}}} \hat{\flow}_{i}$ is a subflow of $\hat{\flow}$. For each $\subconfig \in \bsubs$, there are exactly $\abs{\hat{\demand}(\innode_{\subconfig})}$ unit path flows from the collection whose underlying path has node $\innode_{\subconfig} \in X$ as its source. Moreover, for each $\subconfig \in \bsubs$, there are exactly $\abs{\hat{\demand}(\outnode_{\subconfig})}$ unit path flows from the collection whose underlying path has node $\outnode_{\subconfig} \in Z$ as its sink.
	\end{lemma}
	
	\begin{proof}
		To prove the first statement, consider any fixed $\subconfig \in \bsubs$. By the construction of $\graph$, there is exactly one edge incident to $\innode_{\subconfig}$, i.e., the outgoing edge $\innode_{\subconfig}\internode_{\subconfig}$. Since the collection $\hat{\flow}_1,\ldots,\hat{\flow}_{\abs{\hat{\flow}}}$ consists of unit path flows, each $\hat{\flow}_i$ evaluated on edge $\innode_{\subconfig}\internode_{\subconfig}$ is either 0 or 1, so evaluating the flow $\flowsummation_{i=1}^{\abs{\hat{\flow}}} \hat{\flow}_{i}$ on edge $\innode_{\subconfig}\internode_{\subconfig}$ tells us precisely how many flows in the collection have $\innode_{\subconfig}\internode_{\subconfig}$ as a positive edge. But if a unit path flow has $\innode_{\subconfig}\internode_{\subconfig}$ as a positive edge, then the underlying path of the flow has $\innode_{\subconfig}$ as its source. Thus, the value of $\left(\flowsummation_{i=1}^{\abs{\hat{\flow}}} \hat{\flow}_{i}\right)(\innode_{\subconfig}\internode_{\subconfig})$ is exactly the number of unit path flows from the collection whose underlying path has node $\innode_{\subconfig} \in X$ as its source. By assumption, $\flowsummation_{i=1}^{\abs{\hat{\flow}}} \hat{\flow}_{i}$ is a subflow of $\hat{\flow}$, so the value of $\left(\flowsummation_{i=1}^{\abs{\hat{\flow}}} \hat{\flow}_{i}\right)(\innode_{\subconfig}\internode_{\subconfig})$ is bounded above by $\hat{\flow}(\innode_{\subconfig}\internode_{\subconfig})$. Since $\begin{bmatrix} \hat{\flow} & \hat{\demand} \end{bmatrix}$ is a feasible solution to $\partitionILP$, Constraints \ref{con:one} and $\ref{con:seven}$ of $\partitionILP$ imply that $\hat{\flow}(\innode_{\subconfig}\internode_{\subconfig}) = -\hat{\demand}(\innode_{\subconfig}) = \abs{\hat{\demand}(\innode_{\subconfig})}$.
		Thus, we have shown that there can be at most $\abs{\hat{\demand}(\innode_{\subconfig})}$ unit path flows from the collection whose underlying path has node $\innode_{\subconfig}$ as its source. To complete the proof that there are exactly $\abs{\hat{\demand}(\innode_{\subconfig})}$ unit path flows from the collection whose underlying path has node $\innode_{\subconfig}$ as its source, we give a proof by contradiction: assume that there are strictly fewer than $\abs{\hat{\demand}(\innode_{\subconfig})}$ unit path flows from the collection whose underlying path has node $\innode_{\subconfig}$ as its source. Since the underlying path of each of the $\abs{\hat{\flow}}$ unit path flows in the collection has exactly one source node from $X$, we obtain:
		\begin{align*}
			\abs{\hat{\flow}} &= \sum_{\innode_\subconfig \in X}\left[ \textrm{number of flows in the collection with $\innode_{\subconfig}$ as its source} \right]\\
			&< \sum_{\subconfig \in \bsubs} \abs{\hat{\demand}(\innode_{\subconfig})} \tag{by assumption and previous paragraph} \\
			&= - \sum_{\subconfig \in \bsubs} \hat{\demand}(\innode_{\subconfig}) \tag{Constraint \ref{con:seven}}\\
			&= \sum_{\subconfig \in \bsubs} \hat{\flow}(\outedges(\innode_\subconfig)) \tag{Constraint \ref{con:one}} \\
			&= \abs{\hat{\flow}} \tag{by definition of $\abs{\hat{\flow}}$}
		\end{align*}
		which implies that $\abs{\hat{\flow}} < \abs{\hat{\flow}}$, a contradiction.
		
		The proof of the second statement is symmetric to the above proof and differs by only a few minor details. To prove the second statement, consider any fixed $\subconfig \in \bsubs$. By the construction of $\graph$, there is exactly one edge incident to $\outnode_{\subconfig}$, i.e., the incoming edge $\internode_{\subconfig}\outnode_{\subconfig}$. Since the collection $\hat{\flow}_1,\ldots,\hat{\flow}_{\abs{\hat{\flow}}}$ consists of unit path flows, each $\hat{\flow}_i$ evaluated on edge $\internode_{\subconfig}\outnode_{\subconfig}$ is either 0 or 1, so evaluating the flow $\flowsummation_{i=1}^{\abs{\hat{\flow}}} \hat{\flow}_{i}$ on edge $\internode_{\subconfig}\outnode_{\subconfig}$ tells us precisely how many flows in the collection have $\internode_{\subconfig}\outnode_{\subconfig}$ as a positive edge. But if a unit path flow has $\internode_{\subconfig}\outnode_{\subconfig}$ as a positive edge, then the underlying path of the flow has $\outnode_{\subconfig}$ as its sink. Thus, the value of $\left(\flowsummation_{i=1}^{\abs{\hat{\flow}}} \hat{\flow}_{i}\right)(\internode_{\subconfig}\outnode_{\subconfig})$ is exactly the number of unit path flows from the collection whose underlying path has node $\outnode_{\subconfig} \in Z$ as its sink. By assumption, $\flowsummation_{i=1}^{\abs{\hat{\flow}}} \hat{\flow}_{i}$ is a subflow of $\hat{\flow}$, so the value of $\left(\flowsummation_{i=1}^{\abs{\hat{\flow}}} \hat{\flow}_{i}\right)(\internode_{\subconfig}\outnode_{\subconfig})$ is bounded above by $\hat{\flow}(\internode_{\subconfig}\outnode_{\subconfig})$. Since $\begin{bmatrix} \hat{\flow} & \hat{\demand} \end{bmatrix}$ is a feasible solution to $\partitionILP$, Constraints \ref{con:three} and $\ref{con:seven}$ of $\partitionILP$ imply that $\hat{\flow}(\internode_{\subconfig}\outnode_{\subconfig}) = \hat{\demand}(\outnode_{\subconfig}) = \abs{\hat{\demand}(\outnode_{\subconfig})}$.
		Thus, we have shown that there can be at most $\abs{\hat{\demand}(\outnode_{\subconfig})}$ unit path flows from the collection whose underlying path has node $\outnode_{\subconfig}$ as its sink. To complete the proof that there are exactly $\abs{\hat{\demand}(\outnode_{\subconfig})}$ unit path flows from the collection whose underlying path has node $\outnode_{\subconfig}$ as its sink, we give a proof by contradiction: assume that there are strictly fewer than $\abs{\hat{\demand}(\outnode_{\subconfig})}$ unit path flows from the collection whose underlying path has node $\outnode_{\subconfig}$ as its sink. Since the underlying path of each of the $\abs{\hat{\flow}}$ unit path flows in the collection has exactly one sink node from $Z$, we obtain:
		\begin{align*}
			\abs{\hat{\flow}} &= \sum_{\outnode_\subconfig \in Z}\left[ \textrm{number of flows in the collection with $\outnode_{\subconfig}$ as its sink} \right]\\
			&< \sum_{\subconfig \in \bsubs} \abs{\hat{\demand}(\outnode_{\subconfig})} \tag{by assumption and previous paragraph} \\
			&= \sum_{\subconfig \in \bsubs} \hat{\demand}(\outnode_{\subconfig}) \tag{Constraint \ref{con:seven}}\\
			&= \sum_{\subconfig \in \bsubs} \hat{\flow}(\inedges(\outnode_\subconfig)) \tag{Constraint \ref{con:three}} \\
			&= \sum_{\subconfig \in \bsubs} \hat{\flow}(\outedges(\innode_\subconfig)) \tag{by Lemma \ref{lem:conserve}} \\
			&= \abs{\hat{\flow}} \tag{by definition of $\abs{\hat{\flow}}$}
		\end{align*}
		which implies that $\abs{\hat{\flow}} < \abs{\hat{\flow}}$, a contradiction.
	\end{proof}
	
	Consider a fixed $\subconfig \in \bsubs$. From Lemma \ref{lem:create_partition}, our assignment $\sourceassignment$ has exactly $\abs{\hat{\demand}(\innode_\subconfig)}$ portions that are equal to $\subconfig$, and, from Lemma \ref{lem:number_of_unit_path_flows}, there are exactly $\abs{\hat{\demand}(\innode_{\subconfig})}$ unit path flows from $\hat{\flow}_1,\ldots,\hat{\flow}_{\abs{\hat{\flow}}}$ whose underlying path has node $\innode_{\subconfig} \in X$ as its source. As we will show next, this means that we can create a correspondence between portions and unit path flows. In particular, for the first portion of $\sourceassignment$, which is equal to some $\subconfig \in \bsubs$, we associate one of the unit path flows whose underlying path has $\innode_{\subconfig} \in X$ as its source, and use the edges of this underlying path to create a sequence of adjacent assignments that differ only in the first portion. Then, from the resulting configuration, we take the second portion, which is equal to some $\subconfigtwo \in \bsubs$, associate one of the unit path flows whose underlying path has $\innode_{\subconfigtwo} \in X$ as its source, and use the edges of this underlying path to create a sequence of adjacent assignments that differ only in the second portion. Proceeding in this manner until each of the portions in $\sourceassignment$ has been reconfigured, the resulting sequence of assignments will be a $\partition$-conforming reconfiguration sequence for capacity $\capacity$ from $\source$ to $\target$.
	
	Algorithm \ref{alg:create_partition_sequence} specifies how to implement the procedure described above to produce the desired $\partition$-conforming reconfiguration sequence. For any node $v \in \vertexset$, we write $v.\mathit{subconf}$ to denote the subconfiguration associated with $v$, e.g., $\internode_\subconfig.\mathit{subconf}$ is $\subconfig$. To ensure a one-to-one correspondence between the portions in $\sourceassignment$ and the unit path flows, we maintain an array $\mathit{used}$ that keeps track of which unit path flows have been associated with a portion so far.
	
	\begin{algorithm}
		\caption{Construct a $\partition$-conforming reconfiguration sequence given $\begin{bmatrix} \hat{\flow} & \hat{\demand} \end{bmatrix}$}\label{alg:create_partition_sequence}
		\begin{algorithmic}[1]
			\Require A feasible solution $\begin{bmatrix} \hat{\flow} & \hat{\demand} \end{bmatrix}$ to $\partitionILP$
			\Ensure A sequence $\areconfsequence$ of assignments of $\partition$ for capacity $\capacity$
			\State Construct ${\sourceassignment}$ as in Lemma \ref{lem:create_partition}
			\State $\areconfsequence \gets \sequence{\sourceassignment}$\label{line:firstassign}
			\State $\hat{\flow}_1, \ldots, \hat{\flow}_{\abs{\hat{\flow}}} \gets$ the $\abs{\hat{\flow}}$ unit path flows from $\hat{\flow}$ given by Lemma \ref{lem:flow_decomp}
			\State $\mathit{used} \gets $ an array initialized with $\mathit{used}[1] = \cdots = \mathit{used}[\abs{\hat{\flow}}] = \mathrm{False}$\label{line:usedinit}
			\For{$i = 1, \ldots, \abs{\sourceassignment}$}\label{line:iloop}
			\State $\assignment \gets$ the last assignment in $\areconfsequence$
			\State $\subconfig \gets \portion{\assignment}{i}$ \label{line:getAsub}
			\State $j \gets$ an index such that $\mathit{used}[j]$ is False and $\hat{\flow}_j(\innode_\subconfig\internode_\subconfig)=1$\label{line:choosej}
			\State $\mathbf{S} \gets$ the sequence of vertices in the underlying path of $\hat{\flow}_j$ {\color{gray}\Comment{$\mathbf{S}[1]=\innode_\subconfig$}\label{line:defineS}}
			\For{$k = 3,\ldots,|\mathbf{S}|-1$} \label{line:kloop}
			\State $\assignment' \gets \assignment$ but with the $i$th portion replaced with $\mathbf{S}[k].\mathit{subconf}$\label{line:computeA'}
			\State $\areconfsequence$.append($\assignment'$)\label{line:append}
			\State $\assignment \gets \assignment'$
			\EndFor
			\State $\mathit{used}[j] \gets$ True
			\EndFor
			\State \Return $\areconfsequence$
		\end{algorithmic}
	\end{algorithm}
	
	\begin{remark}\label{rem:paths}
		Recall from the definition of $\edgeset$ that there is no edge that has one endpoint in $\inputnodes$ and one endpoint in $\outputnodes$. In particular, each $(\inputnodes,\outputnodes)$-path in $G$ is of the form $(\innode_\subconfig,\internode_\subconfig,\ldots,\internode_\subconfigtwo,\outnode_\subconfigtwo)$ for some $\subconfig,\subconfigtwo \in \bsubs$. This means that each $(\inputnodes,\outputnodes)$-path has length at least two, which implies that $\abs{\mathbf{S}} \geq 3$ in Algorithm \ref{alg:create_partition_sequence}. The $(\inputnodes,\outputnodes)$-paths of length exactly two are those of the form $(\innode_\subconfig,\internode_\subconfig,\outnode_\subconfig)$ for $\subconfig \in \bsubs$.
	\end{remark}
	
	\begin{remark}\label{rem:exclude}
		At line \ref{line:kloop}, we start the loop counter at 3 so that we exclude the first two nodes of $\mathbf{S}$ when extending $\areconfsequence$. This is because the first two nodes of $\mathbf{S}$ are $\mathbf{S}[1]=\innode_\subconfig$ and $\mathbf{S}[2]=\internode_\subconfig$, where subconfiguration $\subconfig$ is the same as $\portion{\assignment}{i}$. In particular, the edge $\innode_\subconfig\internode_\subconfig$ has two endpoints that correspond to the same subconfiguration, so we don't use it to define reconfiguration step. Similarly, we end the loop counter at $\abs{\mathbf{S}}-1$ so that we exclude the final node $\mathbf{S}[\abs{\mathbf{S}}]$, which is $\outnode_\subconfigtwo$ for some $\subconfigtwo \in \bsubs$, when extending $\areconfsequence$. This is because $\mathbf{S}[\abs{\mathbf{S}}-1]=\internode_\subconfigtwo$, which means that the final two nodes of $\mathbf{S}$ have the same corresponding subconfiguration $\subconfigtwo$, and so the edge $\internode_\subconfigtwo\outnode_\subconfigtwo$ is not used to define a reconfiguration step. All other edges in the underlying $(\inputnodes,\outputnodes)$-path of $\hat{\flow}_j$ are of the form $\internode_\subconfigtwo\internode_\subconfigthree$ with $\subconfigtwo \neq \subconfigthree$, so, for each $k \in \{3,\ldots,\abs{\mathbf{S}}-1\}$, the two nodes $\mathbf{S}[k-1]$ and $\mathbf{S}[k]$ correspond to two adjacent nodes $\internode_\subconfigtwo$,$\internode_\subconfigthree$ with $\subconfigtwo \neq \subconfigthree$ in $G$.
	\end{remark}
	
	\begin{remark}\label{rem:skiploop}
		When $\abs{\mathbf{S}} = 3$, the \textbf{for} loop at line \ref{line:kloop} is not entered since the range $3,\ldots,\abs{\mathbf{S}}-1$ is empty, and hence no assignment is appended to $\areconfsequence$ for the current value of $i$. In this case, the $i$th portion will not change in our reconfiguration sequence $\areconfsequence$, i.e., $\portion{\sourceassignment}{i} = \portion{\targetassignment}{i}$.
	\end{remark}

	Before proving the correctness of Algorithm \ref{alg:create_partition_sequence}, we first show in Lemma \ref{lem:used_index_j} that the operation on line \ref{line:choosej} of Algorithm \ref{alg:create_partition_sequence} is well-defined, i.e., there is always a value that can be chosen for $j$.
	
	\begin{lemma}\label{lem:used_index_j}
		At each execution of line \ref{line:choosej} in Algorithm \ref{alg:create_partition_sequence}, there exists a value $j$ such that $\mathit{used}[j]$ is False and $\hat{\flow}_j(\innode_\subconfig\internode_\subconfig)=1$.
	\end{lemma}
	
	\begin{proof}
		It is sufficient to prove the following loop invariant: immediately before each iteration of the \textbf{for} loop at line \ref{line:iloop}, for each $\subconfig \in \bsubs$, the number of remaining loop iterations with $\portion{\assignment}{i}=\subconfig$ is equal to the number of values of $j$ such that $\mathit{used}[j]$ is False and $\hat{\flow}_j(\innode_\subconfig\internode_\subconfig)=1$.
		
		We first prove that the invariant is true before the first iteration of the \textbf{for} loop at line \ref{line:iloop}. For any fixed $\subconfig \in \bsubs$, we know from Lemma \ref{lem:create_partition} that $\sourceassignment$ has exactly $\abs{\hat{\demand}(\innode_\subconfig)}$ portions that are equal to $\subconfig$, so there are exactly $\abs{\hat{\demand}(\innode_\subconfig)}$ different values for $i$ such that $\portion{\assignment}{i}=\subconfig$ in the \textbf{for} loop at line \ref{line:iloop}. Moreover, for any fixed $\subconfig \in \bsubs$, we know from Lemma \ref{lem:number_of_unit_path_flows} that there are exactly $\abs{\hat{\demand}(\innode_{\subconfig})}$ different unit path flows from $\hat{\flow}_1,\ldots,\hat{\flow}_{\abs{\hat{\flow}}}$ whose underlying path has node $\innode_{\subconfig} \in X$ as its source, which means that there are exactly $\abs{\hat{\demand}(\innode_{\subconfig})}$ values of $j$ such that $\hat{\flow}_j(\innode_\subconfig\internode_\subconfig)=1$. Thus, for each fixed $\subconfig \in \bsubs$, before the first iteration of the \textbf{for} loop at line \ref{line:iloop}, the number of $i$ values with $\portion{\assignment}{i}=\subconfig$ is equal to the number of values of $j$ such that $\hat{\flow}_j(\innode_\subconfig\internode_\subconfig)=1$, and by line \ref{line:usedinit}, we know that all such values of $j$ have $\mathit{used}[j]$ set to False.
		
		At the end of any iteration of the \textbf{for} loop at line \ref{line:iloop}, we set $\mathit{used}[j]$ to True for exactly one value of $j$ such that $\hat{\flow}_j(\innode_\subconfig\internode_\subconfig)=1$, where $\subconfig = \portion{\assignment}{i}$. So, at the end of the loop iteration, the number of remaining loop iterations with $\portion{\assignment}{i}=\subconfig$ goes down by one, and the number of values of $j$ such that $\mathit{used}[j]$ is False and $\hat{\flow}_j(\innode_\subconfig\internode_\subconfig)=1$ also goes down by one, which implies that the loop invariant holds at the start of the next loop iteration.
	\end{proof}

	Next, in Lemma~\ref{lem:properties_of_A_1}, we show that the output $\areconfsequence$ of Algorithm \ref{alg:create_partition_sequence} is a $\partition$-conforming reconfiguration sequence for capacity $\capacity$ from $\source$ to $\target$.
	
	\begin{lemma}\label{lem:properties_of_A_1}
		The sequence $\areconfsequence$ returned by Algorithm \ref{alg:create_partition_sequence} satisfies:
		\begin{enumerate}[label=(\arabic*)]
			\item Each pair of consecutive assignments in $\areconfsequence$ are adjacent assignments, and, the portions that differ across the two assignments both belong to $\bsubs$\label{A:prop1}
			\item For each $\assignment \in \areconfsequence$, $\assignment$ is an assignment of $\partition$ for capacity $\capacity$.\label{A:prop2}
			\item The first assignment and last assignment in $\areconfsequence$ are consistent with $\source$ and $\target$, respectively.\label{A:prop3}
		\end{enumerate}
	\end{lemma}
	
	\begin{proof}
		To prove \ref{A:prop1}, we consider an arbitrary pair of consecutive assignments $\assignment, \assignment'$ in $\areconfsequence$, and consider the moment in the execution when $\assignment'$ was appended to $\areconfsequence$. Since $\assignment'$ is not the first assignment in $\areconfsequence$, it follows that $\assignment'$ was appended to $\areconfsequence$ at line \ref{line:append}. By line \ref{line:computeA'}, assignment $\assignment'$ is the same as $\assignment$ except that the $i$th portion of $\assignment$ is replaced with $\mathbf{S}[k].\mathit{subconf}$. So, to prove \ref{A:prop1}, it suffices to show that, when line \ref{line:computeA'} is executed, $\portion{\assignment}{i}$ and $\mathbf{S}[k].\mathit{subconf}$ are adjacent subconfigurations that both belong $\bsubs$. There are two cases to consider based on the value of the loop counter $k$ when line \ref{line:computeA'} is executed:
		
		\begin{description}
			\item[Case $k > 3$]: 
			In this case, note that $\assignment$ was appended to $\areconfsequence$ in the previous iteration of the \textbf{for} loop at line \ref{line:kloop}. In particular, by line \ref{line:computeA'}, the $i$th portion of $\assignment$ is equal to $\mathbf{S}[k-1].\mathit{subconf}$. By the definition of $\mathbf{S}$ at line \ref{line:defineS}, and by Remark \ref{rem:exclude}, we know that $\mathbf{S}[k-1]$ and $\mathbf{S}[k]$ are two endpoints of an edge $\internode_\subconfigtwo\internode_\subconfigthree$ with $\subconfigtwo \neq \subconfigthree$ in $G$. By the construction of $G$, the fact that $\internode_\subconfigtwo$ and $\internode_\subconfigthree$ are adjacent in $G$ means that $\subconfigtwo$ and $\subconfigthree$ are adjacent subconfigurations from $\bsubs$, and this concludes the proof that $\portion{\assignment}{i}$ and $\mathbf{S}[k].\mathit{subconf}$ are adjacent subconfigurations that both belong to $\bsubs$.
			
			\item[Case $k = 3$]: 
			In this case, by line \ref{line:getAsub}, the $i$th portion of $\assignment$ is equal to $\subconfig$. Moreover, the value of $j$ was chosen at line \ref{line:choosej} such that $\hat{\flow}_j(\innode_\subconfig\internode_\subconfig)=1$, which implies that the first two nodes of the underlying $(\inputnodes,\outputnodes)$-path of $\hat{\flow}_j$ are $\innode_\subconfig$ and $\internode_\subconfig$. By the definition of $\mathbf{S}$ at line \ref{line:defineS}, and by Remarks \ref{rem:exclude} and \ref{rem:skiploop}, we know that $\mathbf{S}[2]=\internode_\subconfig$ and $\mathbf{S}[3]$ are two endpoints of an edge $\internode_\subconfig\internode_\subconfigthree$ with $\subconfig \neq \subconfigthree$ in $G$. By the construction of $G$, the fact that $\internode_\subconfig$ and $\internode_\subconfigthree$ are adjacent in $G$ means that $\subconfig$ and $\subconfigthree$ are adjacent subconfigurations from $\bsubs$, and this concludes the proof that $\portion{\assignment}{i}$ and $\mathbf{S}[3].\mathit{subconf}$ are adjacent subconfigurations that both belong to $\bsubs$.
		\end{description}

		To prove \ref{A:prop2}, we use induction. For the base of the induction, note that $\sourceassignment$ is an assignment of $\partition$ for capacity $\capacity$ by Lemma \ref{lem:create_partition}. As induction hypothesis, assume that for some $\ell \geq 1$, the $\ell$th assignment in $\areconfsequence$, denoted by $\assignment$, is an assignment of $\partition$ for capacity $\capacity$. For the induction step, consider the $(\ell+1)$st assignment in $\areconfsequence$, denoted by $\assignment'$. By \ref{A:prop1}, assignments $\assignment$ and $\assignment'$ are adjacent, and, the portions that differ across the two assignments both belong to $\bsubs$. By the definition of adjacent assignments, there is exactly one index $i$ such that $\portion{\assignment}{i}$ and $\portion{\assignment'}{i}$ differ, and these two portions are adjacent subconfigurations. However, by the definition of adjacent subconfigurations, one of the portions can be formed from the other by the move of a single item, which implies that the two portions have the same underlying multiset of items and the same number of \bins. Further, the fact that $\portion{\assignment'}{i}$ is in $\bsubs$ means that all \bins in $\portion{\assignment'}{i}$ have capacity $\capacity$. It follows that $\assignment'$ is also an assignment of $\partition$ for capacity $\capacity$, which completes the induction step.
		
		To prove \ref{A:prop3}, note that the first assignment in $\areconfsequence$ is $\sourceassignment$ by line \ref{line:firstassign}, and that $\sourceassignment$ is consistent with $\source$ by Lemma \ref{lem:create_partition}. Next, we want show that the last assignment in $\areconfsequence$, denoted by $\targetassignment$, is consistent with the target configuration $\target$. At a high level, we will count the number of times that each subconfiguration $\subconfigtwo \in \bsubs$ appears in $\targetassignment$, and then use this to argue that each \bin type appears the same number of times in $\targetassignment$ as in $\target$.

		Consider an arbitrary $i \in 1,\ldots,\abs{\sourceassignment}$, and consider the value assigned to $j$ at line \ref{line:choosej} during the $i$th iteration of the \textbf{for} loop at line \ref{line:iloop}. Recall that, at line \ref{line:defineS}, the variable $\mathbf{S}$ is set to be the sequence of vertices on the underlying $(\inputnodes,\outputnodes)$-path of the unit path flow $\hat{\flow}_j$, and by Remark \ref{rem:paths}, we know that $\mathbf{S}[\abs{\mathbf{S}}] = \outnode_\subconfigtwo$ and $\mathbf{S}[\abs{\mathbf{S}}-1] = \internode_\subconfigtwo$ for some $\subconfigtwo \in \bsubs$. We claim that:
		\begin{enumerate}[label=(\roman*)]
			\item at the end of the $i$th iteration of the \textbf{for} loop at line \ref{line:iloop}, the last assignment in $\areconfsequence$, denoted by $\assignment_i$, has $\portion{\assignment_i}{i} = \subconfigtwo$, and,\label{A:prop3i}
			\item $\portion{\targetassignment}{i} = \subconfigtwo$.\label{A:prop3ii}
		\end{enumerate}
		
		To prove \ref{A:prop3i}, there are two cases to consider depending on whether or not line \ref{line:append} is executed during the $i$th iteration. By Remarks \ref{rem:paths} and \ref{rem:exclude}, we know that $\abs{\mathbf{S}} \geq 3$ and that line \ref{line:append} is executed during the $i$th iteration if and only if $\abs{\mathbf{S}} > 3$.
		\begin{description}
			\item[Case $\abs{\mathbf{S}} > 3$] : Consider the final iteration of the \textbf{for} loop at line \ref{line:kloop}, and consider the assignment $\assignment'$ appended to $\areconfsequence$ at line \ref{line:append}. By line \ref{line:computeA'}, the $i$th portion of $\assignment'$ is equal to $\mathbf{S}[\abs{\mathbf{S}}-1].\mathit{subconf}$. However, since $\mathbf{S}[\abs{\mathbf{S}}-1] = \internode_\subconfigtwo$, we conclude that the $i$th portion of $\assignment'$ is equal to $\subconfigtwo$. Since no further assignments are appended to $\areconfsequence$ during the $i$th iteration of the \textbf{for} loop at line \ref{line:iloop}, we conclude that $\portion{\assignment_i}{i} = \subconfigtwo$, as desired.
			\item[Case $\abs{\mathbf{S}} = 3$] : Since line \ref{line:append} is not executed during the $i$th iteration, it follows that no assignments are appended to $\areconfsequence$ during this iteration. Thus, by line \ref{line:getAsub}, we know that $\portion{\assignment_i}{i} = \subconfig$. However, by Remark \ref{rem:paths}, the fact that $\abs{\mathbf{S}} = 3$ means that the underlying $(\inputnodes,\outputnodes)$-path of the unit path flow $\hat{\flow}_j$ is $(\innode_\subconfig,\internode_\subconfig,\outnode_\subconfig)$. Since $\subconfigtwo$ was defined as the subconfiguration in the subscripts of the last two nodes of the underlying path, it follows that $\subconfigtwo = \subconfig$, and we conclude that $\portion{\assignment_i}{i} = \subconfigtwo$, as desired.
		\end{description}
		
		To prove \ref{A:prop3ii}, we note that, by \ref{A:prop3i}, at the end of the $i$th iteration of the \textbf{for} loop at line \ref{line:iloop}, the last assignment $\assignment_i$ in $\areconfsequence$ has $\portion{\assignment_i}{i} = \subconfigtwo$. Any assignment that is subsequently appended to $\areconfsequence$ will be appended at line \ref{line:append} during the $(i+c)$th iteration of the \textbf{for} loop at line \ref{line:iloop} for some $c > 0$. Note that each time that line \ref{line:computeA'} is executed during the $(i+c)$th iteration of the \textbf{for} loop at line \ref{line:iloop}, the only portion that is different across $\assignment$ and $\assignment'$ is the $(i+c)$th portion, so, in particular, the $i$th portion is the same in $\assignment$ and $\assignment'$. It follows that, in each assignment $\assignment'$ that is appended to $\areconfsequence$ after the end of the $i$th iteration of the \textbf{for} loop at line \ref{line:iloop}, we have $\portion{\assignment'}{i} = \portion{\assignment_i}{i}$, which, in particular, implies that $\portion{\targetassignment}{i} = \subconfigtwo$, as desired.
		
		Since we proved that \ref{A:prop3ii} applies to all values of $i \in \{1,\ldots,\abs{\sourceassignment}\}$, it follows that, for each fixed $\subconfigtwo \in \bsubs$, the number of different values of $i$ such that $\portion{\targetassignment}{i} = \subconfigtwo$ is equal to the number of different unit path flows $\hat{\flow}_j$ whose underlying path has node $\outnode_\subconfigtwo$ as its sink. By Lemma \ref{lem:number_of_unit_path_flows}, the latter quantity is equal to $\abs{\hat{\demand}(\outnode_\subconfigtwo)}$. In summary, we have shown that the number of times that each subconfiguration $\subconfigtwo \in \bsubs$ appears as a portion in $\targetassignment$ is exactly equal to $\abs{\hat{\demand}(\outnode_\subconfigtwo)}$. Next, to count the number of times a specific \bin $\binb$ appears in $\targetassignment$, we can take a sum over all portions in $\targetassignment$ and add up the number of times that \bin $\binb$ appears in that portion, i.e., $\sum_{\assignmentpart \in \targetassignment} \binmult{B}{\assignmentpart}$. By the above argument, we can re-write this as
		\begin{align*}
			\sum_{\assignmentpart \in \targetassignment} \binmult{B}{\assignmentpart} &=
			\sum_{\subconfigtwo \in \bsubs} \binmult{B}{\subconfigtwo} \cdot \abs{\hat{\demand}(\outnode_\subconfigtwo)}\\
			&= \binmult{B}{\target} \tag{by Constraints \ref{con:five} and \ref{con:seven}}
		\end{align*}
		This proves that each \bin $\binb$ appears the same number of times in $\targetassignment$ as in $\target$, which means that $\targetassignment$ is consistent with $\target$, as desired.
	\end{proof}
	
	Finally, we can prove the forward direction of Lemma~\ref{lem:partition_reduction} (restated below as Lemma~\ref{lem:partition_reduction_forward}) by showing that Algorithm \ref{alg:create_partition_sequence} produces a $\partition$-conforming reconfiguration sequence for capacity $\capacity$ from $\source$ to $\target$. 
	
	\partitionreductionforward*
	
	\begin{proof}
		Given a feasible solution $\begin{bmatrix} \hat{\flow} & \hat{\demand} \end{bmatrix}$ to $\partitionILP$, we use Lemma~\ref{lem:create_partition} to form a $\partitionbound$-bounded partition $\partition$ and an assignment $\sourceassignment$ of $\partition$ for capacity $\capacity$ that is consistent with $\source$. By Lemma \ref{lem:used_index_j}, Algorithm \ref{alg:create_partition_sequence} is well-defined, and executing this algorithm on input $\begin{bmatrix} \hat{\flow} & \hat{\demand} \end{bmatrix}$ constructs a sequence $\areconfsequence$ of assignments that, by Lemma \ref{lem:properties_of_A_1}, is a $\partition$-conforming reconfiguration sequence for capacity $\capacity$ from $\source$ to $\target$. This proves that the $\brepackingk$ instance $(\source, \target)$ is reconfigurable.
	\end{proof}
	
	\subsection{Proof of Lemma \ref{lem:partition_reduction} - backward direction}\label{sec:appendix_partition_backward_new}
	
	The goal of this section is to prove the following lemma.
	\begin{restatable}{lemma}{partitionreductionbackward}\label{lem:partition_reduction_backward}
		Consider any instance $(\source, \target)$ for \textsc{$\partitionbound$-Repacking-$\capacity$}. If $(\source, \target)$ is a yes-instance, then the corresponding $\partitionILP$ has a feasible solution.
	\end{restatable}
	
	Suppose the instance $(\source, \target)$ of \textsc{$\partitionbound$-Repacking-$\capacity$} is a yes-instance. By definition, this means that, for some $\boundedpartition$ $\partition$ of $\universe$, we have a $\partition$-conforming reconfiguration sequence for capacity $\capacity$ from $\source$ to $\target$. To prove the lemma, we devise an algorithm that converts this $\partition$-conforming reconfiguration sequence to a feasible solution of $\partitionILP$.
	
	Recall that, in any $\partition$-conforming reconfiguration sequence, each move occurs within a single part of the partition, i.e., for any two consecutive assignments $\assignment$ and $\assignment'$, there is exactly one value for $i$ such that $\portion{\assignment}{i} \neq \portion{\assignment'}{i}$. We say that a $\partition$-conforming reconfiguration sequence is in \emph{sorted order} if, for each $i,j \in \{1,\ldots,\abs{\partition}\}$ such that $i < j$, all moves involving an item in the $i$th portion of an assignment occur earlier than any move involving an item in the $j$th portion of an assignment. The assignments of any $\partition$-conforming reconfiguration sequence can be re-arranged to obtain a $\partition$-conforming reconfiguration sequence that is in sorted order.
	
	At a high level, we take a $\partition$-conforming reconfiguration sequence that is in sorted order, and we construct a feasible solution to $\partitionILP$ as follows. First, we set the demands of each origin node $\innode_\subconfig$ of $\graph$ by adding a demand of $-1$ for each occurrence of $\subconfig$ in the source assignment $\sourceassignment$. Each such origin node $\innode_\subconfig$ has a single outgoing edge $\innode_\subconfig\internode_\subconfig$, so we set the flow value on that edge to match the absolute value of the demand on $\innode_\subconfig$. Similarly, we set the demands of each destination node $\outnode_\subconfig$ of $\graph$ by adding a demand of 1 for each occurrence of $\subconfig$ in the target assignment $\targetassignment$. Each such destination node $\outnode_\subconfig$ has a single incoming edge $\internode_\subconfig\outnode_\subconfig$, so we set the flow value on that edge to match the absolute value of the demand on $\outnode_\subconfig$. We then construct the flow $\flow$ from the reconfiguration sequence by tracking the changes between each pair of consecutive assignments. More specifically, if the reconfiguration sequence changes a portion from subconfiguration $\subconfig$ to subconfiguration $\subconfigtwo$,  we increment the flow on the edge  $\internode_{\subconfig}\internode_{\subconfigtwo}$. Algorithm \ref{alg:reconstruct_sequence} gives the pseudocode for the procedure described above.
	
	\begin{algorithm}
		\caption{Construct a feasible solution to $\partitionILP$}\label{alg:reconstruct_sequence}
		\begin{algorithmic}[1]
			\Require A $\partition$-conforming reconfiguration sequence $\areconfsequence$ for capacity $\capacity$ from $\source$ to $\target$ in sorted order, with $\areconfsequence[1] = \sourceassignment$
			\Ensure A feasible solution $\begin{bmatrix} \flow & \demand \end{bmatrix}$ to $\partitionILP$
			\State $\sourceassignment \gets \areconfsequence[1]$
			\State $\targetassignment \gets \areconfsequence[\abs{\areconfsequence}]$
			\State $\demand(\innode_{\subconfig}) \gets - \binmult{\subconfig}{\sourceassignment}$ for all $\subconfig \in \bsubs$\label{line:setXdemands}
			\State $\flow(\innode_{\subconfig}\internode_\subconfig) \gets \binmult{\subconfig}{\sourceassignment}$ for all $\subconfig \in \bsubs$\label{line:setXflows}
			\State $\demand(\outnode_{\subconfig}) \gets \binmult{\subconfig}{\targetassignment}$ for all $\subconfig \in \bsubs$\label{line:setZdemands}
			\State $\flow(\internode_\subconfig\outnode_\subconfig) \gets \binmult{\subconfig}{\targetassignment}$ for all $\subconfig \in \bsubs$\label{line:setZflows}
			\State $\flow(\internode_\subconfig\internode_\subconfigtwo) \gets 0$ for all $\subconfig,\subconfigtwo \in \bsubs$\label{line:initflow}
			\For{$i = 2,\ldots,\abs{\areconfsequence}$}\label{line:forloop}
			\State $\assignment \gets \areconfsequence[i-1]$
			\State $\assignmenttwo \gets \areconfsequence[i]$
			\State $j \gets$ the index of the portion at which the assignments $\assignment$ and $\assignmenttwo$ differ
			\State $\subconfig \gets \portion{\assignment}{j}$\label{line:getsubs}
			\State $\subconfigtwo \gets \portion{\assignmenttwo}{j}$
			\State $\flow(\internode_{\subconfig}\internode_{\subconfigtwo}) \gets \flow(\internode_{\subconfig}\internode_{\subconfigtwo}) + 1$\label{line:incflow}
			\EndFor
			\State \Return $\begin{bmatrix} \flow & \demand \end{bmatrix}$
		\end{algorithmic}
	\end{algorithm}
	
	\begin{remark}
		In lines \ref{line:getsubs}-\ref{line:incflow}, we are using the fact that each portion of an assignment in $\areconfsequence$ can be viewed as a subconfiguration in $\bsubs$: each portion has at most $\partitionbound$ \bins and each \bin has volume at most $\capacity$ since $\partition$ is $\partitionbound$-bounded and $\areconfsequence$ is a $\partition$-conforming reconfiguration sequence for capacity $\capacity$.
	\end{remark}
	
	The main idea is that, for each fixed $j \in \{1,\ldots,\abs{\partition}\}$, there is a consecutive subsequence of assignments that changes only the $j$th portion. We can view such a sequence of assignments as a sequence of steps, where each step is from some subconfiguration $\subconfig$ to some subconfiguration $\subconfigtwo$, and these two subconfigurations differ only in their $j$th portion. For each such step, we increase a unit of flow on the edge $\internode_\subconfig\internode_\subconfigtwo$ (line $14$) to reflect this change. 
	
	Our eventual goal is to prove that the output $\begin{bmatrix} \flow & \demand \end{bmatrix}$ of Algorithm \ref{alg:reconstruct_sequence} is a feasible solution to $\partitionILP$. First, we make an observation that will eventually help us prove that $\flow$ satisfies the conservation constraint. In particular, we focus on the values of $\flow$ on all edges of the form $\internode_\subconfig\internode_\subconfigtwo$, i.e., the edges whose flow values are set via line \ref{line:incflow} in iterations of the \textbf{for} loop. These are exactly the edges from $\edgeset$ induced by the vertex set $Y$ in $G$, and we denote this set of edges by $E_Y$. For each $\subconfigthree \in \bsubs$, denote by $\inedges_Y(\internode_\subconfigthree)$ the set of incoming edges to vertex $\internode_\subconfigthree$ that belong to $E_Y$, and denote by $\outedges_Y(\internode_\subconfigthree)$ the set of outgoing edges from vertex $\internode_\subconfigthree$ that belong to $E_Y$.
	
	\begin{lemma}\label{lem:algo_flow_bin_mult}
		Consider the output $\begin{bmatrix} \flow & \demand \end{bmatrix}$ of Algorithm \ref{alg:reconstruct_sequence}. For all $\subconfigthree \in \bsubs$, we have $\flow(\inedges_Y(\internode_\subconfigthree)) - \flow(\outedges_Y(\internode_\subconfigthree)) = \binmult{\subconfigthree}{\targetassignment} - \binmult{\subconfigthree}{\sourceassignment}$.
	\end{lemma}
	
	\begin{proof}
		Since $\areconfsequence$ is a $\partition$-conforming reconfiguration sequence, each pair of consecutive assignments in $\areconfsequence$ differ in exactly one portion. Further, since $\areconfsequence$ is given in sorted ordering, all assignments that change a particular portion form a consecutive subsequence of $\areconfsequence$. For each $j \in 1,\ldots,\abs{\partition}$, we denote by $\areconfsequence_j$ the consecutive subsequence of assignments in $\areconfsequence$ such that each of the assignments in $\areconfsequence_j$ differ from the previous assignment in $\areconfsequence$ at the $j$th portion. In what follows, we will consider, for a fixed value of $j$, all iterations of the \textbf{for} loop at line \ref{line:forloop} in which the pair of assignments $\assignment,\assignmenttwo$ differ at portion $j$. In other words, we are considering all iterations of the \textbf{for} loop at line \ref{line:forloop} for which $\areconfsequence[i] \in \areconfsequence_j$ (which, by definition, means that $\areconfsequence[i-1]$ and $\areconfsequence[i]$ differ only in their $j$th portion), and we call these the \emph{$j$-iterations}. For any edge $e \in E_Y$, denote by $f_j(e)$ the value that would be obtained by initializing $f(e) \gets 0$, then performing only the $j$-iterations, and then taking the resulting value of $f(e)$. In other words, $f_j(e)$ represents how much flow was added to edge $e$ during the reconfiguration of the $j$th part of the partition. Then, for each edge $e \in E_Y$, we have $f(e) = \sum_{j=1}^{\abs{\partition}} f_j(e)$. 
		
		Consider an arbitrary fixed $\subconfigthree \in \bsubs$. We start by considering an arbitrary fixed $j \in 1,\ldots,\abs{\partition}$, and computing the value of $\flow_j(\inedges_Y(\internode_\subconfigthree)) - \flow_j(\outedges_Y(\internode_\subconfigthree))$. Then, the desired result will later be obtained by summing up this expression over all values of $j \in \{1,\ldots,\abs{\partition}\}$.
		
		We investigate various cases based on the length of $\areconfsequence_j$. We define a variable $\chi_{{\cal S},j}$ that indicates whether or not $\portion{\sourceassignment}{j} = \subconfigthree$, i.e., $\chi_{{\cal S},j}$ is 1 if $\portion{\sourceassignment}{j} = \subconfigthree$, and otherwise $\chi_{{\cal S},j}$ is 0. Similarly, we define a variable $\chi_{{\cal T},j}$ that is 1 if $\portion{\targetassignment}{j} = \subconfigthree$, and otherwise $\chi_{{\cal T},j}$ is 0. In all cases, we will prove that the value of $\flow_j(\inedges_Y(\internode_\subconfigthree)) - \flow_j(\outedges_Y(\internode_\subconfigthree))$ is equal to $\chi_{{\cal T},j} - \chi_{{\cal S},j}$, and this fact will be used later to obtain the desired result.
		\begin{enumerate}[label=Case \arabic*:,leftmargin=2cm]
			\item $\abs{\areconfsequence_j}=0$.\\
			In this case, we know that $\portion{\sourceassignment}{j} = \portion{\targetassignment}{j}$ and that there are no values of $i$ such that $\areconfsequence[i-1]$ and $\areconfsequence[i]$ differ in their $j$th portion. It follows that there are no $j$-iterations, so we have $\flow_j(e) = 0$ for all $e \in E_Y$, which implies that $\flow_j(\inedges_Y(\internode_\subconfigthree)) - \flow_j(\outedges_Y(\internode_\subconfigthree)) = 0$. Since $\portion{\sourceassignment}{j} = \portion{\targetassignment}{j}$, it follows that $\chi_{{\cal T},j} - \chi_{{\cal S},j} = 0$ as well. 
			\item $\abs{\areconfsequence_j}=1$.\\
			In this case, there is exactly one value of $i$ such that $\areconfsequence[i-1]$ and $\areconfsequence[i]$ differ in their $j$th portion. It follows that $\portion{\areconfsequence[i-1]}{j} = \portion{\sourceassignment}{j}$ and that $\portion{\areconfsequence[i]}{j} = \portion{\targetassignment}{j}$. 
			
			Consider the iteration of the \textbf{for} loop at line \ref{line:forloop} for the above-specified value of $i$. In this iteration, we have $\subconfig = \portion{\areconfsequence[i-1]}{j} = \portion{\sourceassignment}{j}$ and $\subconfigtwo = \portion{\areconfsequence[i]}{j} = \portion{\targetassignment}{j}$, where $\subconfig \neq \subconfigtwo$. Moreover, the value of $\flow(\internode_\subconfig\internode_\subconfigtwo)$ is incremented by one, which contributes 1 to the value of $\flow_j(\outedges_Y(\internode_\subconfig))$ and contributes 1 to the value of $\flow_j(\inedges_Y(\internode_\subconfigtwo))$. Therefore, considering all of the possible cases:
			\begin{itemize}
				\item If $\portion{\sourceassignment}{j} = \subconfigthree$ and $\portion{\targetassignment}{j} \neq \subconfigthree$, then this iteration contributes 1 to the value of $\flow_j(\outedges_Y(\internode_\subconfigthree))$ and contributes 0 to the value of $\flow_j(\inedges_Y(\internode_\subconfigthree))$. Thus, the value of $\flow_j(\inedges_Y(\internode_\subconfigthree)) - \flow_j(\outedges_Y(\internode_\subconfigthree))$ is $-1$. Moreover, $\chi_{{\cal T},j} = 0$ and $\chi_{{\cal S},j}=1$, so we conclude that $\chi_{{\cal T},j} - \chi_{{\cal S},j} = -1$ as well.
				
				\item If $\portion{\sourceassignment}{j} \neq \subconfigthree$ and $\portion{\targetassignment}{j} = \subconfigthree$, then this iteration contributes 0 to the value of $\flow_j(\outedges_Y(\internode_\subconfigthree))$ and contributes 1 to the value of $\flow_j(\inedges_Y(\internode_\subconfigthree))$. Thus, the value of $\flow_j(\inedges_Y(\internode_\subconfigthree)) - \flow_j(\outedges_Y(\internode_\subconfigthree))$ is $1$. Moreover, $\chi_{{\cal T},j} = 1$ and $\chi_{{\cal S},j}=0$, so we conclude that $\chi_{{\cal T},j} - \chi_{{\cal S},j} = 1$ as well.
				
				\item If $\portion{\sourceassignment}{j} \neq \subconfigthree$ and $\portion{\targetassignment}{j} \neq \subconfigthree$, then this iteration contributes 0 to the value of $\flow_j(\outedges_Y(\internode_\subconfigthree))$ and contributes 0 to the value of $\flow_j(\inedges_Y(\internode_\subconfigthree))$. Thus, the value of $\flow_j(\inedges_Y(\internode_\subconfigthree)) - \flow_j(\outedges_Y(\internode_\subconfigthree))$ is 0. Moreover, $\chi_{{\cal T},j} = 0$ and $\chi_{{\cal S},j}=0$, so we conclude that $\chi_{{\cal T},j} - \chi_{{\cal S},j} = 0$ as well.
			\end{itemize}
			\item $\abs{\areconfsequence_j}>1$\\
			In this case, consider any two consecutive $j$-iterations, i.e., any value $\alpha$ such that the $j$th portion differs between the two assignments $\areconfsequence[\alpha-1]$ and $\areconfsequence[\alpha]$ and between the two assignments $\areconfsequence[\alpha]$ and $\areconfsequence[\alpha+1]$. In the iteration of the \textbf{for} loop with $i = \alpha$, note that $\subconfig = \portion{\areconfsequence[\alpha-1]}{j}$ and $\subconfigtwo = \portion{\areconfsequence[\alpha]}{j}$, and that the value of $f(\internode_\subconfig\internode_\subconfigtwo)$ is incremented. Then, in the iteration of the \textbf{for} loop with $i = \alpha+1$, note that $\subconfig = \portion{\areconfsequence[\alpha]}{j}$ and $\subconfigtwo = \portion{\areconfsequence[\alpha+1]}{j}$, and that the value of $f(\internode_\subconfig\internode_\subconfigtwo)$ is incremented. Therefore, if we consider the node in $Y$ corresponding to the subconfiguration $\portion{\areconfsequence[\alpha]}{j}$, an incoming edge to this node is incremented during the loop iteration with $i = \alpha$, and an outgoing edge from this node is incremented during the loop iteration with $i = \alpha+1$. The effects of these two increments ``cancel each other out", i.e., more formally, letting $\subconfigfour = \portion{\areconfsequence[\alpha]}{j}$, we see that these two loop iterations leave the value of $\flow_j(\inedges_Y(\internode_\subconfigfour)) - \flow_j(\outedges_Y(\internode_\subconfigfour))$ unchanged.
			
			Essentially, the above argument shows that, during the $j$-iterations, the value of $\flow_j(\inedges_Y(\internode_\subconfigfour)) - \flow_j(\outedges_Y(\internode_\subconfigfour))$ doesn't change for each $\subconfigfour$ corresponding to the $j$th portion of the ``intermediate" assignments of $\areconfsequence_j$, i.e., for each $\subconfigfour \in \{ \portion{\areconfsequence_j[1]}{j}$, $\ldots$, $\portion{\areconfsequence_j[\abs{\areconfsequence_j}-1]}{j}\}$. However, during the first $j$-iteration, we observe that $\subconfig = \portion{\sourceassignment}{j}$ and that $\flow(\internode_\subconfig\internode_\subconfigtwo)$ gets incremented, and this contributes 1 to the value of $\flow_j(\outedges_Y(\internode_\subconfig))$. Similarly, during the final $j$-iteration, we observe that $\subconfigtwo = \portion{\targetassignment}{j}$ and that $\flow(\internode_\subconfig\internode_\subconfigtwo)$ gets incremented, and this contributes 1 to the value of $\flow_j(\inedges_Y(\internode_\subconfigtwo))$. Thus, we conclude:
			\begin{itemize}
				\item if $\portion{\sourceassignment}{j} = \portion{\targetassignment}{j} = \subconfigthree$,\\then $\flow_j(\inedges_Y(\internode_\subconfigthree)) - \flow_j(\outedges_Y(\internode_\subconfigthree)) = 1 - 1 = 0$. Moreover, $\chi_{{\cal T},j} = 1$ and $\chi_{{\cal S},j}=1$, so we conclude that $\chi_{{\cal T},j} - \chi_{{\cal S},j} = 0$ as well.
				\item if $\portion{\sourceassignment}{j} \neq \subconfigthree$ and $\portion{\targetassignment}{j} = \subconfigthree$, then\\ $\flow_j(\inedges_Y(\internode_\subconfigthree)) - \flow_j(\outedges_Y(\internode_\subconfigthree)) = 1 - 0 = 1$. Moreover, $\chi_{{\cal T},j} = 1$ and $\chi_{{\cal S},j}=0$, so we conclude that $\chi_{{\cal T},j} - \chi_{{\cal S},j} = 1$ as well.
				\item if $\portion{\sourceassignment}{j} = \subconfigthree$ and $\portion{\targetassignment}{j} \neq \subconfigthree$, then\\ $\flow_j(\inedges_Y(\internode_\subconfigthree)) - \flow_j(\outedges_Y(\internode_\subconfigthree)) = 0 - 1 = -1$. Moreover, $\chi_{{\cal T},j} = 0$ and $\chi_{{\cal S},j}=1$, so we conclude that $\chi_{{\cal T},j} - \chi_{{\cal S},j} = -1$ as well.
				\item if $\subconfigthree \not\in \{\portion{\sourceassignment}{j}, \portion{\targetassignment}{j}\}$, then\\ $\flow_j(\inedges_Y(\internode_\subconfigthree)) - \flow_j(\outedges_Y(\internode_\subconfigthree)) = 0 - 0 = 0$. Moreover, $\chi_{{\cal T},j} = 0$ and $\chi_{{\cal S},j}=0$, so we conclude that $\chi_{{\cal T},j} - \chi_{{\cal S},j} = 0$ as well.
			\end{itemize}
		\end{enumerate}
		
		In all cases, we showed that $\flow_j(\inedges_Y(\internode_\subconfigthree)) - \flow_j(\outedges_Y(\internode_\subconfigthree))$ is equal to $\chi_{{\cal T},j} - \chi_{{\cal S},j}$. Taking the sum over all values of $j \in \{1,\ldots,\abs{\partition}\}$, we conclude that
		\begin{align*}
			& \flow(\inedges_Y(\internode_\subconfigthree)) - \flow(\outedges_Y(\internode_\subconfigthree))\\
			=\ & \sum_{j=1}^{\abs{\partition}} \left(\flow_j(\inedges_Y(\internode_\subconfigthree)) - \flow_j(\outedges_Y(\internode_\subconfigthree))\right)\\
			=\ & \sum_{j=1}^{\abs{\partition}} \left(\chi_{{\cal T},j} - \chi_{{\cal S},j}\right)\\
			=\ & \left(\sum_{j=1}^{\abs{\partition}} \chi_{{\cal T},j}\right) - \left(\sum_{j=1}^{\abs{\partition}}\chi_{{\cal S},j}\right)\\
			=\ & \binmult{\subconfigthree}{\targetassignment} - \binmult{\subconfigthree}{\sourceassignment} 
		\end{align*}
	\end{proof}

	We are now ready to prove the main result of this section, which is restated below.
	
	\partitionreductionbackward*
	
	\begin{proof}
		Suppose that $(\source, \target)$ is a yes-instance. By definition, this means that there is a $\partitionbound$-bounded partition $\partition$ and a $\partition$-conforming reconfiguration sequence for capacity $\capacity$ from $\source$ to $\target$. We rearrange the reconfiguration sequence to produce a $\partition$-conforming reconfiguration sequence $\areconfsequence$ for capacity $\capacity$ from $\source$ to $\target$ that is in sorted order. Then, we apply Algorithm \ref{alg:reconstruct_sequence} to $\areconfsequence$. To prove the desired result, we prove below that the output $\begin{bmatrix} \flow & \demand \end{bmatrix}$ of Algorithm \ref{alg:reconstruct_sequence} is a feasible solution $\partitionILP$ by verifying that $\begin{bmatrix} \flow & \demand \end{bmatrix}$ satisfies each of the constraints set out in $\partitionILP$.
		
		\begin{description}
			\item[Constraint \ref{con:one}] For each fixed $\subconfig \in \bsubs$, by lines \ref{line:setXdemands} and \ref{line:setXflows} of Algorithm \ref{alg:reconstruct_sequence}, we observe that $0 = -\binmult{\subconfig}{\sourceassignment} + \binmult{\subconfig}{\sourceassignment} = -\flow(\innode_\subconfig\internode_\subconfig) - \demand(\innode_\subconfig)$. By the construction of $G$, we know that $\{\innode_\subconfig\internode_\subconfig\} = \outedges(\innode_\subconfig)$, so we conclude that $- \flow(\outedges(\innode_\subconfig)) - \demand(\innode_\subconfig)=0$, as required.
			
			\item[Constraint \ref{con:two}] Consider an arbitrary fixed $\subconfigthree \in \bsubs$. Recall from the construction of $G$ that the incoming edges to $\internode_\subconfigthree$ are exactly $\{\innode_\subconfigthree\internode_\subconfigthree\} \cup \{\internode_\subconfig\internode_\subconfigthree \mid (\subconfig \in \bsubs) \wedge (\subconfig \neq \subconfigthree)\}$. Thus, $\flow(\inedges(\internode_\subconfigthree)) = \flow(\innode_\subconfigthree\internode_\subconfigthree) + \flow(\inedges_Y(\internode_\subconfigthree))$. Moreover, recall from the construction of $G$ that the outgoing edges from $\internode_\subconfigthree$ are exactly $\{\internode_\subconfigthree\outnode_\subconfigthree\} \cup \{\internode_\subconfigthree\internode_\subconfig \mid (\subconfig \in \bsubs) \wedge (\subconfig \neq \subconfigthree)\}$. Thus, $\flow(\outedges(\internode_\subconfigthree)) = \flow(\internode_\subconfigthree\outnode_\subconfigthree) + \flow(\outedges_Y(\internode_\subconfigthree))$. Therefore, we have
			\begin{align*}
				& \flow(\inedges(\internode_\subconfigthree)) - \flow(\outedges(\internode_\subconfigthree)) \\
				=\ &(\flow(\innode_\subconfigthree\internode_\subconfigthree) + \flow(\inedges_Y(\internode_\subconfigthree))) - (\flow(\internode_\subconfigthree\outnode_\subconfigthree) + \flow(\outedges_Y(\internode_\subconfigthree)))\\
				=\ & \flow(\innode_\subconfigthree\internode_\subconfigthree) -\flow(\internode_\subconfigthree\outnode_\subconfigthree) + \flow(\inedges_Y(\internode_\subconfigthree)) - \flow(\outedges_Y(\internode_\subconfigthree))\\
				=\ & \flow(\innode_\subconfigthree\internode_\subconfigthree) -\flow(\internode_\subconfigthree\outnode_\subconfigthree) + \binmult{\subconfigthree}{\targetassignment} - \binmult{\subconfigthree}{\sourceassignment}\tag{by Lemma \ref{lem:algo_flow_bin_mult}} \\
				=\ & \binmult{\subconfigthree}{\sourceassignment} -\binmult{\subconfigthree}{\targetassignment} + \binmult{\subconfigthree}{\targetassignment} - \binmult{\subconfigthree}{\sourceassignment} \tag{by lines \ref{line:setXflows} and \ref{line:setZflows} of Alg. \ref{alg:reconstruct_sequence}} \\
				=\ & 0
			\end{align*}
			as required.

			\item[Constraint \ref{con:three}] 
			For each fixed $\subconfig \in \bsubs$, by lines \ref{line:setZdemands} and \ref{line:setZflows} of Algorithm \ref{alg:reconstruct_sequence}, we observe that $0 = \binmult{\subconfig}{\targetassignment} - \binmult{\subconfig}{\targetassignment} = \flow(\internode_\subconfig\outnode_\subconfig) - \demand(\outnode_\subconfig)$. By the construction of $G$, we know that $\{\internode_\subconfig\outnode_\subconfig\} = \inedges(\outnode_\subconfig)$, so we conclude that $\flow(\inedges(\outnode_\subconfig)) - \demand(\outnode_\subconfig)=0$, as required.
			
			\item[Constraint \ref{con:four}]
			Recall that $\areconfsequence$ is a $\partition$-conforming reconfiguration sequence for capacity $\capacity$ from $\source$ to $\target$, which implies that the first assignment $\sourceassignment$ of $\areconfsequence$ is consistent with $\source$. It follows that, for each $\binb \in \bintypesnoparameters$, we have $\binmult{\binb}{\source} = \binmult{\binb}{\sourceassignment}$. But, $\binmult{\binb}{\sourceassignment}$ can be computed as follows: for each subconfiguration $\subconfig \in \bsubs$, count how many times $\subconfig$ appears in $\sourceassignment$, and multiply by the number of times that $\binb$ appears in $\subconfig$, i.e., $\binmult{\binb}{\sourceassignment} = \sum_{\subconfig \in \bsubs} \binmult{B}{\subconfig} \cdot \binmult{\subconfig}{\sourceassignment}$. Thus, we have shown that
			\begin{align*}
				\binmult{B}{\source}  &= \binmult{\binb}{\sourceassignment}\\
				&= \sum_{\subconfig \in \bsubs} \binmult{B}{\subconfig} \cdot \binmult{\subconfig}{\sourceassignment}  \\
				&= -\sum_{\subconfig \in \bsubs} \binmult{B}{\subconfig} \cdot \demand(\innode_\subconfig)\tag{by line \ref{line:setXdemands} of Algorithm \ref{alg:reconstruct_sequence}}
			\end{align*}
			which implies that $-\binmult{B}{\source} = \sum_{\subconfig \in \bsubs} \binmult{B}{\subconfig} \cdot \demand(\innode_\subconfig)$, as required.
			
			\item[Constraint \ref{con:five}] 
			Recall that $\areconfsequence$ is a $\partition$-conforming reconfiguration sequence for capacity $\capacity$ from $\source$ to $\target$, which implies that the final assignment $\targetassignment$ of $\areconfsequence$ is consistent with $\target$. It follows that, for each $\binb \in \bintypesnoparameters$, we have $\binmult{\binb}{\target} = \binmult{\binb}{\targetassignment}$. But, $\binmult{\binb}{\targetassignment}$ can be computed as follows: for each subconfiguration $\subconfig \in \bsubs$, count how many times $\subconfig$ appears in $\targetassignment$, and multiply by the number of times that $\binb$ appears in $\subconfig$, i.e., $\binmult{\binb}{\targetassignment} = \sum_{\subconfig \in \bsubs} \binmult{B}{\subconfig} \cdot \binmult{\subconfig}{\targetassignment}$. Thus, we have shown that
			\begin{align*}
				\binmult{B}{\target}  &= \binmult{\binb}{\targetassignment}\\
				&= \sum_{\subconfig \in \bsubs} \binmult{B}{\subconfig} \cdot \binmult{\subconfig}{\targetassignment}  \\
				&= \sum_{\subconfig \in \bsubs} \binmult{B}{\subconfig} \cdot \demand(\outnode_\subconfig)\tag{by line \ref{line:setZdemands} of Algorithm \ref{alg:reconstruct_sequence}}
			\end{align*}
			as required.
			
			\item[Constraint \ref{con:six}] We check that the flow assigned to each edge is non-negative. Consider any fixed $\subconfig$. First, for any fixed $\subconfigtwo \in \bsubs$ with $\subconfigtwo \neq \subconfig$, the value of $\flow(\internode_\subconfig\internode_\subconfigtwo)$ is initialized to 0 at line \ref{line:initflow} of Algorithm \ref{alg:reconstruct_sequence}, and the only change to this value is an increment at line \ref{line:incflow}, so we conclude that $\flow(\internode_\subconfig\internode_\subconfigtwo) \geq 0$, as required. Next, the value of $\flow(\innode_\subconfig\innode_\subconfig)$ is set at line \ref{line:setXflows} of Algorithm \ref{alg:reconstruct_sequence}, and we see that $\flow(\innode_\subconfig\internode_\subconfig) = \binmult{\subconfig}{\sourceassignment} \geq 0$, as required. Finally, the value of $\flow(\internode_\subconfig\outnode_\subconfig)$ is set at line \ref{line:setZflows} of Algorithm \ref{alg:reconstruct_sequence}, and we see that $\flow(\internode_\subconfig\outnode_\subconfig) = \binmult{\subconfig}{\targetassignment} \geq 0$, as required.
			
			\item[Constraint \ref{con:seven}] Consider an arbitrary $\subconfig \in \bsubs$. First, at line \ref{line:setXdemands} of Algorithm \ref{alg:reconstruct_sequence}, we see that $\demand(\innode_\subconfig) = -\binmult{\subconfig}{\sourceassignment} \leq 0$, as required. Next, at line \ref{line:setZdemands} of Algorithm \ref{alg:reconstruct_sequence}, we see that $\demand(\outnode_\subconfig) = \binmult{\subconfig}{\targetassignment} \geq 0$, as required.
		\end{description}
		As all constraints have been satisfied, the vector $\begin{bmatrix} \flow & \demand \end{bmatrix}$ generated from Algorithm \ref{alg:reconstruct_sequence} is a feasible solution to $\partitionILP$, as desired.
	\end{proof}
	
	\subsection{Proof of Theorem~\ref{thm-partition}}\label{sec:appendix_thm}
	
	In order to analyze the time complexity of solving $\brepackingk$, we must specify how instances are encoded and determine their length. Recall that each instance of $\brepackingk$ consists of a source configuration $\source$ and target configuration $\target$. Each of these configurations in the instance is encoded as a list of ordered pairs: the first element of the pair is the \bin type, and the second element of the pair is the number of times that this \bin type appears in the configuration, i.e., its multiplicity. By assigning a unique integer to each \bin type, it follows that the first element of each pair can be encoded using $O(\log|\bintypesnoparameters|)$ bits. Next, let $\mathit{maxmult}(\source,\target)$ denote the maximum number of times that a \bin type appears in $\source$ or $\target$, i.e., the maximum of the second entries taken over all ordered pairs in $\source$ and $\target$. Then, the second entry of each pair can be encoded using $\log(\mathit{maxmult}(\source,\target))$ bits. Finally, determine how many ordered pairs define each of $\source$ and $\target$, and let $\mathit{maxtypes}(\source,\target)$ denote the maximum of these two values. For ease of notation, we will refer to $\mathit{maxmult}(\source,\target)$ as $\mathit{maxmult}$ and we will refer to $\mathit{maxtypes}(\source,\target)$ as $\mathit{maxtypes}$. We conclude that each configuration can be encoded using $O((\log|\bintypesnoparameters|+\log(\mathit{maxmult}))\cdot \mathit{maxtypes})$ bits, and it follows that the binary encoding of each instance has length $O((\log|\bintypesnoparameters|+\log(\mathit{maxmult}))\cdot \mathit{maxtypes})$. Moreover, the binary encoding of each instance has length $\Omega(\mathit{maxtypes}+\log(\mathit{maxmult}))$ since all of the ordered pairs must be encoded in the instance (and there are at least $\mathit{maxtypes}$ of these) and the multiplicity of each \bin type must be encoded in the instance (and encoding the largest such multiplity uses $\Omega(\log(\mathit{maxmult}))$ bits). We now proceed to prove that $\brepackingk$ can be solved in $O(\mathit{maxtypes}+\log(\mathit{maxmult}))$ time.

	\thmpartition*
	
	\begin{proof}   
		At a high level, we first describe how to convert $\partitionILP$ into an integer linear program of the form $Ax \leq b$ with $A \in \mathbb{Z}^{\mu \times \nu}, b \in \mathbb{Z}^{\mu}$ with solution $x \in \mathbb{Z}_{\geq 0}^{\nu}$. The value of $\mu$ is the number of constraints, and the value of $\nu$ is the number of variables. Then, we will prove that $\mu,\nu \in O(1)$ and also provide an upper bound $s$ on the length of the binary encoding of each entry of the integer linear program, where $s$ is linear in the length of the binary encoding of the input instance of $\brepackingk$. Finally, to obtain the desired result, we will apply a result by Eisenbrand \cite{Eisenbrand2003} which proves that an integer linear program in the form described above with $O(1)$ variables and $O(1)$ constraints can be solved using $O(s)$ arithmetic operations.
		
		It will be helpful to re-write $\partitionILP$ as the equivalent integer linear program below, which was obtained by taking each term of the form $\flow(\edgesetgeneric)$ for an edge set $\edgesetgeneric$ and replacing it with $\sum_{\euv \in \edgesetgeneric} \flow(\euv)$, and writing out the non-negativity/non-positivity requirements as inequalities.
		
		\setcounter{equation}{1}
		{\footnotesize
			\begin{align}
				\nonumber \text{minimize} ~ 0^T \begin{bmatrix} \flow \\ \demand \end{bmatrix} & ~ \\
				- \left(\sum_{\euv \in \outedges(\innode_\subconfig)} \flow(\euv)\right)
				- \demand(\innode_\subconfig) &= 0 & ~ \forall \subconfig \in \bsubs \\
				\left(\sum_{\euv \in \inedges(\internode_{\subconfig})} \flow(\euv)\right)
				- 
				\left(\sum_{\euv \in \outedges(\internode_{\subconfig})} \flow(\euv)\right)
				&= 0 & ~ \forall \subconfig \in \bsubs \\
				\left(\sum_{\euv \in \inedges(\outnode_\subconfig)} \flow(\euv)\right)
				- \demand(\outnode_\subconfig) &= 0 & ~ \forall \subconfig \in \bsubs \\
				\sum_{\subconfig \in \bsubs} \binmult{B}{\subconfig} \cdot \demand(\innode_\subconfig) &= -\binmult{B}{\source} & ~ \forall B \in \bintypesnoparameters \\
				\sum_{\subconfig \in \bsubs} \binmult{B}{\subconfig} \cdot \demand(\outnode_\subconfig) &= \binmult{B}{\target} & ~ \forall B \in \bintypesnoparameters \\
				\flow(uv) &\geq 0 & ~ \forall uv \in \edgeset \\
				\demand(\innode_\subconfig) & \leq 0 & ~ \forall \subconfig \in \bsubs \\
				\demand(\outnode_\subconfig) & \geq 0 & ~ \forall \subconfig \in \bsubs\label{con:eight} \\
				\flow(uv), \demand(\innode_\subconfig), \demand(\outnode_\subconfig) &\in \mathbb{Z} & ~ \forall uv \in \edgeset,\forall \subconfig \in \bsubs\label{con:nine}
			\end{align}
		}
		To obtain a formulation $Ax \leq b$, we first convert the constraints so that they are all $\leq$-inequalities. First, duplicate each equality constraint and change the $=$ in one copy to $\leq$ and change the $=$ in the other copy to $\geq$. Next, convert each $\geq$-inequality into a $\leq$-inequality by multiplying both sides by $-1$. The resulting equivalent formulation is given below.
		{\footnotesize
			\begin{align*}
				\nonumber \text{minimize} ~ 0^T \begin{bmatrix} \flow \\ \demand \end{bmatrix} & ~ \\
				- \left(\sum_{\euv \in \outedges(\innode_\subconfig)} \flow(\euv)\right)
				- \demand(\innode_\subconfig) &\leq 0 & ~ \forall \subconfig \in \bsubs\tag{\ref{con:one}}\\
				\left(\sum_{\euv \in \outedges(\innode_\subconfig)} \flow(\euv)\right)
				+ \demand(\innode_\subconfig) &\leq 0 & ~ \forall \subconfig \in \bsubs\tag{\ref{con:one}'} \\
				\left(\sum_{\euv \in \inedges(\internode_{\subconfig})} \flow(\euv)\right)
				- 
				\left(\sum_{\euv \in \outedges(\internode_{\subconfig})} \flow(\euv)\right)
				&\leq 0 & ~ \forall \subconfig \in \bsubs\tag{\ref{con:two}} \\
				-\left(\sum_{\euv \in \inedges(\internode_{\subconfig})} \flow(\euv)\right)
				+\left(\sum_{\euv \in \outedges(\internode_{\subconfig})} \flow(\euv)\right)
				&\leq 0 & ~ \forall \subconfig \in \bsubs\tag{\ref{con:two}'} \\
				\left(\sum_{\euv \in \inedges(\outnode_\subconfig)} \flow(\euv)\right)
				- \demand(\outnode_\subconfig) &\leq 0 & ~ \forall \subconfig \in \bsubs\tag{\ref{con:three}} \\
				-\left(\sum_{\euv \in \inedges(\outnode_\subconfig)} \flow(\euv)\right)
				+ \demand(\outnode_\subconfig) &\leq 0 & ~ \forall \subconfig \in \bsubs\tag{\ref{con:three}'} \\
				\sum_{\subconfig \in \bsubs} \binmult{B}{\subconfig} \cdot \demand(\innode_\subconfig) &\leq -\binmult{B}{\source} & ~ \forall B \in \bintypesnoparameters\tag{\ref{con:four}} \\
				-\sum_{\subconfig \in \bsubs} \binmult{B}{\subconfig} \cdot \demand(\innode_\subconfig) &\leq \binmult{B}{\source} & ~ \forall B \in \bintypesnoparameters\tag{\ref{con:four}'} \\
				\sum_{\subconfig \in \bsubs} \binmult{B}{\subconfig} \cdot \demand(\outnode_\subconfig) &\leq \binmult{B}{\target} & ~ \forall B \in \bintypesnoparameters\tag{\ref{con:five}} \\
				-\sum_{\subconfig \in \bsubs} \binmult{B}{\subconfig} \cdot \demand(\outnode_\subconfig) &\leq -\binmult{B}{\target} & ~ \forall B \in \bintypesnoparameters\tag{\ref{con:five}'} \\
				-\flow(uv) &\leq 0 & ~ \forall uv \in \edgeset\tag{\ref{con:six}} \\
				\demand(\innode_\subconfig) & \leq 0 & ~ \forall \subconfig \in \bsubs\tag{\ref{con:seven}} \\
				-\demand(\outnode_\subconfig) & \leq 0 & ~ \forall \subconfig \in \bsubs & \tag{\ref{con:eight}} \\
				\flow(uv), \demand(\innode_\subconfig), \demand(\outnode_\subconfig) &\in \mathbb{Z} & ~ \forall uv \in \edgeset,\forall \subconfig \in \bsubs & \tag{\ref{con:nine}}
			\end{align*}
		}
		
		The set of variables of $\partitionILP$ consists of: one variable of the form $\flow(\euv)$ for each $\euv \in E(G)$, one variable of the form $\demand(\innode_\subconfig)$ for each $\subconfig \in \bsubs$, and one variable of the form $\demand(\outnode_\subconfig)$ for each $\subconfig \in \bsubs$. Thus, the number of variables is $\nu = |\edgeset|+2\cdot|\bsubs|$.
		
		To define the matrix $A$, we start by creating one column for each of the variables listed above. Then, we create one row in $A$ and a corresponding entry in the vector $b$ for each of the constraints of types \ref{con:one} through \ref{con:eight}, as follows:
		
		\begin{description}
			\item[Constraint (\ref{con:one}):] For each $\subconfig \in \bsubs$, create a row in $A$ and a corresponding entry in $b$. In each such row of $A$, write $-1$ in the column corresponding to variable $\demand(\innode_\subconfig)$, write $-1$ in the column corresponding to $\flow(\euv)$ for each $\euv \in \outedges(\innode_\subconfig)$, and write $0$ in all other columns. In the corresponding entry of $b$, write $0$.
			\item[Constraint (\ref{con:one}'):] For each $\subconfig \in \bsubs$, create a row in $A$ and a corresponding entry in $b$. In each such row of $A$, write $1$ in the column corresponding to variable $\demand(\innode_\subconfig)$, write $1$ in the column corresponding to $\flow(\euv)$ for each $\euv \in \outedges(\innode_\subconfig)$, and write $0$ in all other columns. In the corresponding entry of $b$, write $0$.
			\item[Constraint (\ref{con:two}):] For each $\subconfig \in \bsubs$, create a row in $A$ and a corresponding entry in $b$. In each such row of $A$, write $1$ in the column corresponding to $\flow(\euv)$ for each $\euv \in \inedges(\internode_\subconfig)$, write $-1$ in the column corresponding to $\flow(\euv)$ for each $\euv \in \outedges(\internode_\subconfig)$, and write $0$ in all other columns. In the corresponding entry of $b$, write $0$.
			\item[Constraint (\ref{con:two}'):] For each $\subconfig \in \bsubs$, create a row in $A$ and a corresponding entry in $b$. In each such row of $A$, write $-1$ in the column corresponding to $\flow(\euv)$ for each $\euv \in \inedges(\internode_\subconfig)$, write $1$ in the column corresponding to $\flow(\euv)$ for each $\euv \in \outedges(\internode_\subconfig)$, and write $0$ in all other columns. In the corresponding entry of $b$, write $0$.
			\item[Constraint (\ref{con:three}):] For each $\subconfig \in \bsubs$, create a row in $A$ and a corresponding entry in $b$. In each such row of $A$, write $-1$ in the column corresponding to variable $\demand(\outnode_\subconfig)$, write $1$ in the column corresponding to $\flow(\euv)$ for each $\euv \in \inedges(\outnode_\subconfig)$, and write $0$ in all other columns. In the corresponding entry of $b$, write $0$.
			\item[Constraint (\ref{con:three}'):] For each $\subconfig \in \bsubs$, create a row in $A$ and a corresponding entry in $b$. In each such row of $A$, write $1$ in the column corresponding to variable $\demand(\outnode_\subconfig)$, write $-1$ in the column corresponding to $\flow(\euv)$ for each $\euv \in \inedges(\outnode_\subconfig)$, and write $0$ in all other columns. In the corresponding entry of $b$, write $0$.
			\item[Constraint (\ref{con:four}):] For each $\binb \in \bintypesnoparameters$, create a row in $A$ and a corresponding entry in $b$. In each such row of $A$, write $\binmult{\binb}{\subconfig}$ in the column corresponding to variable $\demand(\innode_\subconfig)$ for each $\subconfig \in \bsubs$, and write $0$ in all other columns. In the corresponding entry of $b$, write $-\binmult{\binb}{\source}$.
			\item[Constraint (\ref{con:four}'):] For each $\binb \in \bintypesnoparameters$, create a row in $A$ and a corresponding entry in $b$. In each such row of $A$, write $-\binmult{\binb}{\subconfig}$ in the column corresponding to variable $\demand(\innode_\subconfig)$ for each $\subconfig \in \bsubs$, and write $0$ in all other columns. In the corresponding entry of $b$, write $\binmult{\binb}{\source}$.
			\item[Constraint (\ref{con:five}):] For each $\binb \in \bintypesnoparameters$, create a row in $A$ and a corresponding entry in $b$. In each such row of $A$, write $\binmult{\binb}{\subconfig}$ in the column corresponding to variable $\demand(\outnode_\subconfig)$ for each $\subconfig \in \bsubs$, and write $0$ in all other columns. In the corresponding entry of $b$, write $\binmult{\binb}{\target}$.
			\item[Constraint (\ref{con:five}'):] For each $\binb \in \bintypesnoparameters$, create a row in $A$ and a corresponding entry in $b$. In each such row of $A$, write $-\binmult{\binb}{\subconfig}$ in the column corresponding to variable $\demand(\outnode_\subconfig)$ for each $\subconfig \in \bsubs$, and write $0$ in all other columns. In the corresponding entry of $b$, write $-\binmult{\binb}{\target}$.
			\item[Constraint (\ref{con:six}):] For each $\euv \in \edgeset$, create a row in $A$ and a corresponding entry in $b$. In each such row of $A$, write $-1$ in the column corresponding to $\flow(\euv)$, and write $0$ in all other columns. In the corresponding entry of $b$, write $0$.
			\item[Constraint (\ref{con:seven}):] For each $\subconfig \in \bsubs$, create a row in $A$ and a corresponding entry in $b$. In each such row of $A$, write $1$ in the column corresponding to $\demand(\innode_\subconfig)$, and write $0$ in all other columns. In the corresponding entry of $b$, write $0$.
			\item[Constraint (\ref{con:eight}):] For each $\subconfig \in \bsubs$, create a row in $A$ and a corresponding entry in $b$. In each such row of $A$, write $-1$ in the column corresponding to $\demand(\outnode_\subconfig)$, and write $0$ in all other columns. In the corresponding entry of $b$, write $0$.
		\end{description}
		
		This completes the construction of $A$ and $b$, and thus we now have an integer linear program of the form $Ax \leq b$ that is equivalent to $\partitionILP$. The number of rows in each of $A$ and $b$ is calculated by adding up the number of constraints of types \ref{con:one} through \ref{con:eight}, i.e., $\mu = 8\cdot|\bsubs| + 4\cdot|\bintypesnoparameters| + |\edgeset|$. The number of columns in $A$ is the number of variables $\nu = |\edgeset|+2\cdot|\bsubs|$.
		
		We next set out to provide upper bounds on the values of $\mu$ and $\nu$ using expressions involving only $\partitionbound$ and $\capacity$. First, recall that $\bintypesnoparameters$ was defined as the set of all possible \bins with volume at most $\capacity$ that can be created using integer items from the universe $U$. To obtain an upper bound on the number of such \bins, note that each possible \bin is an integer partition of $\capacity$ with at most one part representing the amount of slack in the \bin, so $|\bintypesnoparameters|$ is bounded above by the number of integer partitions of $\capacity$ (denoted by the \emph{partition function} $p(\capacity)$), multiplied by $\capacity+1$ to account for each possible part of the integer partition that can represent the slack (if any) since a \bin can contain at most $\capacity$ items. Using the asymptotic expression from \cite{hardy1918asymptotic} for the partition function, we conclude that $|\bintypesnoparameters| \leq (\capacity+1)p(\capacity) \in \Theta(\frac{(\capacity+1)e^{\pi\sqrt{2\kappa/3}}}{4\kappa\sqrt{3}})$.

		Next, recall that $\bsubs$ was defined as the set of all non-empty configurations consisting of at most $\partitionbound$ \bins of capacity $\capacity$ that can be created using integer items from the universe $U$. To obtain an upper bound on the number of such configurations, note that each such configuration can be identified by how many of each \bin type it contains, i.e., it can be identified as a set of ordered pairs $\{(1,\#_1),\ldots,(|\bintypesnoparameters|,\#_{|\bintypesnoparameters|})\}$ where each tuple $(i,\#_i)$ is interpreted as $i \in \{1,\ldots,|\bintypesnoparameters|\}$ representing a fixed \bin type and $\#_i$ representing how many times \bin type $i$ appears in the configuration. Since the number of \bins in the configuration is at most $\partitionbound$, we conclude that $\#_i \in \{0,\ldots,\partitionbound\}$ for each $i \in \{1,\ldots,|\bintypesnoparameters|\}$, and it follows that the number of such sets of ordered pairs is bounded above by $(\partitionbound+1)^{|\bintypesnoparameters|}$, which implies that $|\bsubs| \leq (\partitionbound+1)^{|\bintypesnoparameters|} \leq (\partitionbound+1)^{(\capacity+1)p(\capacity)}$. Finally, recall that the graph $\graph$ has three vertices $\innode_\subconfig, \internode_\subconfig, \outnode_\subconfig$ for each $\subconfig \in \bsubs$ and has edge set $\edgeset = \{\innode_\subconfig \internode_\subconfig: \subconfig \in \bsubs \} 
		\cup \{\internode_{\subconfig} \internode_{\subconfigtwo}: \subconfig, \subconfigtwo \in \bsubs \text{, $\subconfig$ and $\subconfigtwo$ are adjacent} \}
		\cup \{\internode_\subconfig \outnode_\subconfig: \subconfig \in \bsubs \}$. Thus, by bounding the size of each of the edge sets in the definition of $\edgeset$, we can bound $|\edgeset|$ from above using the expression $|\bsubs| + \binom{|\bsubs|}{2} + |\bsubs| \leq |\bsubs| + |\bsubs|^2 + |\bsubs| \leq 3\cdot|\bsubs|^2 \leq 3\cdot\left((\partitionbound+1)^{(\capacity+1)p(\capacity)}\right)^2$. Substituting the computed upper bounds on $|\bintypesnoparameters|$, $|\bsubs|$, and $|\edgeset|$ back into our expressions for $\mu$ and $\nu$, we get $\mu \leq 8\cdot\left((\partitionbound+1)^{(\capacity+1)p(\capacity)}\right) + 4\cdot\left((\capacity+1)p(\capacity)\right) + \left(3\cdot\left((\partitionbound+1)^{(\capacity+1)p(\capacity)}\right)^2\right)$ and $\nu \leq \left(3\cdot\left((\partitionbound+1)^{(\capacity+1)p(\capacity)}\right)^2\right)+2\cdot\left((\partitionbound+1)^{(\capacity+1)p(\capacity)}\right)$.
		
		Recall that $\brepackingk$ problem is defined with respect to fixed constants $\partitionbound$ and $\capacity$, i.e., these values do not vary across input instances. Using this fact along with the upper bounds established in the previous paragraph, we conclude that $|\bintypesnoparameters|,|\bsubs|,|\edgeset|,\mu,\nu \in O(1)$ with respect to the size of the instances of $\brepackingk$. Also, since the number of \bins in any $\subconfig \in \bsubs$ is bounded above by $\beta \in O(1)$, it follows that the value of any expression of the form $\binmult{\binb}{\subconfig}$ for $\subconfig \in \bsubs$ is at most $\beta \in O(1)$. Moreover, since we defined $\mathit{maxmult}$ to be the maximum number of times that a \bin type appears in $\source$ or $\target$, it follows that the value of any expression of the form $\binmult{\binb}{\source}$ or $\binmult{\binb}{\target}$ is at most $\mathit{maxmult}$.
		
		Next, we show that the integer linear program $Ax \leq b$ can be solved quickly. Theorem 1 in Eisenbrand \cite{Eisenbrand2003} showed that such an integer linear program with $O(1)$ variables and $O(1)$ constraints can be solved using $O(s)$ arithmetic operations, where $s$ is an upper bound on the length of the binary encoding of each entry. Since we previously showed that the number of variables $\nu$ and the number of constraints $\mu$ are both $O(1)$, it is sufficient for us to compute such an upper bound $s$. Namely, it is sufficient to find the maximum absolute value taken over all entries in $A$ and $b$. However, by the construction of $A$ and $b$, we note that each entry is taken from the set {\small $\{-1,0,1\} \cup \left(\bigcup_{\subconfig \in \bsubs}\{\binmult{\binb}{\subconfig},-\binmult{\binb}{\subconfig}\}\right) \cup \left(\bigcup_{\binb \in \bintypesnoparameters} \{\binmult{\binb}{\source},-\binmult{\binb}{\source}\}\right) \cup \left(\bigcup_{\binb \in \bintypesnoparameters} \{\binmult{\binb}{\target},-\binmult{\binb}{\target}\}\right)$}. From the previous paragraph, we conclude that the absolute value of each entry of $A$ and $b$ is bounded above by $\mathit{maxmult}$, and thus the length of the binary encoding of each entry is bounded above by $O(\log(\mathit{maxmult}))$. So, we have proven that there is an upper bound $s \in O(\log(\mathit{maxmult}))$ on the length of the binary encoding of each constraint, which, by Theorem 1 in \cite{Eisenbrand2003}, proves that the integer linear program $Ax \leq b$ can be solved in $O(\log(\mathit{maxmult}))$ time.
		
		We now prove that $\partitionILP$ can be solved in $O(\mathit{maxtypes} + \log(\mathit{maxmult}))$ time by verifying that it can be converted to the integer linear program $Ax \leq b$ in $O(\mathit{maxtypes})$ time. The conversion described earlier in the proof describes how each entry of $A$ and $b$ is computed, and since there are $\mu\nu + \mu \in O(1)$ such entries, the overall computation time is bounded asymptotically by the time needed to compute the entry with largest absolute value. Since all entries come from the set {\small $\{-1,0,1\} \cup \left(\bigcup_{\subconfig \in \bsubs}\{\binmult{\binb}{\subconfig},-\binmult{\binb}{\subconfig}\}\right) \cup \left(\bigcup_{\binb \in \bintypesnoparameters} \{\binmult{\binb}{\source},-\binmult{\binb}{\source}\}\right) \cup \left(\bigcup_{\binb \in \bintypesnoparameters} \{\binmult{\binb}{\target},-\binmult{\binb}{\target}\}\right)$}, it suffices to determine the time to compute the \bin multiplicities. For any value of form $\binmult{\binb}{\subconfig}$ with $\subconfig \in \bsubs$, we note that $\subconfig$ consists of at most $\partitionbound$ \bins, each of capacity $\capacity$, where $\partitionbound,\capacity \in O(1)$, so scanning the entire configuration $\subconfig$ and counting how many times $\binb$ occurs will take $O(1)$ time. Computing a value of the form $\binmult{\binb}{\source}$ or $\binmult{\binb}{\target}$ takes $O(\mathit{maxtypes})$ time, since computing each quantity involves scanning all ordered pairs in the corresponding configuration in the input to find the pair containing $\binb$ as its first element. 
		
		Thus, from the discussion above, we conclude that $\partitionILP$ can be solved in $O(\mathit{maxtypes} + \log(\mathit{maxmult}))$ time. By Lemma~\ref{lem:partition_reduction}, solving $\partitionILP$ tells us whether or not the corresponding instance of $\brepackingk$ is reconfigurable, which proves that $\brepackingk$ can be solved in $O(\mathit{maxtypes} + \log(\mathit{maxmult}))$ time.
	\end{proof}  
	
	\section{Conclusions and Future Work}\label{sec-future}
	
	We have introduced the area of the reconfiguration of multisets, demonstrating the hardness of the general problem, and providing algorithms for situations in which items are of bounded size, item and \bin capacities are powers of two, and items can be partitioned into smaller groups of \bins. Our results are applicable to a variety of application areas, as well as to the problem of \textsc{Bin Packing}. 
	
	In this paper, we have restricted our attention to instances in which all \bins have the same capacity and in which items and \bins are indistinguishable. 
	Future directions of research 
	include the non-uniform case, in which \bins have different capacities; variants could include restrictions on the minimum total size, restrictions on both the minimum and maximum total size, or other specifications for particular sets of \bins. The instances could be further generalized by allowing us to distinguish among items or \bins, even those of the same size or capacity. Another direction for future work is to optimize the time complexities of the algorithms studied in this paper: while our focus has been to establish the feasibility of reconfiguration, it is natural to ask how quickly one can reconfigure one configuration to another.
	
	Our work demonstrates the use of integer programming to make use of parallel computation in reconfiguration. Further investigations are required to determine whether similar techniques may be applicable more generally to other types of reconfiguration problems, and whether other techniques are amenable to parallelization of other reconfiguration problems.
	
	Each of the algorithms presented in the paper is restricted to using the amount of space provided in the source configuration. A natural extension is to ask whether there are no-instances for our problems that can be rearranged using extra space, and if so, how much space would be needed. In shifting the focus of reconfiguration from length of reconfiguration sequences to the impact of extra space, we open up a new type of reconfiguration problem. 
	
	We anticipate that {\em resource-focused reconfiguration} will find widespread use in practical settings for a host of various problems. As examples, we consider two commonly-studied reconfiguration problems, namely \textsc{Power Supply Reconfiguration}~\cite{Ito2011} and \textsc{Independent Set Reconfiguration}~\cite{Ito2011}. In the former problem, the goal is to reassign customers (each with a specified requirement) to power stations (each with a specified capacity) in such a way that aside from the customer being moved, all other customers have an uninterrupted flow of power. In a resource-focused setting, our goal would be determine how many generators might be needed to temporarily provide power in order to make reconfiguration possible. For the latter problem, a possible extra resource could be isolated vertices to which tokens could be temporarily moved in order to allow reconfiguration to occur. Depending on the problem, there may be more than one resource that can be measured in the reconfiguration process, yielding further new research directions as the impacts of resources are considered individually and in concert.
	

	\bibliographystyle{plain}
	\bibliography{sn-bibliography-tidy}
	
	\newpage
 
	\appendix

    \section{Hardness of \textsc{Restricted Bin Packing}}\label{app:hardness}
The reduction in the proof of Lemma \ref{lem:restrictedBP} involves the following two problems.\\

\textsc{Bin Packing}
    \begin{description}
		\item[Input:] a multiset of $n$ positive integers $\multiset{z_1,...,z_n}$, a positive integer $m$, and a positive integer $\alpha$.
		\item[Question:] Can $\multiset{z_1,...,z_n}$ be partitioned into at most $m$ multisets such that the numbers in each multiset sum to at most $\alpha$?
	\end{description}
	
 \textsc{Restricted Bin Packing}
    \begin{description}
		\item[Input:] a multiset of $n$ positive integers $\multiset{z_1,...,z_n}$, a positive integer $m$, and a positive integer $\alpha$ such that $\alpha \geq \max\{z_1,\ldots,z_n\}$, $\alpha \geq 2$, and $n \geq 2m+2$.
		\item[Question:] Can $\multiset{z_1,...,z_n}$ be partitioned into at most $m$ multisets such that the numbers in each multiset sum to at most $\alpha$?
	\end{description}
 
    \restrictedBP*
    \begin{proof}
        We provide a polynomial-time reduction from \textsc{Bin Packing} to \textsc{Restricted Bin Packing}. 
    
    We first deal with several classes of instances of \textsc{Bin Packing} that can be solved quickly:
		\begin{itemize}
			\item If $\alpha < \max\{z_1,\ldots,z_n\}$, then $P$ is a no-instance. This is because at least one of $z_1,\ldots,z_n$ is too large to be placed in any part of the partition.
			\item If $\alpha \geq \max\{z_1,\ldots,z_n\}$ and $m \geq n$, then $P$ is a yes-instance. This is because we can form a partition by placing each $z_i$ into its own part.
			\item If $\alpha \geq \max\{z_1,\ldots,z_n\}$ and $m=n-1$, then we can determine whether or not $P$ is a yes-instance by checking whether or not the two smallest elements in $\multiset{z_1,\ldots,z_n}$ sum to at most $\alpha$. If the sum of the two smallest elements is at most $\alpha$, then we can form a partition by placing the two smallest elements together in one part, and the remaining elements of $\multiset{z_1,\ldots,z_n}$ can each be placed in their own part. If the sum of the two smallest elements is greater than $\alpha$, then no two elements of $\multiset{z_1,\ldots,z_n}$ can be placed in the same part of the partition, so there is no partition into $n-1$ parts where each part's sum is at most $\alpha$.
			\item If $\alpha = 1$ and $m \leq n-2$, then $P$ is a no-instance. Since $n > m$, the Pigeonhole Principle implies that at least two elements of $\multiset{z_1,\ldots,z_n}$ must be placed in the same part, and since all elements of $\multiset{z_1,\ldots,z_n}$ are positive integers, it follows that the sum of any two elements is strictly bigger than $1$, i.e., strictly bigger than $\alpha$.
		\end{itemize}
		Consider any instance of \textsc{Bin Packing}, and observe that, in polynomial time, we can determine if it falls into one of the above cases, and, if so, we can determine in polynomial time whether it is a yes-instance or a no-instance. For each no-instance that falls into one of the above cases, we map it to the following no-instance of \textsc{Restricted Bin Packing}: $z_1=z_2=z_3=z_4=1$, $m=1$, $\alpha=2$. For each yes-instance that falls into one of the above cases, we map it to the following yes-instance of \textsc{Restricted Bin Packing}: $z_1=z_2=z_3=z_4=1$, $m=1$, $\alpha=4$.
        So, to complete the reduction from \textsc{Bin Packing} to \textsc{Restricted Bin Packing}, it remains to consider instances of \textsc{Bin Packing} that do not fall into any of the above cases, which we do in the remainder of the proof.
  
        Consider any instance $P=(z_1,...,z_n,m,\alpha)$ of \textsc{Bin Packing} such that $\alpha \geq \max\{z_1,\ldots,z_n\}$, $\alpha \geq 2$, and $m \leq n-2$. Create an instance $P'$ starting from $P$ and scaling each of $z_1,\ldots,z_n,\alpha$ up by a factor of $4m+4$, i.e., $P' = (z_1',\ldots,z_n',m,\alpha')$ where $z_i' = (4m+4)z_i$ for $i \in \{1,\ldots,n\}$ and $\alpha' = (4m+4)\alpha$. Note that $P$ is a yes-instance if and only if the scaled instance $P'$ is a yes-instance, since, for any multiset $\mathcal{M}$ taken from $z_1,\ldots,z_n$, we see that $\sum_{z_i \in \mathcal{M}} z_i \leq \alpha$ if and only if $\sum_{z_i \in \mathcal{M}} (4m+4)z_i \leq (4m+4)\alpha$. Next, we create an instance $P''$ starting from $P'$ and adding $2m+2$ copies of $1$ to the multiset of integers to be partitioned, and increasing the bound $\alpha'$ by an additive $2m+2$ term, i.e., $P''=(1_1,\ldots,1_{2m+2},z_1',\ldots,z_n',m,\alpha'')$ where $1_j$ represents the $j$th added copy of 1 for each $j \in \{1,\ldots,2m+2\}$, the integer $z_i'$ is equal to $(4m+4)z_i$ for each $i \in \{1,\ldots,n\}$, and $\alpha'' = (4m+4) \alpha + (2m + 2)$. Note that constructing this instance takes polynomial time since adding $2m+2$ copies of 1 to the multiset of integers to be partitioned takes $O(n)$ time (since we are only considering instances of \textsc{Bin Packing} where $m \leq n-2$). We see that $P'$ is a yes-instance if and only if $P''$ is a yes-instance, since all of the scaled integers $z_1',\ldots,z_n'$ are strictly larger than $2m+2$, so the extra $2m+2$ included in the bound $\alpha''$ can only accommodate the $2m+2$ added copies of 1. Finally, we confirm that the instance $P''$ satisfies the three input restrictions given in the definition of \textsc{Restricted Bin Packing}:
        \begin{itemize}
            \item Due to including $2m+2$ copies of 1 alongside the integers $z_1',\ldots,z_n'$, we see that the new instance $P''$ has $n + (2m+2) > 2m+2$ elements in the multiset of integers to partition, as required.
            \item Since we are only considering instances of \textsc{Bin Packing} where $\alpha \geq 2$, we see that $\alpha'' = (4m+4)\alpha + (2m+2) \geq 10m+10 > 2$, as required.
            \item Consider any $i \in \{1,\ldots,n\}$. Since we are only considering instances of \textsc{Bin Packing} where $\alpha \geq \max\{z_1,\ldots,z_n\}$, we know that $\alpha \geq z_i$. Along with the fact that we set $z_i' = (4m+4)z_i$, it follows that $\alpha'' = (4m+4)\alpha + (2m+2) \geq (4m+4)z_i + (2m+2) = z_i' + (2m+2) > z_i'$. Thus, we have shown that $\alpha'' > \max\{z_1',\ldots,z_n'\}$. We also proved above that $\alpha'' > 2$, so we conclude that $\alpha'' \geq \max\{1_1,\ldots,1_{2m+2},z_1',\ldots,z_n'\}$, as required.
        \end{itemize}
    \end{proof}

	\section{Proof of Correctness for \textit{Check-Settled}}\label{App:CheckSettled}
	\CheckSettledCorrect*
	\begin{proof}
		Suppose that, during the execution of \textit{Check-Settled$(s,\configuration)$}, when considering all item sizes $u \geq s$ in $\configuration$ in increasing order, it is possible to find a bijection $\varphi_u$ from the multiset of \bins $B$ in $\configuration\setminus\varphi$ that contain an item of size $u$ to the multiset of \bins in $\target\setminus\varphi$ that contain an item of size $u$ in such a way that $\AtLeast{u}{B}{\configuration} = \AtLeast{u}{\varphi_u(B)}{\target}$. We consider the mapping $\varphi$ at the end of the execution of \textit{Check-Settled$(s,\configuration)$} and prove that it is a settling bijection for item size $s$. Consider two item sizes $u',u'' \geq s$ such that $u' < u''$, and notice that the domains of $\varphi_{u'}$ and $\varphi_{u''}$ are disjoint, and that the ranges of $\varphi_{u'}$ and $\varphi_{u''}$ are disjoint: since $u' < u''$, the bijection $\varphi_{u'}$ was already included in $\varphi$ when constructing $\varphi_{u''}$, so the \bins in the domain and range of $\varphi_{u'}$ were excluded when constructing $\varphi_{u''}$. From this fact, we conclude that the size of the domain of $\varphi$ is equal to the sum of the sizes of the domains of all the $\varphi_u$ for $u \geq s$, and, the size of the range of $\varphi$ is equal to the sum of the sizes of the ranges of all the $\varphi_u$ for $u \geq s$. But since each $\varphi_u$ is a bijection, the size of its domain and range are equal, so it follows that the domain and range of $\varphi$ have equal size, which proves the $\varphi$ is a bijection. Next, for any $u \geq s$, notice that any \bin $\binb$ in $\configuration$ that contains an item of size $u$ is in the domain of $\varphi$, since, if $\binb$ was not included in the domain of any $\varphi_{u'}$ with $u' < u$, then it was included in the domain of $\varphi_u$ when $\varphi_u$ was constructed. Similarly, any \bin $\binb$ in $\target$ that contains an item of size $u$ is in the range of $\varphi$. Thus, we conclude that $\varphi$ is a bijection from the multiset of \bins containing an item of size at least $s$ in $\configuration$ to the multiset of \bins containing an item of size at least $s$ in $\target$. Finally, we need show that $\AtLeast{s}{B}{\configuration} = \AtLeast{s}{\varphi(B)}{\target}$ for each \bin $\binb$ in the domain of $\varphi$. Consider any \bin $\binb$ in the domain of $\varphi$, and of all the items in $\binb$ that have size at least $s$, let $u$ be the size of the smallest. Since the item sizes were processed in increasing order, it follows that $\binb$ is in the domain of $\varphi_u$, and, since we assumed that $\AtLeast{u}{B}{\configuration} = \AtLeast{u}{\varphi_u(B)}{\target}$, it follows that $\AtLeast{u}{B}{\configuration} = \AtLeast{u}{\varphi(B)}{\target}$, and, since no items smaller than $u$ in $B$ have size at least $s$, it follows that $\AtLeast{s}{B}{\configuration} = \AtLeast{s}{\varphi(B)}{\target}$. This concludes the proof that $\varphi$ is a settling bijection for item size $s$.
		
		The above fact immediately gives us one direction of the correctness proof: if the execution of \textit{Check-Settled$(s,\configuration)$} is successful in finding each bijection $\varphi_u$ for $u \geq s$, i.e., it outputs ``settled", then we can conclude that $s$ is settled in $\configuration$ since the final mapping $\varphi$ is in fact a settling bijection for item size $s$.
		
		For the other direction of the correctness proof, suppose that $s$ is settled in $\configuration$. We prove that the execution of \textit{Check-Settled$(s,\configuration)$} returns ``settled", i.e., for each $u \geq s$, it is possible to find a bijection $\varphi_u$ from the multiset of \bins $B$ in $\configuration\setminus\varphi$ that contain an item of size $u$ to the multiset of \bins in $\target\setminus\varphi$ that contain an item of size $u$ in such a way that $\AtLeast{u}{B}{\configuration} = \AtLeast{u}{\varphi_u(B)}{\target}$. From the first fact we proved above, this implies the desired claim that \textit{Check-Settled$(s,\configuration)$} also produces a settling bijection for $s$. Now, starting from the fact that $s$ is settled in $\configuration$, we know by definition that there exists a settling bijection for $s$, which we call $\varphi^*$. We will use $\varphi^*$ to show, by induction, the existence of the desired $\varphi_u$ for each $u \geq s$. Consider an arbitrary $u \geq s$, assume that the desired bijection $\varphi_{u'}$ was found for all item sizes $u'$ with $s \leq u' < u$, and denote by $\varphi$ the union of all such $\varphi_{u'}$. In the base case, for item size $u=s$, note that it is vacuously true that the desired bijection $\varphi_{u'}$ was found for all item sizes $u'$ with $s \leq u' < u$. To construct $\varphi_u$, first add to the domain of $\varphi_u$ all of the \bins of $\configuration$ whose smallest item of size at least $s$ has size $u$. Then, for each \bin $\binb$ that was added to the domain of $\varphi_u$, we set $\varphi_u(B) = \varphi^*(B)$. We prove that the resulting $\varphi_u$ satisfies all the desired properties:
		\begin{itemize}
			\item Since each \bin $\binb$ that was added to the domain of $\varphi_u$ contains no items smaller than $u$ and at least $s$, we know that it is not added to the domain of any $\varphi_{u'}$ for any $s \leq u' < u$, so $\binb \in \configuration\setminus\varphi$.
			\item We confirm that the \bin $\varphi_u(B)$ was not included in the range of any $\varphi_{u'}$ with $s \leq u' < u$, i.e., we confirm that $\varphi_u(B) \in \target\setminus\varphi$. To obtain a contradiction, assume there is a \bin $B' \neq B$ in the domain of some $\varphi_{u'}$ such that $\varphi_{u'}(B') = \varphi_{u}(B)$, and notice by our construction that $\varphi^*(B') = \varphi_{u'}(B') = \varphi_{u}(B) = \varphi^*(B)$, which contradicts the fact that $\varphi^*$ is a bijection.
			\item We confirm that $\varphi_u$ is a mapping from the multiset of \bins $B$ in $\configuration\setminus\varphi$ that contain an item of size $u$ to the multiset of \bins in $\target\setminus\varphi$ that contain an item of size $u$. First, consider any \bin $\binb$ in $\configuration$ that contains an item of size $u \geq s$: if the smallest item with size at least $s$ in this \bin has size $u' < u$, then $\binb$ would be in the domain of $\varphi_{u'}$, and so it would be in the domain of $\varphi$; however, if the smallest item with size at least $s$ in $\binb$ has size $u$, then, by construction, it would be in the domain of $\varphi_{u}$. It follows that every \bin in $\configuration\setminus\varphi$ that contain an item of size $u$ is in the domain of $\varphi_u$. A similar proof shows that every \bin in $\target\setminus\varphi$ that contain an item of size $u$ is in the range of $\varphi_u$. 
			\item We prove that $\varphi_u$ is a bijection. To see why it is one-to-one: since $\varphi^*$ is one-to-one and $\varphi_u(B) = \varphi^*(B)$ for each $B$ in the domain of $\varphi_u$, then for any two distinct $B,B'$, we see that $\varphi_u(B) = \varphi^*(B) \neq \varphi^*(B') = \varphi_u(B')$. To see why it is onto: consider any \bin $B_T \in \target\setminus\varphi$ that contains an item of size $u$, and since $\varphi^*$ is onto, there exists a \bin $B \in \configuration$ such that $\varphi^*(B) = B_T$. Since $\varphi^*$ is a settling bijection, this $B$ satisfies $\AtLeast{s}{B}{\configuration} = \AtLeast{s}{\varphi^*(B)}{\target}$, so it follows that $B$ also contains an item of size $u$. Thus, $B$ is in the domain of some $\varphi_{u'}$ for some $u' \leq u$. However, it must be the case that $B$ is in the domain of $u$: otherwise, i.e., if $B$ is in the domain of some $\varphi_{u'}$ with $u' < u$, then $\varphi_{u'}(B) = B_T$ would be in the range of $\varphi_{u'}$, which contradicts the fact that $B_T \in \target\setminus\varphi$. Thus, we have shown that there exists a $B$ in the domain of $\varphi_u$ such that $\varphi_u(B) = B_T$, which concludes the proof that $\varphi_u$ is onto. Since $\varphi_u$ is one-to-one and onto, it is a bijection.
			\item Since $\varphi^*$ is a settling bijection for $s$, it satisfies the property $\AtLeast{s}{B}{\configuration} = \AtLeast{s}{\varphi^*(B)}{\target}$, and, because $u \geq s$,  it also satisfies $\AtLeast{u}{B}{\configuration} = \AtLeast{u}{\varphi^*(B)}{\target}$ , and since we set $\varphi_u(B) = \varphi^*(B)$, it follows that $\AtLeast{u}{B}{\configuration} = \AtLeast{u}{\varphi_u(B)}{\target}$.
		\end{itemize}  
	\end{proof}
	
	\section{Proofs for Section \ref{sec:appendix_flow}}\label{app:flows}
	
	\flowsumsub*
	\begin{proof}
		The proof is based on Lemma 2.19 in \cite{Williamson2019}. For the first case, we have to show that $\flow = \flow_1 \flowsum \flow_2$ is a flow, namely $\flow(\inedges(v)) - \flow(\outedges(v)) = 0$ for all $v \in \vertexset \setminus (\inputnodes \cup \outputnodes)$. 
		Since $\flow_1$ and $\flow_2$ are flows, they each satisfy the conservation constraint, so, for all $v \in \vertexset \setminus (\inputnodes \cup \outputnodes)$, we have
		\begin{align*}
			\flow(\inedges(v)) - \flow(\outedges(v)) &= [\flow_1(\inedges(v)) + \flow_2(\inedges(v))] - \left[\flow_1(\outedges(v)) + \flow_2(\outedges(v))\right] \\
			&= [\flow_1(\inedges(v)) - \flow_1(\outedges(v))] + [\flow_2(\inedges(v)) - \flow_2(\outedges(v))] \\
			&= 0.
		\end{align*}
		Furthermore, by the definition of the sum of flows,
		\begin{align*}
			\abs{\flow} = \sum_{\innode \in \inputnodes} \flow(\outedges(\innode)) &= \sum_{\innode \in \inputnodes} [\flow_1(\outedges(\innode)) + \flow_2(\outedges(\innode))]  \\
			&= \left[\sum_{\innode \in \inputnodes} \flow_1(\outedges(\innode))\right] + \left[\sum_{\innode \in \inputnodes} \flow_2(\outedges(\innode))\right]\\
			&= \abs{\flow_1} + \abs{\flow_2}.
		\end{align*}
		
		The proof of the flow subtraction case is similar.
		If $\flow_2$ is a subflow of $\flow_1$, then $\flow(uv) = \flow_1(uv) - \flow_2(uv)$ will always be non-negative. So, to show that $\flow$ is a flow, we just have to show that it satisfies the conservation constraint, i.e., for all $v \in \vertexset \setminus (\inputnodes \cup \outputnodes)$, we have
		\begin{align*}
			\flow(\inedges(v)) - \flow(\outedges(v)) &= [\flow_1(\inedges(v)) - \flow_2(\inedges(v))] - [\flow_1(\outedges(v)) - \flow_2(\outedges(v))] \\
			&= [\flow_1(\inedges(v)) - \flow_1(\outedges(v))] - [\flow_2(\inedges(v)) - \flow_2(\outedges(v))] \\
			&= 0.
		\end{align*}
		Similarly, by the definition of the subtraction of flows, we have 
		\begin{align*}
			\abs{\flow} = \sum_{\innode \in \inputnodes} \flow(\outedges(\innode)) &= \sum_{\innode \in \inputnodes} [\flow_1(\outedges(\innode)) - \flow_2(\outedges(\innode))]  \\
			&= \left[\sum_{\innode \in \inputnodes} \flow_1(\outedges(\innode))\right] - \left[\sum_{\innode \in \inputnodes} \flow_2(\outedges(\innode))\right] \\
			&= \abs{\flow_1} - \abs{\flow_2}.
		\end{align*}   
	\end{proof}

	\flowdecomp*
	\begin{proof}
		This proof uses the main ideas from the proof of Lemma 2.20 in \cite{Williamson2019}, and our result differs in that we are not providing a decomposition of the flow $\flow$ into $\ell$ subflows (where $\ell$ is the number of positive edges) whose sum is $\flow$: instead, we are providing a collection of $|\flow|$ unit path flows (each is a subflow of $\flow$, and its path includes an origin and destination node) whose sum is not necessarily equal to $\flow$.   
		
		In what follows, for any flow $\flow$ and any edge $e$, we say that $e$ is a \emph{positive edge in $\flow$} if $\flow(e) > 0$, and otherwise we say that $e$ is a \emph{zero edge in $\flow$}. The proof proceeds by induction on the number of positive edges in the flow. 
		
		For the base case, consider a flow $\flow$ with no positive edges, i.e., all edges in $\flow$ are zero edges. Then $\abs{\flow} = 0$ and the lemma statement is vacuously true.
		
		As induction hypothesis, assume that, for some positive integer $\ell$, the lemma statement is true for all flows $\flow$ on $\graph$ that have strictly fewer than $\ell$ positive edges. 
		
		For the inductive step, we prove that the lemma statement holds for an arbitrary flow $\flow$ on $\graph$ that has exactly $\ell >0$ positive edges. Starting from an arbitrary edge $uv$ such that $\flow(uv) > 0$, we repeatedly add positive edges until we form a $(\inputnodes, \outputnodes)$-path or a directed cycle consisting of positive edges (where the latter need not contain $uv$). Such a process is possible due to flow conservation: if $u \not\in \inputnodes$, there exists a node $u'$ such that $\flow(u'u) > 0$, and if $v \notin \outputnodes$, then there exists a node $v'$ such that $\flow(vv') > 0$. We consider two cases depending on the outcome of the above process:
		\begin{itemize}
			\item {\bf Case 1: The process forms a $(\inputnodes, \outputnodes)$-path $H$.} \\
			Let $\delta$ be the minimum value of $\flow(e)$ taken over all edges $e$ in $H$, and notice that $\delta$ is a positive integer since the process used to construct $H$ only included positive edges.
			
			We define an integral path flow $\flow^{*}$ by setting $\flow^{*}(e)=\delta$ if edge $e$ is in $H$, and $\flow^{*}(e)=0$ otherwise. We provide some facts about $\flow^{*}$ for later use:
			\begin{itemize}
				\item {\bf Fact 1:} $\abs{\flow^{*}} = \delta$. \\
				To see why, note that since $H$ is an $(\inputnodes, \outputnodes)$-path, there is exactly one edge in $H$ that is incident to a node in $X$. By definition, $\flow^{*}$ evaluates to $\delta$ on this edge, and $\flow^{*}$ evaluates to 0 on all other edges incident to a node in $X$. So, the fact follows from the definition of $\abs{\flow^{*}}$.
				\item {\bf Fact 2:} $\flow^{*}$ is a subflow of $\flow$. \\
				To see why, note that, by our definition of $\delta$, we know that $\flow(e) \geq \delta$ for each edge $e$ in $H$. So, from our definition of $\flow^{*}$, it follows that $\flow^{*}(e) = \delta \leq \flow(e)$ for each edge $e$ in $H$, and, $\flow^{*}(e) = 0 \leq \flow(e)$ for all other edges in $\graph$. Thus, the fact follows from the definition of subflow.
				\item {\bf Fact 3:} There exist $\delta$ unit path flows, denoted by $\flow^{*}_1,\ldots,\flow^{*}_\delta$, such that $\flow^{*}_i$ is a subflow of $\flow^{*}$ for each $i \in \{1,\ldots,\delta\}$, and, $\flow^{*} = \flowsummation_{i=1}^{\delta} \flow^{*}_{i}$. \\
				To see why, for each $i \in \{1,\ldots,\delta\}$, define $\flow^{*}_i$ by setting $\flow^{*}_i(e) = 1$ for each edge $e$ in $H$, and setting $\flow^{*}_i(e) = 0$ for all other edges in $\graph$. First, since $H$ is an $(\inputnodes, \outputnodes)$-path, it is clear that each $\flow^{*}_i$ is a unit path flow. Next, for each edge $e$ in $H$, the flow $\flowsummation_{i=1}^{\delta} \flow^{*}_{i}$ evaluates to $\delta$ on $e$, and evaluates to 0 on all other edges, which proves that $\flow^{*} = \flowsummation_{i=1}^{\delta} \flow^{*}_{i}$. Finally, consider an arbitrary $i \in \{1,\ldots,\delta\}$. For each edge $e$ in $H$, we have $\flow^{*}_{i}(e) = 1 \leq \delta = \flow^{*}(e)$, and, for all other edges in $\graph$, we have $\flow^{*}_{i}(e) = 0 = \flow^{*}(e)$, which proves that $\flow^{*}_i$ is a subflow of $\flow^{*}$.
			\end{itemize}

			Next, we create another integral flow $\flow'$ by subtracting $\flow^{*}$ from $\flow$, i.e., define $\flow' = \flow \flowsub \flow^{*}$. We provide some facts about $\flow'$ for later use:
			\begin{itemize}
				\item {\bf Fact 4:} $\flow'$ is a flow on $\graph$, and $\abs{\flow'} = \abs{\flow} - \abs{\flow^{*}} = \abs{\flow} - \delta$. \\
				To see why, note that $\flow^{*}$ is a subflow of $\flow$ by Fact 2. Then, the fact follows by applying Lemma \ref{lem:flow_sum_sub}.
				\item {\bf Fact 5:} $\flow'$ is a subflow of $\flow$. \\
				To see why, note that, for each edge $e$ in $H$, we have $\flow'(e) = \flow(e)-\flow^{*}(e) = \flow(e)-\delta < \flow(e)$, and, for all other edges in $\graph$, we have $\flow'(e) = \flow(e)-\flow^{*}(e) = \flow(e)-0 = \flow(e)$.
				\item {\bf Fact 6:} $\flow'$ has strictly fewer positive edges than $\flow$. \\
				To see why, note that, by our definition of $\delta$, there must be at least one edge $e_\delta$ in $H$ such that $\flow(e_\delta) = \delta$, and from the definition of $\flow'$, we see that $\flow'(e_\delta) = \flow(e_\delta) - \delta = 0$. This proves that there is at least one edge that is a positive edge in $\flow$ but is a zero edge in $\flow'$. Further, since we showed above that $\flow'$ is a subflow of $\flow$, we know that $\flow'(e) \leq \flow(e)$ for each edge $e$ in $G$, so it follows that each zero edge in $\flow$ is also a zero edge in $\flow'$. Altogether, this proves that $\flow'$ has more zero edges than $\flow$, which implies that $\flow'$ has strictly fewer positive edges than $\flow$.
			\end{itemize}
			
			By Fact 6, we can apply the induction hypothesis to $\flow'$. By doing so, we obtain $\abs{\flow'}$ unit path flows on $\graph$, denoted by $\flow'_1,\ldots,\flow'_{\abs{\flow'}}$, such that: $\flow'_i$ is a subflow of $\flow'$ for each $i \in \{1,\ldots,\abs{\flow'}\}$, and, $\flowsummation_{i=1}^{\abs{\flow'}} \flow'_{i}$ is a subflow of $\flow'$. Also, by applying Fact 3, we obtain $\delta$ unit path flows, denoted by $\flow^{*}_1,\ldots,\flow^{*}_\delta$, such that $\flow^{*}_i$ is a subflow of $\flow^{*}$ for each $i \in \{1,\ldots,\delta\}$, and, $\flow^{*} = \flowsummation_{i=1}^{\delta} \flow^{*}_{i}$. We claim (and prove) that $\flow'_1,\ldots,\flow'_{\abs{\flow'}},\flow^{*}_1,\ldots,\flow^{*}_\delta$ is a collection of $\abs{\flow}$ unit path flows on $G$ that satisfy the conditions of the lemma.
			
			First, we confirm that there are $\abs{\flow}$ elements in this collection. The number of elements is $\abs{\flow'}+\delta$, which, by Fact 4, is equal to $(\abs{\flow}-\delta)+\delta = \abs{\flow}$, as required.
			
			Next, we confirm that each element in this collection is a subflow of $\flow$. Consider an arbitrary $\flow'_i$ for some $i \in \{1,\ldots,\abs{\flow'}\}$. By the induction hypothesis, $\flow'_i$ is a subflow of $\flow'$, and, by Fact 5, we know that $\flow'$ is a subflow of $\flow$, which implies that $\flow'_i$ is a subflow of $\flow$, as required. Further, consider an arbitrary $\flow^{*}_i$ for some $i \in \{1,\ldots,\delta\}$. By Fact 3, $\flow^{*}_i$ is a subflow of $\flow^{*}$, and, by Fact 2, we know that $\flow^{*}$ is a subflow of $\flow$, which implies that $\flow^{*}_i$ is a subflow of $\flow$, as required.
			
			Finally, we confirm that, by taking the sum of all flows in the collection, the resulting flow is a subflow of $\flow$. In particular, our goal is to show that $\left(\flowsummation_{i=1}^{\abs{\flow'}} \flow'_{i}\right) \flowsum \left(\flowsummation_{i=1}^{\delta} \flow^{*}_{i}\right)$ is a subflow of $\flow$. By the definition of subflow, we must show that, for each edge $e$ in $\graph$:
			$$ \left(\left(\flowsummation_{i=1}^{\abs{\flow'}} \flow'_{i}\right) \flowsum \left(\flowsummation_{i=1}^{\delta} \flow^{*}_{i}\right)\right)(e) \leq \flow(e) $$ 
			To do this, we take the left side of this inequality, and re-write it using the definition of flow sum:
			$$ \left(\left(\flowsummation_{i=1}^{\abs{\flow'}} \flow'_{i}\right) \flowsum \left(\flowsummation_{i=1}^{\delta} \flow^{*}_{i}\right)\right)(e) = \left(\flowsummation_{i=1}^{\abs{\flow'}} \flow'_{i}\right)(e) + \left(\flowsummation_{i=1}^{\delta} \flow^{*}_{i}\right)(e) $$
			
			By the induction hypothesis, we know that $\flowsummation_{i=1}^{\abs{\flow'}} \flow'_{i}$ is a subflow of $\flow'$, so, by the definition of subflow, we can bound the term $\left(\flowsummation_{i=1}^{\abs{\flow'}} \flow'_{i}\right)(e)$ from above by $\flow'(e)$. Similarly, by Fact 3, we know that $\flowsummation_{i=1}^{\delta} \flow^{*}_{i}$ is equal to $\flow^{*}$, so we can replace the term $\left(\flowsummation_{i=1}^{\delta} \flow^{*}_{i}\right)(e)$ with $\flow^{*}(e)$. Thus, we have 
			$$ \left(\flowsummation_{i=1}^{\abs{\flow'}} \flow'_{i}\right)(e) + \left(\flowsummation_{i=1}^{\delta} \flow^{*}_{i}\right)(e) \leq \flow'(e) + \flow^{*}(e) $$ 
			
			Finally, by the definition of $\flow'$, we have $\flow'(e) = \flow(e) - \flow^{*}(e)$, so the righthand side of the previous inequality is equal to $\flow'(e)+\flow^{*}(e) = (\flow(e) - \flow^{*}(e))+\flow^{*}(e) = \flow(e)$, as required. \\
			
			\item {\bf Case 2: The process forms a directed cycle $H$.} \\
			Let $\delta$ be the minimum value of $\flow(e)$ taken over all edges $e$ in $H$, and notice that $\delta$ is a positive integer since the process used to construct $H$ only included positive edges.
			
			We define $\flow^{*}$ by setting $\flow^{*}(e)=\delta$ if edge $e$ is in $H$, and $\flow^{*}(e)=0$ otherwise. We provide some facts about $\flow^{*}$ for later use:
			\begin{itemize}
				\item {\bf Fact 7:} $\flow^{*}$ is an integral flow.\\
				To see why, note that each edge is assigned a non-negative integer (either 0 or $\delta$). Further, since $H$ is a directed cycle, each node in $H$ has exactly one incoming positive edge and exactly one outgoing positive edge. Since all positive edges evaluate to $\delta$ under $\flow^{*}$, it follows that the conservation constraint is satisfied at each node in $H$. Moreover, at all other nodes in $\vertexset \setminus (\inputnodes \cup \outputnodes)$, all incident edges evaluate to 0 under $\flow^{*}$, so they also satisfy the conservation constraint.
				\item {\bf Fact 8:} $\abs{\flow^{*}} = 0$.\\
				To see why, note that since $H$ is a directed cycle, each node in $H$ has exactly one incoming positive edge and one outgoing positive edge. By the construction of $G$, the nodes in $X$ have no incoming edges, so $H$ cannot contain any nodes from $X$. In particular, this means that each edge in $H$ is not incident to any node in $X$, so it follows that $\flow^{*}(e)=0$ for each edge $e$ that is incident to a node in $X$. Therefore, from the definition of $\abs{\flow^{*}}$, it follows that $\abs{\flow^{*}} = 0$.
				\item {\bf Fact 9:} $\flow^{*}$ is a subflow of $\flow$.\\
				To see why, note that, by our definition of $\delta$, we know that $\flow(e) \geq \delta$ for each edge $e$ in $H$. So, from our definition of $\flow^{*}$, it follows that $\flow^{*}(e) = \delta \leq \flow(e)$ for each edge $e$ in $H$, and, $\flow^{*}(e) = 0 \leq \flow(e)$ for all other edges in $\graph$. Thus, the fact follows from the definition of subflow.
			\end{itemize}
			
			Next, we create another integral flow $\flow'$ by subtracting $\flow^{*}$ from $\flow$, i.e., define $\flow' = \flow \flowsub \flow^{*}$. We provide some facts about $\flow'$ for later use:
			\begin{itemize}
				\item {\bf Fact 10:} $\flow'$ is a flow on $\graph$, and $\abs{\flow'} = \abs{\flow} - \abs{\flow^{*}} = \abs{\flow}$.\\
				To see why, note that $\flow^{*}$ is a subflow of $\flow$ by Fact 9. Then, the fact follows by applying Lemma \ref{lem:flow_sum_sub} and Fact 8.
				\item {\bf Fact 11:} $\flow'$ is a subflow of $\flow$.\\
				To see why, note that, for each edge $e$ in $H$, we have $\flow'(e) = \flow(e)-\flow^{*}(e) = \flow(e)-\delta < \flow(e)$, and, for all other edges in $\graph$, we have $\flow'(e) = \flow(e)-\flow^{*}(e) = \flow(e)-0 = \flow(e)$.
				\item {\bf Fact 12:} $\flow'$ has strictly fewer positive edges than $\flow$.\\
				To see why, note that, by our definition of $\delta$, there must be at least one edge $e_\delta$ in $H$ such that $\flow(e_\delta) = \delta$, and from the definition of $\flow'$, we see that $\flow'(e_\delta) = \flow(e_\delta) - \delta = 0$. This proves that there is at least one edge that is a positive edge in $\flow$ but is a zero edge in $\flow'$. Further, since we showed above that $\flow'$ is a subflow of $\flow$, we know that $\flow'(e) \leq \flow(e)$ for each edge $e$ in $G$, so it follows that each zero edge in $\flow$ is also a zero edge in $\flow'$. Altogether, this proves that $\flow'$ has more zero edges than $\flow$, which implies that $\flow'$ has strictly fewer positive edges than $\flow$.
			\end{itemize}
			
			By Fact 12, we can apply the induction hypothesis to $\flow'$. By doing so, we obtain $\abs{\flow'}$ unit path flows on $\graph$, denoted by $\flow'_1,\ldots,\flow'_{\abs{\flow'}}$, such that: $\flow'_i$ is a subflow of $\flow'$ for each $i \in \{1,\ldots,\abs{\flow'}\}$, and, $\flowsummation_{i=1}^{\abs{\flow'}} \flow'_{i}$ is a subflow of $\flow'$.  We claim (and prove) that $\flow'_1,\ldots,\flow'_{\abs{\flow'}}$ is a collection of $\abs{\flow}$ unit path flows on $G$ that satisfy the conditions of the lemma.
			
			First, we confirm that there are $\abs{\flow}$ elements in this collection. The number of elements is $\abs{\flow'}$, which, by Fact 10, is equal to $\abs{\flow}$, as required.
			
			Next, we confirm that each element in this collection is a subflow of $\flow$. Consider an arbitrary $\flow'_i$ for some $i \in \{1,\ldots,\abs{\flow'}\}$. By the induction hypothesis, $\flow'_i$ is a subflow of $\flow'$, and, by Fact 11, we know that $\flow'$ is a subflow of $\flow$, which implies that $\flow'_i$ is a subflow of $\flow$, as required.
			
			Finally, we confirm that, by taking the sum of all flows in the collection, the resulting flow is a subflow of $\flow$. In particular, our goal is to show that $\left(\flowsummation_{i=1}^{\abs{\flow'}} \flow'_{i}\right)$ is a subflow of $\flow$. By the induction hypothesis, we know that $\left(\flowsummation_{i=1}^{\abs{\flow'}} \flow'_{i}\right)$ is a subflow of $\flow'$, and, by Fact 11, we know that $\flow'$ is a subflow of $\flow$, as required. 
		\end{itemize}
	\end{proof}
	
	\flowconserve*
	\begin{proof}
		For any fixed $\subconfig \in \bsubs$, recall from the construction of $\graph$ that the directed edges involving a node $\internode_{\subconfig}$ are exactly: $\innode_\subconfig\internode_\subconfig$, $\internode_\subconfig\outnode_\subconfig$, and $\internode_\subconfig\internode_{\subconfigtwo}$ for each $\subconfigtwo \in \bsubs$ adjacent to $\subconfig$. Therefore, 
		\begin{align*}
			\hat{\flow}(\inedges(\internode_{\subconfig})) & = \hat{\flow}(\innode_\subconfig\internode_\subconfig) + \sum_{\subconfigtwo \textrm{ adj. to } \subconfig}\hat{\flow}(\internode_\subconfigtwo\internode_\subconfig)\\
			& = \hat{\flow}(\outedges(\innode_\subconfig)) + \sum_{\subconfigtwo \textrm{ adj. to } \subconfig}\hat{\flow}(\internode_\subconfigtwo\internode_\subconfig)
		\end{align*}
		and 
		\begin{align*}
			\hat{\flow}(\outedges(\internode_{\subconfig})) &= \hat{\flow}(\internode_\subconfig\outnode_\subconfig) + \sum_{\subconfigtwo \textrm{ adj. to } \subconfig}\hat{\flow}(\internode_\subconfig\internode_\subconfigtwo)\\
			&= \hat{\flow}(\inedges(\outnode_\subconfig)) + \sum_{\subconfigtwo \textrm{ adj. to } \subconfig}\hat{\flow}(\internode_\subconfig\internode_\subconfigtwo).
		\end{align*}
		
		Summing over all $\subconfig \in \bsubs$, we get
		\begin{align*}
			\sum_{\subconfig \in \bsubs}\hat{\flow}(\inedges(\internode_{\subconfig}))& = \left[\sum_{\subconfig \in \bsubs}\hat{\flow}(\outedges(\innode_\subconfig))\right] + \left[\sum_{\subconfig \in \bsubs}\sum_{\subconfigtwo \textrm{ adj. to } \subconfig}\hat{\flow}(\internode_\subconfigtwo\internode_\subconfig)\right]
		\end{align*}
		and 
		\begin{align*}
			\sum_{\subconfig \in \bsubs}\hat{\flow}(\outedges(\internode_{\subconfig})) & = \left[\sum_{\subconfig \in \bsubs}\hat{\flow}(\inedges(\outnode_\subconfig))\right] + \left[\sum_{\subconfig \in \bsubs}\sum_{\subconfigtwo \textrm{ adj. to } \subconfig}\hat{\flow}(\internode_\subconfig\internode_\subconfigtwo)\right]
		\end{align*}
		
		But, recall that the edge set of $\graph$ is defined as $\edgeset = \{\innode_\subconfig \internode_\subconfig: \subconfig \in \bsubs \} 
		\cup \{\internode_{\subconfig} \internode_{\subconfigtwo}: \subconfig, \subconfigtwo \in \bsubs \text{, $\subconfig$ and $\subconfigtwo$ are adjacent} \}
		\cup \{\internode_\subconfig \outnode_\subconfig: \subconfig \in \bsubs \}$, so the edge set of $\graph$ induced by the nodes of $\internodes$, which we denote by $E_{\internodes}$, is exactly $\{\internode_{\subconfig} \internode_{\subconfigtwo}: \subconfig, \subconfigtwo \in \bsubs \text{, $\subconfig$ and $\subconfigtwo$ are adjacent} \}$. Each directed edge of $E_{\internodes}$ is included exactly once in the double summation $\sum_{\subconfig \in \bsubs}\sum_{\subconfigtwo \textrm{ adj. to } \subconfig}\hat{\flow}(\internode_\subconfigtwo\internode_\subconfig)$, and is also included exactly once in the double summation $\sum_{\subconfig \in \bsubs}\sum_{\subconfigtwo \textrm{ adj. to } \subconfig}\hat{\flow}(\internode_\subconfig\internode_\subconfigtwo)$. This allows us to conclude that
		
		\begin{align*}
			\sum_{\subconfig \in \bsubs}\hat{\flow}(\inedges(\internode_{\subconfig}))
			& = \left[\sum_{\subconfig \in \bsubs}\hat{\flow}(\outedges(\innode_\subconfig))\right] + \left[\sum_{e \in E_{\internodes}}\hat{\flow}(e)\right]
		\end{align*}
		and 
		\begin{align*}
			\sum_{\subconfig \in \bsubs}\hat{\flow}(\outedges(\internode_{\subconfig})) 
			& = \left[\sum_{\subconfig \in \bsubs}\hat{\flow}(\inedges(\outnode_\subconfig))\right] + \left[\sum_{e \in E_{\internodes}}\hat{\flow}(e)\right].
		\end{align*}

		From the conservation constraint, we know that $\hat{\flow}(\inedges(\internode_{\subconfig})) - \hat{\flow}(\outedges(\internode_{\subconfig})) = 0$ for each $\subconfig \in \bsubs$. It follows that $\sum_{\subconfig \in \bsubs}\left[ \hat{\flow}(\inedges(\internode_{\subconfig})) - \hat{\flow}(\outedges(\internode_{\subconfig})) \right] = 0$. 
		
		Thus, we get
		\begin{align*}
			0 & = \sum_{\subconfig \in \bsubs}\left[ \hat{\flow}(\inedges(\internode_{\subconfig})) - \hat{\flow}(\outedges(\internode_{\subconfig})) \right]\\
			& = \left[\sum_{\subconfig \in \bsubs}\hat{\flow}(\inedges(\internode_{\subconfig}))\right] - \left[\sum_{\subconfig \in \bsubs}\hat{\flow}(\outedges(\internode_{\subconfig}))\right]\\
			& = \left[\left[\sum_{\subconfig \in \bsubs}\hat{\flow}(\outedges(\innode_\subconfig))\right] + \left[\sum_{e \in E_{\internodes}}\hat{\flow}(e)\right]\right] - \left[\left[\sum_{\subconfig \in \bsubs}\hat{\flow}(\inedges(\outnode_\subconfig))\right] + \left[\sum_{e \in E_{\internodes}}\hat{\flow}(e)\right]\right]\\
			& = \left[\sum_{\subconfig \in \bsubs}\hat{\flow}(\outedges(\innode_\subconfig)) \right] - \left[\sum_{\subconfig \in \bsubs}\hat{\flow}(\inedges(\outnode_\subconfig))\right]
		\end{align*}
		which implies the desired result.
	\end{proof}
	
\end{document}